\documentclass[11pt,a4paper,12pt,ruled,linesnumbered,vlined]{article}
\usepackage[T1]{fontenc}
\usepackage[latin9]{inputenc}
\usepackage{color}
\usepackage{array}
\usepackage{textcomp}
\usepackage{mathrsfs}
\usepackage{multirow}
\usepackage{algorithm2e}
\usepackage{amsmath}
\usepackage{amsthm}
\usepackage{amssymb}
\usepackage{stmaryrd}
\usepackage{graphicx}
\usepackage{setspace}
\usepackage{esint}
\usepackage[authoryear]{natbib}
\usepackage[unicode=true,pdfusetitle,
 bookmarks=true,bookmarksnumbered=false,bookmarksopen=false,
 breaklinks=false,pdfborder={0 0 1},backref=false,colorlinks=true]
 {hyperref}
\hypersetup{
 citecolor=blue}

\makeatletter

\pdfpageheight\paperheight
\pdfpagewidth\paperwidth

\providecommand{\tabularnewline}{\\}

\newcommand{\lyxaddress}[1]{
	\par {\raggedright #1
	\vspace{1.4em}
	\noindent\par}
}
\theoremstyle{definition}
 \newtheorem{example}{\protect\examplename}
\theoremstyle{plain}
\newtheorem{thm}{\protect\theoremname}
\theoremstyle{remark}
\newtheorem{rem}{\protect\remarkname}
\theoremstyle{plain}
\newtheorem{cor}{\protect\corollaryname}
\theoremstyle{plain}
\newtheorem{lem}{\protect\lemmaname}

\usepackage{amsmath}
\usepackage{amsfonts}
\usepackage{algorithm2e}
\usepackage{multirow}
\usepackage{enumerate}
\usepackage{natbib}
\usepackage{fullpage}
\usepackage{subfigure}
\usepackage{comment}
\usepackage{epsfig}\usepackage{authblk}
\usepackage{bbm}

\usepackage{tikz}
\usetikzlibrary{decorations.pathreplacing,calc}

\usepackage{algpseudocode}
\usepackage[font=footnotesize, labelfont=bf]{caption}

\providecommand{\examplename}{Example}
\providecommand{\lemmaname}{Lemma}

\providecommand{\remarkname}{Remark}
\providecommand{\corollaryname}{Corollary}
\providecommand{\theoremname}{Theorem}

\makeatother

\providecommand{\corollaryname}{Corollary}
\providecommand{\examplename}{Example}
\providecommand{\lemmaname}{Lemma}
\providecommand{\remarkname}{Remark}
\providecommand{\theoremname}{Theorem}

\begin{document}
\title{Metropolis-Hastings with Averaged Acceptance Ratios}
\author{Christophe Andrieu$^{*}$, Sinan Y\i ld\i r\i m$^{+}$, Arnaud Doucet$^{\dagger}$,
and Nicolas Chopin$^{\blacklozenge}$}
\maketitle

\lyxaddress{$^{*}$School of Mathematics, University of Bristol, U.K. \\
 $^{\dagger}$Department of Statistics, University of Oxford, U.K.\\
 $^{+}$Faculty of Engineering and Natural Sciences, Sabanc\i{} University,
Turkey.\\
 $^{\blacklozenge}$ENSAE, France.}
\begin{abstract}
Markov chain Monte Carlo (MCMC) methods to sample from a probability
distribution $\pi$ defined on a space $(\Theta,\mathcal{T})$ consist
of the simulation of realisations of Markov chains $\{\theta_{n},n\geq1\}$
of invariant distribution $\pi$ and such that the distribution of
$\theta_{i}$ converges to $\pi$ as $i\rightarrow\infty$. In practice
one is typically interested in the computation of expectations of
functions, say $f$, with respect to $\pi$ and it is also required
that averages $M^{-1}\sum_{n=1}^{M}f(\theta_{n})$ converge to the
expectation of interest. The iterative nature of MCMC makes it difficult
to develop generic methods to take advantage of parallel computing
environments when interested in reducing time to convergence. While
numerous approaches have been proposed to reduce the variance of ergodic
averages, including averaging over independent realisations of $\{\theta_{n},n\geq1\}$
simulated on several computers, techniques to reduce the ``burn-in''
of MCMC are scarce. In this paper we explore a simple and generic
approach to improve convergence to equilibrium of existing algorithms
which rely on the Metropolis-Hastings (MH) update, the main building
block of MCMC. The main idea is to use averages of the acceptance
ratio w.r.t. multiple realisations of random variables involved, while
preserving $\pi$ as invariant distribution. The methodology requires
limited change to existing code, is naturally suited to parallel computing
and is shown on our examples to provide substantial performance improvements
both in terms of convergence to equilibrium and variance of ergodic
averages. In some scenarios gains are observed even on a serial machine.
\end{abstract}
{\footnotesize{}Keywords: Doubly intractable distributions; Intractable
likelihood; Markov chain Monte Carlo; Pseudo-marginal Metropolis-Hastings;
Reversible jump Monte Carlo; Sequential Monte Carlo; State-space models;
Particle MCMC.} 

\newpage{}

{\footnotesize{}{\tableofcontents{}}}\newpage{}


\section{Introduction\label{sec: Introduction} }

Suppose\nocite{andrieu2018utility} we wish to sample from a given
probability distribution $\pi$ on some measurable space $(\Theta,\mathcal{T})$.
When it is impossible or too difficult to generate perfect samples
from $\pi$, one practical resource is to use a Markov chain Monte
Carlo (MCMC) algorithm generating an ergodic Markov chain $\{\theta_{n},n\geq0\}$
whose invariant distribution is $\pi$. Among MCMC methods, the Metropolis\textendash Hastings
(MH) algorithm plays a central rôle. The MH update proceeds as follows:
given $\theta_{n}=\theta$ and a Markov transition kernel $q\big(\theta,\cdot\big)$
on $(\Theta,\mathcal{T})$, we propose $\vartheta\sim q(\theta,\cdot)$
and set $\theta_{n+1}=\vartheta$ with probability $\alpha(\theta,\vartheta):=\min\left\{ 1,r(\theta,\vartheta)\right\} $,
where 
\begin{equation}
r(\theta,\vartheta):=\frac{\pi({\rm d}\vartheta)q(\vartheta,{\rm d}\theta)}{\pi({\rm d}\theta)q(\theta,{\rm d}\vartheta)}\label{eq:genericMHacceptratio}
\end{equation}
for $(\theta,\vartheta)\in\mathsf{S}\subset\mathsf{\Theta}^{2}$ (see
\citet{Tierney_1998} for a definition of $\mathsf{S}$) is a well
defined Radon\textendash Nikodym derivative, and $r(\theta,\vartheta)=0$
otherwise. When the proposed value is rejected, we set $\theta_{n+1}=\theta$.
We will refer to $r(\theta,\vartheta)$ as the acceptance ratio. The
transition kernel of the Markov chain $\{\theta_{n},n\geq0\}$ generated
with the MH algorithm with proposal kernel $q(\cdot,\cdot)$ is 
\begin{equation}
P(\theta,A)=\int_{A}\alpha(\theta,\vartheta)q(\theta,{\rm d}\vartheta)+\rho(\theta)\mathbb{I}\{\theta\in A\},\quad(\theta,A)\in\Theta\times\mathcal{T},\label{eq: MH transition kernel}
\end{equation}
where $\rho(\theta)$ is the rejection probability such that $P(\theta,\Theta)=1$
and $\mathbb{I}\{\cdot\in A\}$ is the indicator function for set
$A$. Expectations of functions, say $f$, with respect to $\pi$
can be estimated with $S_{M}:=M^{-1}\sum_{n=1}^{M}f(\theta_{n})$
for $M\in\mathbb{N}$, which is consistent under mild assumptions.

Being able to evaluate the acceptance ratio $r(\theta,\vartheta)$
is therefore central to implementing the MH algorithm in practice.
Recently, there has been much interest in expanding the scope of the
MH algorithm to situations where this acceptance ratio is intractable,
that is, impossible or very expensive to compute. A canonical example
of intractability is when $\pi$ can be written as the marginal of
a given joint probability distribution for $\theta$ and some latent
variable $z$. A classical way of addressing this problem consists
of running an MCMC algorithm targeting the joint distribution, which
may however become very inefficient in situations where the size of
the latent variable is high\textendash this is for example the case
for general state-space models. In what follows, we will briefly review
generic ways of tackling this problem. To that purpose we will use
the following simple running example to illustrate various methods.
This example has the advantage that its setup is relatively simple
and of clear practical relevance. We postpone developments for much
more complicated scenarios to Sections \ref{sec: Pseudo-marginal ratio algorithms using averaged acceptance ratio estimators},
\ref{sec: MHAAR via Rao-Blackewellisation of PMR}, and \ref{sec: State-space models: SMC and conditional SMC within MHAAR}. 
\begin{example}[\textbf{Inference with doubly intractable models}]
\label{ex:doublyintractable} In this scenario the likelihood function
of the unknown parameter $\theta\in\Theta$ for the dataset $y\in\mathsf{Y}$,
$\ell_{\theta}(y)$, is only known up to a normalising constant, that
is $\ell_{\theta}(y)=g_{\theta}(y)/C_{\theta},$ where $C_{\theta}$
is unknown, while $g_{\theta}(y)$ can be evaluated pointwise for
any value of $\theta\in\Theta$. In a Bayesian framework, for a prior
density $\eta(\theta)$, we are interested in the posterior density
$\pi(\theta)$, given by $\pi(\theta)\propto\eta(\theta)\ell_{\theta}(y).$
The acceptance ratio of the MH algorithm associated to a proposal
density $q(\theta,\vartheta)$ is 
\begin{align}
r(\theta,\vartheta)=\frac{q(\vartheta,\theta)}{q(\theta,\vartheta)}\frac{\eta(\vartheta)}{\eta(\theta)}\frac{g_{\vartheta}(y)}{g_{\theta}(y)}\frac{C_{\theta}}{C_{\vartheta}},\label{eq: MCMC acceptance probability with intractable likelihood-1}
\end{align}
which cannot be calculated because of the unknown ratio $C_{\theta}/C_{\vartheta}$.
While the likelihood function may be intractable, sampling artificial
datasets $u\sim\ell_{\theta}(y_{*}){\rm d}y_{*}$ may be possible
for any $\theta\in\Theta$, and sometimes computationally cheap. We
will describe two known approaches which exploit and expand this property
in order to design Markov kernels preserving $\pi(\theta)$ as invariant
density. 
\end{example}

\subsection{Estimating the target density\label{subsec: Estimating the target density}}

Assume for simplicity of exposition that $\pi$ has a probability
density with respect to some $\sigma$-finite measure. We will abuse
notation slightly by using $\pi$ for both the probability distribution
and its density. A simple method to tackle intractability which has
recently attracted interest consists of replacing the value of $\pi(\theta)$
with a non-negative random estimator $\hat{\pi}(\theta)$ whenever
it is required in the implementation of the MH algorithm above. If
there exists a constant $C>0$ such that $\mathbb{E}[\hat{\pi}(\theta)]=C\pi(\theta)$
for all $\theta\in\Theta$, a property we refer to abusively as unbiasedness,
this strategy turns out to lead to exact algorithms, that is sampling
from $\pi$ is guaranteed at equilibrium under very mild assumptions
on $\hat{\pi}(\theta)$. This approach leads to so called pseudo-marginal
algorithms \citep{Beaumont_2003,Andrieu_and_Roberts_2009}. In what
follows, for $a,b\in\mathbb{R}$ we let $\llbracket a,b\rrbracket:=[a,b]\cap\mathbb{Z}$
and use the specialised notation $\llbracket a\rrbracket:=\llbracket1,a\rrbracket$.
\begin{example}[\textbf{Example 1, ctd}]
\label{ex: pseudo-marginal for doubly intractable models} Let $h:\mathsf{Y}\rightarrow[0,\infty)$
be an integrable non-negative function of integral equal to $1$.
For a given $\theta$, an unbiased estimate of $\pi(\theta)$ can
be obtained via importance sampling whenever the support of $g_{\theta}$
includes that of $h$: 
\begin{equation}
\hat{\pi}^{N}(\theta)\propto\eta(\theta)g_{\theta}(y)\left\{ \frac{1}{N}\sum_{i=1}^{N}\frac{h(u^{(i)})}{g_{\theta}(u^{(i)})}\right\} ,\quad u^{(i)}\overset{{\rm iid}}{\sim}\ell_{\theta}(y_{*}){\rm d}y_{*},\quad i\in\llbracket N\rrbracket,
\end{equation}
since the normalised sum is an unbiased estimator of $1/C_{\theta}$.
The auxiliary variable method of \citet{Muller_et_al_2006} corresponds
to $N=1$. An interesting feature of this approach is that $N$ is
a free parameter of the algorithm which reduces the variability of
this estimator. It is shown in \citet{Andrieu_and_Vihola_2014} that
increasing $N$ always reduces the asymptotic variance of averages
using this chain and will in most cases of interest improve convergence
to equilibrium. This is particularly interesting in a parallel computing
environment but, as we shall see, can prove of interest on serial
machines.
\end{example}

\subsection{Estimating the acceptance ratio \label{subsec: Estimating the acceptance ratio}}

One can in fact push the idea of replacing algebraic expressions with
estimators further. Instead of approximating the numerator and denominator
of the acceptance ratio $r(\theta,\vartheta)$ independently, it is
indeed possible to use directly estimators of the acceptance ratio
$r(\theta,\vartheta)$ and still obtain algorithms guaranteed to sample
from $\pi$ at equilibrium. An interesting feature of these algorithms
is that we estimate the ratio $r(\theta,\vartheta)$ afresh whenever
it is required. On the contrary, in algorithms using unbiased estimates
of the target density, the estimate $\hat{\pi}(\theta)/C$ is used
in the acceptance ratio until a transition is accepted. As a consequence
whenever $\hat{\pi}(\theta)/C$ significantly overestimates $\pi(\theta)$
the algorithm spends a long period of time stuck in a particular state,
resulting in poor performance. In the following continuation of Example
\ref{ex:doublyintractable}, we present a particular case of estimating
the acceptance ratio, proposed by \citet{Murray_et_al_2006}.
\begin{example}[\textbf{Example 1, ctd}]
\label{ex: exchange algorithm} The exchange algorithm of \citet{Murray_et_al_2006}
is motivated by the realisation that while for $u\sim\ell_{\vartheta}(y_{\ast}){\rm d}y_{\ast}$
and $h(u)/g_{\vartheta}(u)$ is an unbiased estimator of $1/C_{\vartheta}$,
the particular choice $h(u)=g_{\theta}(u)$ leads to an unbiased estimator
$g_{\theta}(u)/g_{\vartheta}(u)$ of $C_{\theta}/C_{\vartheta}$ required
in \eqref{eq: MCMC acceptance probability with intractable likelihood-1}.
This suggests the following MH type update. Given $\theta\in\Theta$,
sample $\vartheta\sim q(\theta,\cdot)$, then $u\sim\ell_{\vartheta}(y_{*}){\rm d}y_{*}$
and use the acceptance ratio 
\begin{equation}
r_{u}(\theta,\vartheta)=\frac{q(\vartheta,\theta)}{q(\theta,\vartheta)}\frac{\eta(\vartheta)}{\eta(\theta)}\frac{g_{\vartheta}(y)}{g_{\theta}(y)}\frac{g_{\theta}(u)}{g_{\vartheta}(u)},\label{eq: Murray's acceptance ratio-1}
\end{equation}
which is an unbiased estimator of the acceptance ratio in \eqref{eq: MCMC acceptance probability with intractable likelihood-1}.
Remarkably this algorithm admits $\pi$ as an invariant distribution
and hence, under additional mild assumptions, is guaranteed to produce
samples asymptotically distributed according to $\pi$. 
\end{example}

\subsection{Contribution\label{subsec: Contribution}}

As we shall see numerous MH algorithms of interest to sample from
$\pi$ have a tractable acceptance ratio of the form $r_{u}(\theta,\vartheta)$
where $u$ is sampled afresh at each iteration, as is the case in
Example \ref{ex: exchange algorithm}. Such sampling induces variability
of the acceptance ratio which, as we shall see, is undesirable and
a natural question is whether this can be alleviated by averaging
multiple realisations of some of the variables involved. More specifically,
given multiple realisations $r_{u^{(i)}}(\theta,\vartheta)$, $i\in\llbracket N\rrbracket$
is it possible to design an algorithm leaving $\pi$ invariant and
of superior performance? While the naïve approach consisting of using
$N^{-1}\sum_{i=1}^{N}r_{u^{(i)}}(\theta,\vartheta)$ in place of $r_{u}(\theta,\vartheta)$
is not valid, in that $\pi$ is not guaranteed to be an invariant
distribution anymore, we show that a solution alternating between
the use of this average and its inverse leads to a correct algorithm.
These algorithms naturally lend themselves to parallel computations
as independent ratio estimators can be computed in parallel at each
iteration \citet{lee2010utility,suchard2010understanding}. Provided
access to a parallel machine is available and the cost of computing
$r_{u}(\theta,\vartheta)$ dominates communication cost, which is
the case in challenging applications, we show that this approach can
reduce the burn-in-period, sometimes substantially\textendash in fact
the higher the variability of $r_{u}(\theta,\vartheta)$ for $\theta,\vartheta\in\Theta$
the more substantial the gains are. As a by-product the induced rapid
mixing also leads to reduced asymptotic variance of ergodic averages,
even when implemented on a serial machine in some scenarios. Generic
methods to reduce burn-in and utilise parallel architectures are scarce
\citep{sohn1995parallel}, in contrast with variance reduction techniques
for which better embarrassingly parallel solutions \citep{doi:10.1093/biomet/asx031,bornn2017use}
and/or post-processing methods are available \citep{delmas2009does,dellaportas2012control}.
An interesting practical point is that the approach we advocate requires
only limited adaptation of the specific, and often intricate, code
for an existing algorithm beyond the generic management of the parallel
environment. Note however that the actual implementation of our algorithms
on a parallel computer is beyond of the present manuscript which focuses
primarily on developping sound methodology and provide initial evaluation
of expected performance.

In Sections \ref{sec: Pseudo-marginal ratio algorithms using averaged acceptance ratio estimators}
we introduce the MHAAR methodology in full generality, providing some
theoretical analysis supporting their correctness and claimed efficiency
while we illustrate its interest in the context of reversible jump
MCMC algorithms. In Section \ref{sec: MHAAR via Rao-Blackewellisation of PMR}
we specialise MHAAR to latent variable models and present an alternative
to pseudo-marginal algorithms \citet{Beaumont_2003,Andrieu_and_Roberts_2009}
which is shown to have far superior performance properties, even on
a serial machine. In Section \ref{sec: State-space models: SMC and conditional SMC within MHAAR},
we show how MHAAR can be advantageous in the context of inference
in state-space models when it is utilised in combination with sequential
Monte Carlo (SMC) algorithms. In particular, we expand the scope of
particle MCMC algorithms \citep{Andrieu_et_al_2010} and show novel
ways of using multiple or all possible paths obtainable from a conditional
SMC (cSMC) run to estimate the marginal acceptance ratio.   of MHAAR.
We again assess gain performance numerically, demonstrating the interest
of the approach.  The proofs of the validity of our algorithms as
well as additional discussion on the generalisation of the methods
can found in the Appendices.

\section{Using averaged acceptance ratio estimators\label{sec: Pseudo-marginal ratio algorithms using averaged acceptance ratio estimators}}

\subsection{A general perspective on MH based algorithms\label{subsec:A-general-perspective}}

Before describing our novel algorithms we briefly outline a framework,
fully developed in \citet{andrieu:lee:livingstone:2020}, which allows
for a systematic and concise presentation of complex MH updates. In
particular the presentation adopted makes validating, that is establishing
reversibility with respect to the distribution of interest, fairly
direct and is helpful to establish the expression for the acceptance
ratio involved in the update. 

The key idea here is that in order to describe and validate a MH update
it is sufficient to identify all the random variables $\xi$ involved
in the update before the accept/reject step, their distribution, the
mapping $\varphi$ used to determine the next state of the Markov
chain from $\varphi(\xi)$ and check that it satisfies $\varphi\circ\varphi={\rm Id}$.
Consider for example the standard update given at the beginning of
Section \ref{sec: Introduction}: here the variables involved are
$\xi:=(\theta,\vartheta)\in\Theta^{2}$, their distribution before
the accept/reject step is $\mathring{\pi}({\rm d}\xi)=\pi({\rm d}\theta)q(\theta,{\rm d}\vartheta)$,
and the involution used to determine the next state is $\varphi(\theta,\vartheta):=(\vartheta,\theta)$
for $\theta,\vartheta\in\Theta^{2}$, leading to the familiar acceptance
ratio \eqref{eq:genericMHacceptratio}. The popular random walk Metropolis
algorithm corresponds to the choices $\xi=(\theta,\zeta)\in\Theta^{2}$,
where $\zeta\in\Theta$ is the increment used to perturb $\theta$,
$\mathring{\pi}({\rm d}\xi)=\pi({\rm d}\theta)q({\rm d}\zeta)$ and
$\varphi(\theta,\zeta)=(\theta+\zeta,-\zeta)$. In the situation where
$\Theta=\mathbb{R}^{d}$, $\pi$ and $q$ admit densities with respect
to the Lebesgue measure (also denoted $\pi$ and $q$ and assumed
to be strictly positive for simplicity) and $q$ is symmetric, the
resulting acceptance ratio is of form
\[
\frac{\pi(\theta+\zeta)q(-\zeta)}{\pi(\theta)q(\zeta)}=\frac{\pi(\theta+\zeta)}{\pi(\theta)},
\]
where the numerator is the density resulting from the change of variable
$\varphi(\theta,\zeta)=\varphi^{-1}(\theta,\zeta)=(\theta+\zeta,-\zeta)$,
of Jacobian $1$. This can be generalised as follows. Let $\mathring{\pi}$
be a probability distribution on some measurable space $(\mathsf{X},\mathcal{X})$
and let $\varphi:\mathsf{\mathsf{X}\rightarrow\mathsf{\mathsf{X}}}$
be a measurable mapping, we define the push forward distribution $\mathring{\pi}^{\varphi}$
to be the probability distribution of $\varphi(\xi)$ when $\xi\sim\mathring{\pi}$,
that is such that for any measurable $A\in\mathcal{X}$, $\mathring{\pi}^{\varphi}(A):=\mathring{\pi}\big(\varphi^{-1}(A)\big)$.
Assume further that $\mathring{\pi}$ has marginal $\pi$, say $\mathring{\pi}({\rm d}\xi)=\pi({\rm d}\xi_{0})\mathring{\pi}({\rm d}\xi_{1}\mid\xi_{0})$,
and that $\varphi\colon\mathsf{X}\rightarrow\mathsf{X}$ is an involution.
Then the following update is a valid MH update, that is ignoring the
second components $\xi_{1}$ and $\xi'_{1}$ it is reversible with
respect to $\pi$ and hence leaves this distribution invariant:
\begin{enumerate}
\item given $\xi_{0}$ sample $\xi_{1}\sim\mathring{\pi}(\cdot\mid\xi_{0})$, 
\item compute
\begin{equation}
\alpha(\xi):=\min\left\{ 1,\mathring{r}(\xi)\right\} \text{ with }\mathring{r}(\xi):=\frac{\mathring{\pi}^{\varphi}({\rm d}\xi)}{\mathring{\pi}({\rm d}\xi)},\label{eq:generic-accept-ratio}
\end{equation}
\item with probability $\alpha(\xi)$ return $\xi'=\varphi(\xi)$, otherwise
return $\xi'=\xi$.
\end{enumerate}
The quantity $r(\xi)$ is a so-called Radon-Nikodym derivative, guaranteed
to exist under very mild assumptions. In this manuscript $\mathring{\pi}$
will always be assumed to have a known density with respect to a product
of counting and Lebesgue measures and $r(\xi)$ will be either zero
whenever either densities of $\xi'=\varphi(\xi)$ or $\xi$ is zero,
or the ratio of these densities otherwise. The notation above allows
us, for the moment, to avoid the distinction between discrete and
real valued variables and the possible presence of a Jacobian. Naturally
another practical requirement is that sampling from the ``proposal
distribution'' $\mathring{\pi}(\cdot\mid\xi_{0})$ should be computationally
tractable. 

To summarize, in what follows we adopt the following systematic presentation
of MH updates: 
\begin{enumerate}
\item identify all the instrumental variables $\xi_{1}$ and the distribution
$\mathring{\pi}({\rm d}\xi)$ involved in the parameter update, 
\item identify the involution $\varphi\colon\mathsf{X}\rightarrow\mathsf{X}$,
\item find an expression for $\mathring{r}(\xi)$.
\end{enumerate}
Note that the above does not ensure convergence to equilibrium of
the Markov chain, which is problem dependent. 

The following property can be established and will be used on several
occasions in the remainder of the manuscript, for $\xi\in\mathsf{X}$
\begin{align}
\mathring{r}\circ\varphi(\xi) & =\begin{cases}
1/\mathring{r}(\xi) & \text{if }\mathring{r}(\xi)>0\\
0 & \text{otherwise}
\end{cases}.\label{eq:prop-r-circ-phi-inverse-r}
\end{align}
In order to simplify presentation we will always assume that $\mathring{r}(\xi)>0$
for any $\xi\in\mathsf{X}$\textendash the general scenario is a straightforward
adaptation.

\subsection{Motivation: an idealised algorithm\label{subsec: Pseudo-marginal ratio algorithms}}

Consider the generic algorithm given in the previous subsection. Our
primary aim here is to show that it is possible to improve performance
of this algorithm by using a modification where the acceptance ratio
$\mathring{r}(\xi)$ in \eqref{eq:generic-accept-ratio} is integrated
with respect to a subset of the proposed variables $\xi_{1}$. In
the case of Example~\ref{ex:doublyintractable}-\ref{ex: exchange algorithm},
we have $\xi_{1}=(\vartheta,u)\in\mathsf{\Theta}\times\mathsf{Y}$
and marginalisation with respect to $u$, that is the simulated artificial
datasets, is sought. The motivation for this is that removing dependence
of $\mathring{r}(\xi)$ on $u$ removes variability and will result
in a better expected acceptance rate and, in the spirit of \citep{Andrieu_and_Vihola_2014}
lead to algorithms of improved performance. The algorithm is not implementable
in general but captures in a simple setup the main idea we develop
further in this paper. Indeed MHAAR algorithms are exact numerical
approximations of this idealised algorithm, in that they preserve
the desired distribution invariant, and the latter algorithm can be
thought of as a `lower bound' on what the approximations can achieve
in terms of performance.

Motivated by applications, we consider the scenario where the target
distribution of interest $\pi({\rm d}\theta)$ is not tractable, but
arises from a tractable latent variable model $\pi\big({\rm d}(\theta,z)\big)$
defined on some space $(\Theta\times\mathsf{Z},\mathcal{T}\otimes\mathscr{Z})$.
As a result the target distribution of interest is now $\pi\big({\rm d}(\theta,z)\big)$
and Example \ref{ex:doublyintractable} can be recovered by simply
ignoring $z$. We first describe a standard instance of the MH update
to sample from this target. Let $\left(\mathsf{U},\mathcal{U}\right)$
be some probability space, and $\phi_{\theta,\vartheta}:\mathsf{Z}\times\mathsf{U}\mapsto\mathsf{Z}\times\mathsf{U}$
for all $\theta,\vartheta\in\Theta^{2}$ be invertible mappings such
that $\phi_{\theta,\vartheta}=\phi_{\vartheta,\theta}^{-1}$. Using
the framework of the previous section we consider the set of variables
$\xi:=(\theta,\vartheta,z,u)\in\Theta^{2}\times\mathsf{Z}\times\mathsf{U}$,
an involution of the type 
\begin{equation}
\varphi(\theta,\vartheta,z,u):=\big(\vartheta,\theta,\phi_{\theta,\vartheta}(z,u)\big),\label{eq:involution for the core joint distribution}
\end{equation}
and the probability distribution
\begin{equation}
\mathring{\pi}(\mathrm{d}(\theta,\vartheta,z,u)):=\pi\big(\mathrm{d}(\theta,z)\big)q\big(\theta,{\rm d}\vartheta\big)Q_{\theta,\vartheta,z}({\rm d}u),\label{eq:core joint distribution for MHAAR}
\end{equation}
for a family of probability distributions $Q_{\theta,\vartheta,z}({\rm d}u)$
defined on the probability space $\left(\mathsf{U},\mathcal{U}\right)$
and $q(\theta,\cdot)$ as in Section \ref{sec: Introduction}. Here
the nature of $u$ is problem dependent, guided by the choice of involution
$\varphi$ and tractability of acceptance ratios of the type \eqref{eq:generic-accept-ratio}.
This generality allows us to cover scenarios where the latent variable
$z$ is updated thanks to a mapping from $z,u$ to $z',u'$ by $\phi_{\theta,\vartheta}(\cdot)$.
For example, again ignoring the latent variable $z$ from the notation
and letting $u'=u\in\mathsf{U}=\mathsf{Y}$ corresponds to the exchange
algorithm of Example~\ref{ex:doublyintractable}-\ref{ex: exchange algorithm}.
 We now introduce an improved MH update which uses the integrated
acceptance ratio
\begin{equation}
\mathring{r}(\theta,\vartheta,z):=\int\mathring{r}(\theta,\vartheta,z,u)Q_{\theta,\vartheta,z}(\mathrm{d}u),\label{eq:integrated-accept-ratio}
\end{equation}
assumed to be tractable for the moment. We note that in Example \ref{ex:doublyintractable}-\ref{ex: exchange algorithm}
the choice $Q_{\theta,\vartheta,z}(\mathrm{d}u)=\ell_{\vartheta}(u){\rm d}u$,
where we keep $z$ for notational compatibility but recall that $z$
is not needed for this example, the integrated acceptance ratio simplifies
to \eqref{eq: MCMC acceptance probability with intractable likelihood-1}.
A solution around intractability is the topic of the next section.
The following update can be shown to be $\mathring{\pi}-$reversible. 
\begin{enumerate}
\item sample $\vartheta\sim q(\theta,\cdot)$ and $c\sim{\rm Unif}\{1,2\}$ 
\item sample 
\begin{align}
u\sim\begin{cases}
Q_{\theta,\vartheta,z}({\rm d}u)\mathring{r}(\theta,\vartheta,z,u)/\mathring{r}(\theta,\vartheta,z) & \text{if }c=1\\
Q_{\theta,\vartheta,z}({\rm d}u) & \text{if }c=2
\end{cases}\label{eq: intractable proposal for u}
\end{align}
and form $\xi:=(\theta,\vartheta,z,u,c)$,
\item with $\varphi$ as above, compute $\xi'=\varphi^{\ast}(\xi):=\big(\varphi(\theta,\vartheta,z,u),3-c\big)=:\big(\vartheta,\theta,z',u',3-c\big)$,
\item return $\xi'$ with probability $\alpha(\theta,\vartheta,z,u,c)=\min\{1,\mathring{r}(\theta,\vartheta,z,u,c)\}$
where
\begin{align*}
\mathring{r}(\theta,\vartheta,z,u,c)=\begin{cases}
\mathring{r}(\theta,\vartheta,z) & \text{if }c=1\\
1/\mathring{r}(\vartheta,\theta,z') & \text{if }c=2
\end{cases},
\end{align*}
otherwise return $\xi=(\theta,\vartheta,z,u,c).$
\end{enumerate}
The essential idea here is that alternating between the use of two
appropriately chosen sampling schemes for $u$, the acceptance probability
depends on the integrated acceptance ratio only. What's more in the
case where $c=1$ we see that the proposal distribution for $u$ is
biased towards values leading to high acceptance ratios for the algorithm
defined by \eqref{eq:involution for the core joint distribution}
and \eqref{eq:core joint distribution for MHAAR}, in the spirit of
\citet[Chapter 4]{cainey} and \citet{zanella:2020} where the proposal
distribution is weighted by a function of the target density $\pi$.
We now briefly outline why the acceptance ratio appears to be integrated:
\[
\mathring{\pi}^{\ast}(\mathrm{d}(\theta,\vartheta,z,u,c))=\begin{cases}
\mathring{\pi}(\mathrm{d}(\theta,\vartheta,z,u))\mathring{r}(\theta,\vartheta,z,u)/\mathring{r}(\theta,\vartheta,z)\frac{1}{2} & \text{if }c=1\\
\mathring{\pi}(\mathrm{d}(\theta,\vartheta,z,u))\frac{1}{2} & \text{if }c=2
\end{cases},
\]
using \eqref{eq:generic-accept-ratio} we see that the acceptance
ratio is of the form claimed, as for $c=1$
\[
\frac{(\mathring{\pi}^{\ast})^{\varphi^{\ast}}(\mathrm{d}(\theta,\vartheta,z,u,2))}{\mathring{\pi}^{\ast}(\mathrm{d}(\theta,\vartheta,z,u,1))}=\frac{\mathring{\pi}^{\varphi}(\mathrm{d}(\theta,\vartheta,z,u))}{\mathring{\pi}(\mathrm{d}(\theta,\vartheta,z,u))}\frac{\mathring{r}(\theta,\vartheta,z)}{\mathring{r}(\theta,\vartheta,z,u)}=\mathring{r}(\theta,\vartheta,z),
\]
which does not depend on $u$, and for $c=2$ we use \eqref{eq:prop-r-circ-phi-inverse-r},
yielding $\mathring{r}(\theta,\vartheta,z,u,2)=1/\mathring{r}\circ\varphi^{\ast}(\theta,\vartheta,z,u,1)=1/\mathring{r}(\vartheta,\theta,z')$,
which depends on $u$ through $z'$. The acceptance ratio for $c=2$
may seem disappointing, but it can be shown that reversibility implies
\[
\int\alpha(\theta,\vartheta,z,u,1)\mathring{\pi}^{\ast}(\mathrm{d}(\theta,\vartheta,z,u,1))=\int\alpha(\theta,\vartheta,z,u,2)\mathring{\pi}^{\ast}(\mathrm{d}(\theta,\vartheta,z,u,2)),
\]
that is the expected acceptance probabilities when $c=1$ or $c=2$
are equal. Further application of Jensen's inequality to the concave
function $a\mapsto\min\{1,a\}$ establishes that 
\[
\frac{1}{2}\int\min\{1,\mathring{r}(\theta,\vartheta,z,u)\}\mathring{\pi}(\mathrm{d}(\theta,\vartheta,z,u))\leq\int\min\{1,\mathring{r}(\theta,\vartheta,z)\}\mathring{\pi}^{\ast}(\mathrm{d}(\theta,\vartheta,z,u,1)),
\]
implying that for a given proposal mechanism $q(\theta,{\rm d}\vartheta)$,
the algorithm using the integrated acceptance ratio accepts more proposed
transitions. We will see that this leads to improved performance.

\subsection{MH with Averaged Acceptance Ratio \label{subsec: Pseudo-marginal algorithm with averaged acceptance ratio estimator}}

While valid theoretically, the algorithm of Subsection~\ref{subsec: Pseudo-marginal ratio algorithms}
is rarely implementable in practice since $\mathring{r}(\theta,\vartheta,z)$
is typically intractable and sampling $u$ from \eqref{eq: intractable proposal for u}
when $c=1$ potentially difficult. Instead we develop here a very
closely related update relying on averages of 
\[
r_{u^{(i)}}(\theta,\vartheta,z):=\mathring{r}(\theta,\vartheta,z,u^{(i)})
\]
for, say, $N>1$ realisations $\mathfrak{u}:=u^{(1:N)}=\big(u^{(1)},\ldots,u^{(N)}\big)\in\mathfrak{U}:=\mathsf{U}^{N}$
of $u$, that is the accept/reject mechanism will now rely on
\begin{equation}
r_{\mathfrak{u}}^{N}(\theta,\vartheta,z):=\frac{1}{N}\sum_{i=1}^{N}r_{u^{(i)}}(\theta,\vartheta,z).\label{eq: average acceptance ratio}
\end{equation}
The novel scheme, called MH with Averaged Acceptance Ratio (MHAAR),
relies on the set of variables $\xi:=(\theta,\vartheta,z,\mathfrak{u},k,c)\in\Theta^{2}\times\mathsf{Z}\times\mathfrak{U}\times\llbracket N\rrbracket\times\{1,2\}$
and the joint distribution
\begin{equation}
\mathring{\pi}(\mathrm{d}\xi):=\pi\big(\mathrm{d}(\theta,z)\big)\frac{1}{2}Q_{c}^{N}\big(\theta,z;{\rm d}(\vartheta,\mathfrak{u},k)\big),\label{eq: joint distribution for MHAAR - sinle variable case}
\end{equation}
where the probability distributions $Q_{c}^{N}\big(\theta,z;{\rm d}(\vartheta,\mathfrak{u},k)\big)$,
$c=1,2$, are given by 
\begin{align*}
Q_{1}^{N}\big(\theta,z;{\rm d}(\vartheta,\mathfrak{u},k)\big) & :=q(\theta,{\rm d}\vartheta)\prod_{i=1}^{N}Q_{\theta,\vartheta,z}({\rm d}u^{(i)})\frac{\frac{1}{N}r_{u^{(k)}}(\theta,\vartheta,z)}{r_{\mathfrak{u}}^{N}(\theta,\vartheta,z)},\\
Q_{2}^{N}\big(\theta,z;{\rm d}(\vartheta,\mathfrak{u},k)\big) & :=q(\theta,{\rm d}\vartheta)Q_{\theta,\vartheta,z}({\rm d}u^{(k)})\prod_{i=1,i\neq k}^{N}Q_{\vartheta,\theta,\phi_{\theta,\vartheta}^{[1]}(z,u^{(k)})}(\mathrm{d}u^{(i)})\frac{1}{N},
\end{align*}
where, with $\phi_{\theta,\vartheta}$ as in Section \ref{subsec: Pseudo-marginal ratio algorithms},
we have defined the functions $\phi_{\theta,\vartheta}^{[1]}\colon\mathsf{Z}\times\mathsf{U}\rightarrow\mathsf{Z}$
and $\phi_{\theta,\vartheta}^{[2]}\colon\mathsf{Z}\times\mathsf{U}\rightarrow\mathsf{U}$
such that $\phi_{\theta,\vartheta}:=(\phi_{\theta,\vartheta}^{[1]},\phi_{\theta,\vartheta}^{[2]})$.
Here, $r_{u^{(k)}}(\theta,\vartheta,z)$ is the acceptance ratio corresponding
to the the joint distribution in \eqref{eq:core joint distribution for MHAAR}
along with the involution in \eqref{eq:involution for the core joint distribution}.
As in the previous section, the role of $\phi_{\theta,\vartheta}$
is to parametrise how a new value $z'$ of $z$ is proposed in an
MH update using a variable $u\in\mathsf{U}$ i.e. $(z',u')=\phi_{\theta,\vartheta}(z,u)$.
A simple example corresponds to $\mathsf{U}=\mathsf{Z}$ and the choice
$\phi_{\theta,\vartheta}(z,u)=(u,z)$; a more sophisticated example
will be given in Section \ref{subsec:Example:-improving-transdimensio}.
A MHAAR update consists of the following steps. Given $(\theta,z)\in\Theta\times\mathsf{Z}$,
\begin{enumerate}
\item sample $c\sim\text{Unif(\ensuremath{\left\{  1,2\right\} } )}$, $(\vartheta,\mathfrak{u},k)\sim Q_{c}^{N}(\theta,z;\cdot)$
and form $\xi:=(\theta,\vartheta,z,\mathfrak{u},k,c)$, 
\item compute 
\begin{equation}
\xi'=\varphi(\xi):=\big(\vartheta,\theta,\phi_{\theta,\vartheta}^{[1]}(z,u^{(k)}),u^{(1:k-1)},\phi_{\theta,\vartheta}^{[2]}(z,u^{(k)}),u^{(k+1:N)},k,3-c\big),\label{eq: involution for MHAAR - single variable case}
\end{equation}
\item return $\xi'$ with probability $\min\left\{ 1,\mathring{r}(\xi)\right\} $
where with $(\theta',\vartheta',z',\mathfrak{u}',k',c')=\xi'$ and
$r_{\mathfrak{u}}^{N}(\theta,\vartheta,z)$ given in \eqref{eq: average acceptance ratio},
\begin{equation}
\mathring{r}(\xi)=\begin{cases}
r_{\mathfrak{u}}^{N}(\theta,\vartheta,z), & \text{for }c=1\\
1/r_{\mathfrak{u}'}(\vartheta,\theta,z'), & \text{for }c=2
\end{cases},\label{eq:accept-ratio-MHAAR}
\end{equation}
 otherwise return $\xi$.
\end{enumerate}
It is not difficult to check that the mapping $\varphi$ is an involution,
and Theorem \ref{thm:expression-accept-ration-MHAAR} below establishes
that $\mathring{r}(\xi)$ indeed simplifies to the desired form \eqref{eq:accept-ratio-MHAAR}. 
\begin{thm}
\label{thm:expression-accept-ration-MHAAR}For the probability distribution
$\mathring{\pi}$ and involution $\varphi$ defined in \eqref{eq: joint distribution for MHAAR - sinle variable case}
and \eqref{eq: involution for MHAAR - single variable case} respectively
the acceptance ratio $\mathring{r}(\xi)$ is as in \eqref{eq:accept-ratio-MHAAR}. 
\end{thm}
Details of the proof can be found in Appendix \ref{subsec: Acceptance ratio of MHAAR for latent variable models}
and pseudo-code is given in Algorithm \ref{alg: MHAAR for pseudo-marginal ratio in latent variable models-1}
\textendash we will refer to the corresponding Markov kernel as $\mathring{P}^{N}$
for $N\in\mathbb{N}_{*}$ with the simplification $\mathring{P}$
for $N=1$. For $w_{1},w_{2},\ldots,w_{m}$ such that for $m\in\mathbb{N}$,
$w_{k}\geq0$ for $k=1,\ldots,m$ and $\sum_{k=1}^{m}w_{k}>0$, we
define $K\sim\mathcal{P}(w_{1},\ldots,w_{m})$ to mean that $\mathbb{P}(K=k)\propto w_{k}$.
\begin{rem}
Note that when $c=2$ the simulation of $k$ is in practice not required.
Also, when $c=1$, the acceptance probability does not depend on $k$,
hence sampling of $k$ is necessary only if the move is accepted,
which can be exploited for faster implementation.
\end{rem}
\begin{algorithm}[!h]
\caption{MHAAR for averaging PMR estimators}
\label{alg: MHAAR for pseudo-marginal ratio in latent variable models-1}

\KwIn{Current sample $(\theta,z)$} 

\KwOut{New sample} 

Sample $\vartheta\sim q(\theta,\cdot)$ and $c\sim\text{Unif}(\left\{ 1,2\right\} )$.
\\
 \If{$c=1$}{ \For{$i=1,\ldots,N$}{Sample $u^{(i)}\sim Q_{\theta,\vartheta,z}(\cdot)$}

Sample $k\sim\mathcal{P}\big(r_{u^{(1)}}(\theta,\vartheta,z),\ldots,r_{u^{(N)}}(\theta,\vartheta,z)\big)$,
and set $z'=\phi_{\theta,\vartheta}^{[1]}(z,u^{(k)})$.\\
Return $(\vartheta,z')$ with probability $\min\{1,r_{\mathfrak{u}}^{N}(\theta,\vartheta,z)\})$,
otherwise return $(\theta,z)$ . 

} 

\Else{Sample $k\sim\text{Unif}(\left\llbracket N\right\rrbracket )$
and $u^{(k)}\sim Q_{\theta,\vartheta,z}(\cdot)$, $z'=\phi_{\theta,\vartheta}^{[1]}(z,u^{(k)})$.\\
 \label{line:elsepseudomarginal-1} \For{ $i=1,\ldots,N,i\neq k$,
}{ Sample $u^{(i)}\sim Q_{\vartheta,\theta,z'}(\cdot)$. }

Return $(\vartheta,z')$ with probability $\min\{1,1/r_{\mathfrak{u}'}^{N}(\vartheta,\theta,z')\})$,
otherwise return $(\theta,z)$.

}
\end{algorithm}

\begin{rem}
\label{rem: Choosing Q1 and Q2 with different probabilities}In some
scenarios, for given values $(\theta,\vartheta)\in\Theta^{2}$ it
may be preferable for computational efficiency to sample $(\vartheta,\mathfrak{u},k)\sim Q_{c}^{N}(\theta,z;\cdot)$
for $c=1$ rather than $c=2$, or vice versa. This will be the case
in Example \ref{ex: Poisson multiple change-point model}. This is
possible by changing the distribution of $c$: define a function $\omega:\Theta^{2}\times\{1,2\}\rightarrow[0,1]$
such that $\omega(\theta,\vartheta,1)+\omega(\theta,\vartheta,2)=1$,
therefore defining a probability distribution for $c$, for any $(\theta,\vartheta)\in\Theta^{2}$.
The resulting averaged acceptance ratio is now
\begin{equation}
r_{\mathfrak{u}}^{N}(\theta,\vartheta,z):=\frac{\omega(\vartheta,\theta,2)}{\omega(\theta,\vartheta,1)}\frac{1}{N}\sum_{i=1}^{N}r_{u^{(i)}}(\theta,\vartheta,z).\label{eq: acceptance ratio modified by alpha}
\end{equation}
\end{rem}
\begin{rem}
We remark the link to some of the ideas developed in \citet{zanella:2020},
but also the differences in terms of what is being averaged and the
fact that we are not constrained to finite discrete spaces.
\end{rem}
We now turn to two illustrative examples.

\subsection{Example: exchange algorithm and some analysis\label{ex:doublyintractaveraging} }

The exchange algorithm \citep{Murray_et_al_2006} in Example \ref{ex: exchange algorithm}
lends itself to acceptance ratio averaging and can serve to illustrate
precisely the gains one may expect from the approach. Here the model
does not involve the auxiliary variable $z$, $\mathsf{U}=\mathsf{Y}$,
$Q_{\theta,\vartheta}(\cdot)$ corresponds to $\ell_{\vartheta}(\cdot)$
and $\phi$ is the identity function on $\mathsf{Y}$. We can therefore
apply the MHAAR approach described in Algorithm~\ref{alg: MHAAR for pseudo-marginal ratio in latent variable models-1}.
The algorithm takes the following form. Sample $\vartheta\sim q(\theta,\cdot)$,
then with probability $1/2$ sample $u^{(1)},\ldots,u^{(N)}\overset{{\rm iid}}{\sim}\ell_{\vartheta}(\cdot)$
and compute 
\[
r_{\mathfrak{u}}^{N}(\theta,\vartheta)=\frac{q(\vartheta,\theta)}{q(\theta,\vartheta)}\frac{\eta(\vartheta)}{\eta(\theta)}\frac{g_{\vartheta}(y)}{g_{\theta}(y)}\frac{1}{N}\sum_{i=1}^{N}\frac{g_{\theta}(u^{(i)})}{g_{\vartheta}(u^{(i)})},
\]
or (i.e., with probability 1/2) sample $u^{(1)}\sim\ell_{\vartheta}(\cdot)$
and $u^{(2)},\ldots,u^{(N)}\overset{{\rm iid}}{\sim}\ell_{\theta}(\cdot)$,
and compute $r_{\mathfrak{u}}^{N}(\vartheta,\theta)$. The interpretation
of what MHAAR achieves in this particularly simple scenario is transparent:
when $c=1$ the right hand side average is a consistent estimator
of $C_{\theta}/C_{\vartheta}$, suggesting that the algorithm can
approximate the algorithm we would have liked to implement initially.
The simplicity of this scenario, where the latent variable $z$ is
absent, also allows for a simple analysis illustrating the theoretical
benefits of Algorithm~\ref{alg: MHAAR for pseudo-marginal ratio in latent variable models-1}.
Establishing these results in full generality requires the use of
convex order tools as in \citep{Andrieu_and_Vihola_2014}, which is
far beyond the scope of this paper. Instead performance improvement
will be illustrated through numerical experiments.

Consider standard performance measures associated to a Markov transition
probability $\Pi$ of invariant distribution $\nu$ defined on some
measurable space $\big(\mathsf{E},\mathcal{E}\big)$. Let $L^{2}(\mathsf{E},\nu):=\big\{ f\colon\mathsf{E}\rightarrow\mathbb{R},\mathsf{var}_{\nu}(f)<\infty\big\}$
and $L_{0}^{2}(\mathsf{E},\nu):=L^{2}(\mathsf{E},\nu)\cap\{f\colon\mathsf{E}\rightarrow\mathbb{R},\mathbb{E}_{\nu}(f)=0\}$.
For any $f\in L^{2}(\mathsf{E},\nu)$ the asymptotic variance is defined
as 
\[
\mathsf{var}(f,\Pi):=\lim_{M\rightarrow\infty}\mathsf{var}_{\nu}\left(M^{-1/2}{\textstyle \sum}_{i=1}^{M}f(X_{i})\right),
\]
which is guaranteed to exist for reversible Markov chains (although
it may be infinite) and for a $\nu-$reversible kernel $\Pi$ its
right spectral gap 
\[
{\rm Gap}_{R}\left(\Pi\right):=\inf\{\mathcal{E}_{\Pi}(f)\,:\,f\in L_{0}^{2}(\mathsf{E},\nu),\,{\rm var}_{\nu}(f)=1\},
\]
where for any $f\in L^{2}\big(\mathsf{E},\nu\big)$ $\mathcal{E}_{\Pi}(f):=\frac{1}{2}\int_{\mathsf{E}}\nu\big({\rm d}x\big)\Pi\big(x,{\rm d}y\big)\big[f(x)-f(y)\big]^{2}$
is the so-called Dirichlet form. The right spectral gap is particularly
useful in the situation where $\Pi$ is a positive operator, in which
case ${\rm Gap}_{R}\left(\Pi\right)$ is related to the geometric
rate of convergence of the Markov chain. 

Hereafter we let $\mathring{P}^{N}(\theta,{\rm d}\vartheta)$ be the
Markov chain transition kernel corresponding to Algorithm \ref{alg: MHAAR for pseudo-marginal ratio in latent variable models-1}
in the absence of $z$. 
\begin{thm}
\label{thm:theoreticaljustification}With $P$ and $\mathring{P}^{N}$
as defined in \eqref{eq: MH transition kernel} and corresponding
to Algorithm \ref{alg: MHAAR for pseudo-marginal ratio in latent variable models-1},
respectively, 
\begin{enumerate}
\item for all $N$, ${\rm Gap}_{R}(\mathring{P}^{N})\leq{\rm Gap}_{R}(P)$
and $N\mapsto{\rm Gap}_{R}(\mathring{P}^{N})$ is non decreasing, 
\item for any $f\in L^{2}(\mathsf{X},\pi)$,
\begin{enumerate}
\item $N\mapsto\mathsf{var}(f,\mathring{P}^{N})$ is non increasing,
\item for all $N$, $\mathsf{var}(f,\mathring{P}^{N})\geq\mathsf{var}(f,P)$. 
\end{enumerate}
\end{enumerate}
\end{thm}
The proof can be found in Appendix \ref{subsec:Proof-of-Theorem justification}.
This result motivates the practical usefulness of the algorithm, in
particular in a parallel computing environment. Indeed, one crucial
property of Algorithm \ref{alg: MHAAR for pseudo-marginal ratio in latent variable models-1}
is that for both updates $Q_{1}^{N}(\cdot)$ and $Q_{2}^{N}(\cdot)$,
sampling of $u^{(1)},\ldots,u^{(N)}$ and computation of $r_{u^{(1)}}(\theta,\vartheta),\ldots,r_{u^{(N)}}(\theta,\vartheta)$
can be performed in a parallel fashion therefore opening the possibility
to improve on the variance $\mathsf{var}(f,\mathring{P})$ of estimators,
but more significantly the burn-in period of algorithms. Indeed one
could object that running $M\in\mathbb{N}^{+}$ independent chains
in parallel with $N=1$ and combining their averages, instead of using
the output from a single chain with $N=M$ would achieve variance
reduction. However our point is that the former does not speed up
convergence to equilibrium, while the latter will, in general. Unfortunately,
while estimating the asymptotic variance $\mathsf{var}(f,\mathring{P}^{N})$
from simulations is achievable, estimating time to convergence to
equilibrium is far from standard in general. The following toy example
is an exception and illustrates our point. 
\begin{example}
Here we let $\pi$ be the uniform distribution on $\Theta=\{-1,1\}$,
$\mathsf{U}=\{a,a^{-1}\}$ for $a>0$, $Q_{\theta,-\theta}(u=a)=1/(1+a)$,
$Q_{\theta,-\theta}(u=1/a)=a/(1+a)$ and 
\[
\varphi(\theta,\vartheta,\mathfrak{u},k,c)=(\vartheta,\theta,1/u^{(1)},u^{(2)},\ldots,u^{(N)},k,3-c).
\]
In other words $\mathring{P}$ can be reparametrised in terms of $a$
and with the choice $q(\theta,-\theta)=1-\alpha$ for $\alpha\in[0,1)$
we obtain 
\begin{align*}
\mathring{P}(\theta,-\theta) & =(1-\alpha)\left[\frac{1}{1+a}\min\big\{1,a\big\}+\frac{a}{1+a}\min\big\{1,a^{-1}\big\}\right].
\end{align*}
Note that there is no need to be more specific than say $Q_{\theta,\theta}(u)>0$
for $(\theta,u)\in\mathsf{X}\times\mathsf{U}$ as then a proposed
``stay'' is always accepted. Now for $N\geq2$ and $\theta\in\Theta$
we have 
\begin{align*}
\mathring{P}^{N}(\theta,-\theta) & =\frac{1-\alpha}{2}\left[\sum_{k=0}^{N}\beta^{N}(k)\min\big\{1,w_{k}(N)\big\}\right.\\
 & \hspace{1.5cm}+\left.\sum_{k=0}^{N}\left(\frac{a}{1+a}\beta^{N-1}(k-1)+\frac{1}{1+a}\beta^{N-1}(k)\right)\min\big\{1,w_{k}^{-1}(N)\big\}\right],
\end{align*}
where $\beta^{N}(k)$ is the probability mass function of the binomial
distribution of parameters $N$ and $1/(1+a)$ and $w_{k}(N):=ka/N+\big(1-k/N\big)a^{-1}.$The
second largest eigenvalue of the corresponding Markov transition matrix
is $\lambda_{2}(N)=1-2\mathring{P}^{N}(\theta,-\theta)$ from which
we find the relaxation time $T_{{\rm relax}}(N):=1/\big(2\mathring{P}^{N}(\theta,-\theta)\big)$,
and bounds on the mixing time $T_{{\rm mix}}(\epsilon,N)$, that is
the number of iterations required for the Markov chain to have marginal
distribution within $\epsilon$ of $\pi$, in the total variation
distance, \citet[Theorem 12.3 and Theorem 12.4]{levin2017markov}
\[
-(T_{{\rm relax}}(N)-1)\log(2\epsilon)\leq T_{{\rm mix}}(\epsilon,N)\leq-T_{{\rm relax}}(N)\log(\epsilon/2).
\]
We define the relative burn-in time fraction, $\gamma(N):=T_{{\rm relax}}(N)/T_{{\rm relax}}(1)$,
which is independent of $\alpha$ and captures the benefit of MHAAR
in terms of convergence to equilibrium. In Figure~\ref{fig:toy-example-relaxation}
we present the evolution of $N\mapsto\gamma(N)$ for $a=2,5,10$ and
$\gamma(1000)$ as a function of $a$. As expected the worse the algorithm
corresponding to $\mathring{P}$ is, the more beneficial averaging
is: for $a=2,5,10$ we observe running time reductions of approximately
$35\%$, $65\%$ and $80\%$ respectively. This suggests that computationally
cheap, but possibly highly variable, estimators of the acceptance
ratio may be preferable to reduce burn-in when a parallel machine
is available and communication costs are negligible.

\begin{figure}[!h]
\centerline{\includegraphics[width=0.7\textwidth]{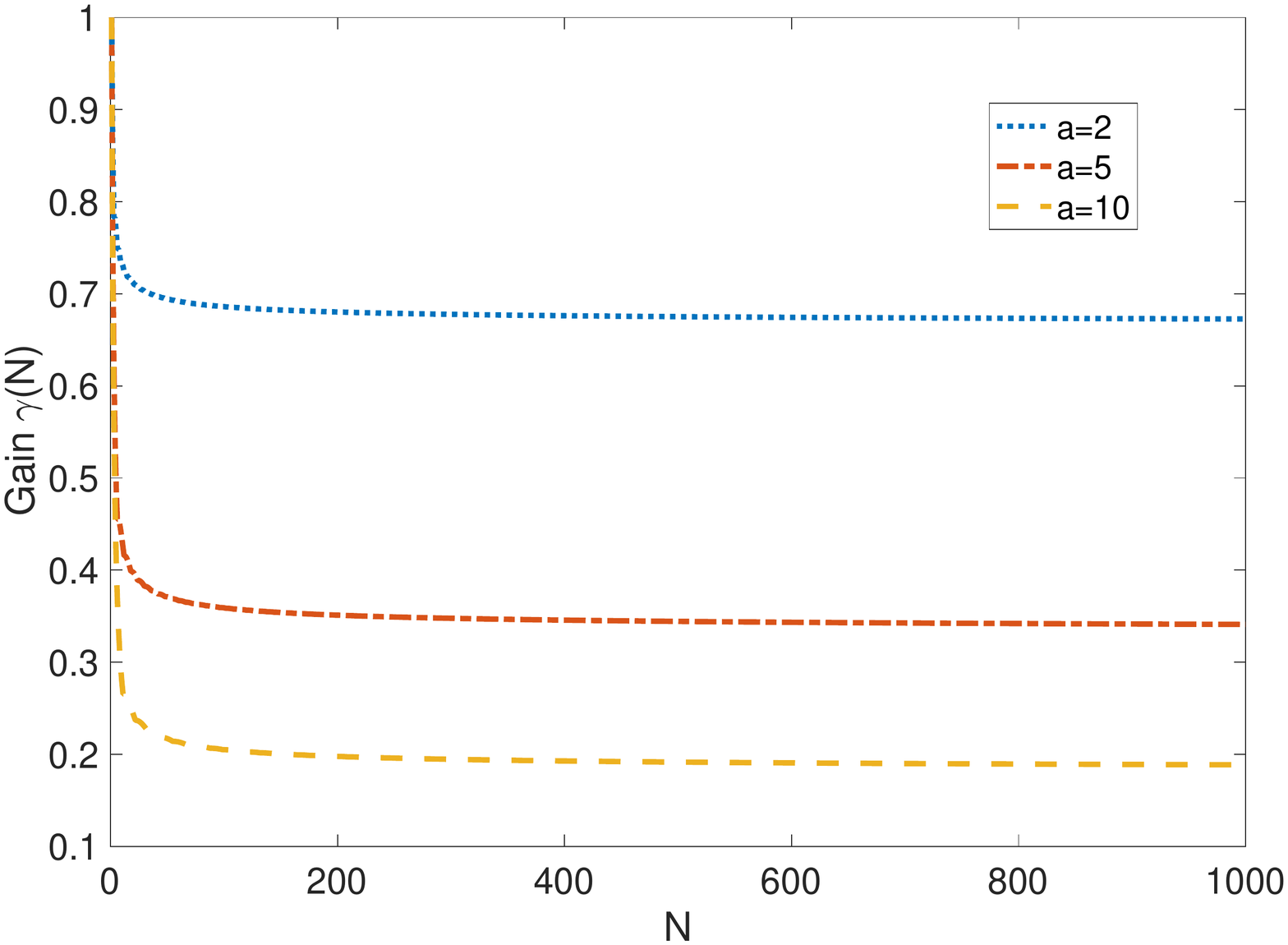}
\includegraphics[width=0.35\textwidth]{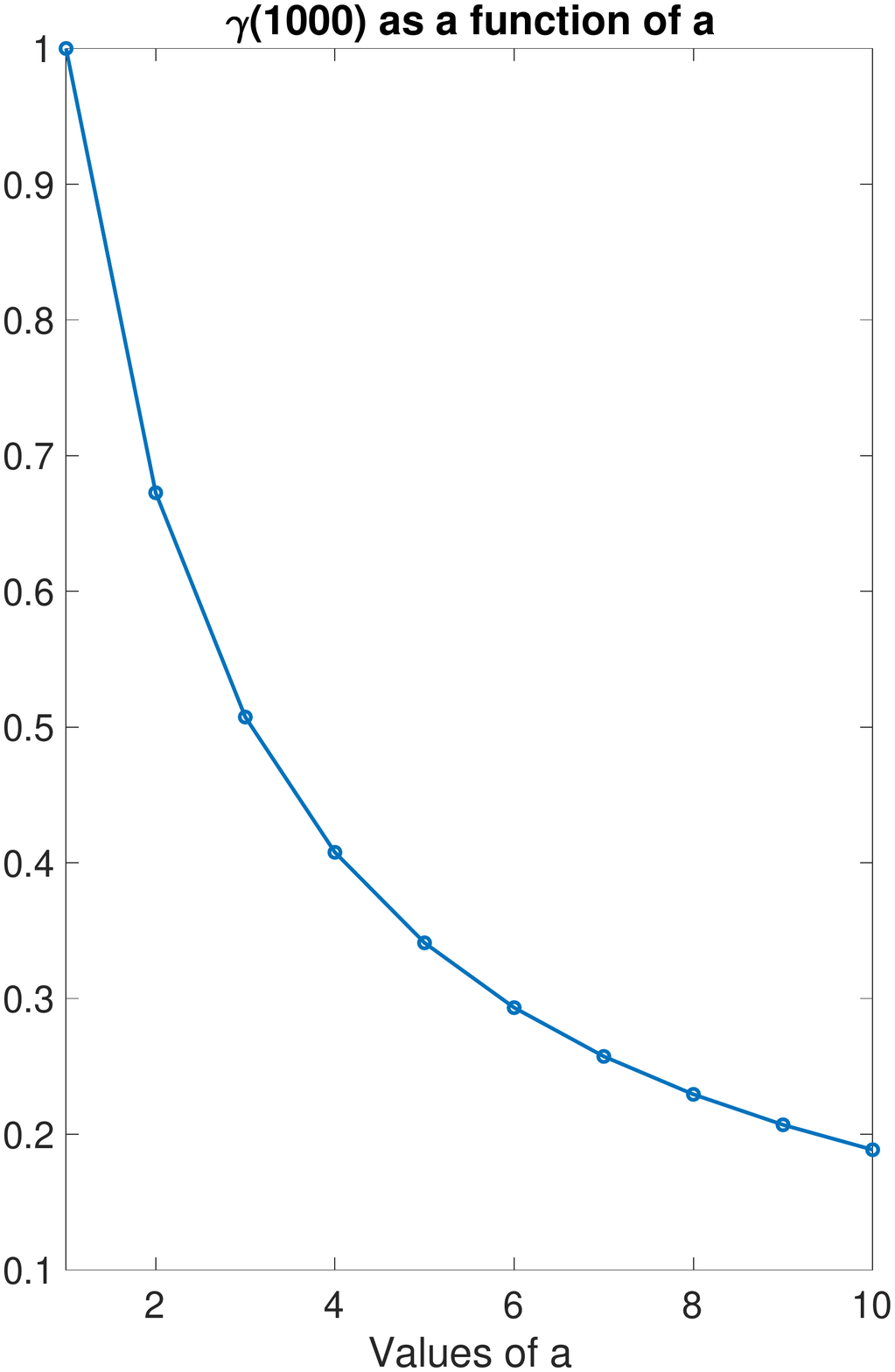}}
\caption{\label{fig:toy-example-relaxation}Evolution of $N\mapsto\gamma(N)$
for $a=2,5,10$ and $\gamma(1000)$ as a function of $a$.}
\end{figure}
\end{example}

\subsection{Example: improving transdimensional samplers\label{subsec:Example:-improving-transdimensio}}

The following example motivates the scenario considered in this section
and on which we illustrate the interest of the proposed approach.
\begin{example}[Poisson multiple change-point model]
\label{ex: Poisson multiple change-point model} The UK coal-mining
disasters dataset consists of $n$ records $y_{1:n}$ of the number
of disasters at a given set of dates. In \citep{Green_1995}, it is
proposed to model the dataset with a non-homogenous Poisson process
model on the time interval $[0,L]$ with a step-wise constant intensity
function with changepoints $0=s_{0}<s_{1}\ldots<s_{m}=L$ and heights
$h_{1},\ldots,h_{m}\in\mathbb{R}_{+}$ for some $m\in\mathbb{N}$.
Letting $z_{m}:=\big(\{s_{j}\}_{j=0}^{m},\{h_{j}\}_{j=1}^{m}\big)$
the data likelihood under `model' $m$ is therefore
\[
\log\mathcal{L}_{m}\left(y_{1:n};z_{m}\right)=\sum_{j=1}^{m}\log h_{j}\left(\sum_{i=1}^{n}\mathbb{I}_{[s_{j-1},s_{j})}(y_{i})\right)-\sum_{j=1}^{m}h_{j}(s_{j}-s_{j-1}),
\]
and inferring $(m,z_{m})$ is of interest. In a Bayesian framework
one can ascribe a prior to $(m,z_{m})$ and infer both model and within
model parameters from the associated posterior distribution. Sampling
from such transdimensional distribution requires the use of a particular
type of MH update, as proposed in \citet{Green_1995}. Such algorithms
may be difficult to design and we show how they can benefit from our
approach.
\end{example}
In this section we consider target distributions $\pi(\theta,{\rm d}z_{\theta})$
on $\cup_{\vartheta\in\Theta}\{\vartheta\}\times\mathsf{Z}_{\vartheta}$,
where in general $\Theta\subseteq\mathbb{N}$ and the dimension $d_{\theta}$
of $\mathsf{Z}_{\theta}\subset\mathbb{R}^{d_{\theta}}$ depends on
$\theta$. We assume that $\pi(\theta,{\rm d}z_{\theta})$ admits
a density $\pi(\theta,z_{\theta})$ known up to a normalising constant,
where $z_{\theta}$ is a within model parameter. When sampling from
this distribution a particular challenge is to define transdimensional
transitions from $(\theta,z_{\theta})$ to $(\vartheta,z_{\vartheta})$
in situations where $d_{\theta}\neq d_{\vartheta}$ and we focus on
such updates only here. Practical algorithms consists of mixtures
of such updates and more traditional within model updates \citet{Green_1995}.
Our aim here is to outline the solution proposed by \citet{Green_1995}
and show how it fits, up to minor modifications, in the framework
outlined in Subsection \ref{subsec: Pseudo-marginal algorithm with averaged acceptance ratio estimator}
and can benefit from the MHAAR methodology. 

The main idea of \citet{Green_1995} consists, for $\theta,\vartheta\in\Theta$,
of augmenting the within model parameters to ensure dimension matching,
that is $(z_{\theta},u_{\theta,\vartheta})\in\mathsf{Z}_{\theta,\vartheta}:=\mathsf{Z}_{\theta}\times\mathsf{U}_{\theta,\vartheta}$,
$(z_{\vartheta},u_{\vartheta,\theta})\in\mathsf{Z}_{\vartheta,\theta}:=\mathsf{Z}_{\vartheta}\times\mathsf{U}_{\vartheta,\theta}$,
with $\mathsf{U}_{\theta,\vartheta}\subset\mathbb{R}^{d_{\theta,\vartheta}}$,
$\mathsf{U}_{\vartheta,\theta}\subset\mathbb{R}^{d_{\vartheta,\theta}}$
for $d_{\theta,\vartheta},d_{\vartheta,\theta}\in\mathbb{N}$ such
that $d_{\theta}+d_{\theta,\vartheta}=d_{\vartheta}+d_{\vartheta,\theta}$,
and defining extended distributions
\[
\pi(\theta,{\rm d}(z_{\theta},u_{\theta,\vartheta}))=\pi(\theta,{\rm d}z_{\theta})Q_{\theta,\vartheta,z_{\theta}}({\rm d}u_{\theta,\vartheta}),
\]
for some probability distribution $Q_{\theta,\vartheta,z_{\theta}}(\cdot)$
on $\mathsf{U}_{\theta,\vartheta}$. Note that in some scenarios we
may have $d_{\theta,\vartheta}=0$ (resp. or $d_{\vartheta,\theta}=0$),
in which case $u_{\theta,\vartheta}$ (resp. $u_{\vartheta,\theta})$
should be ignored. 
\begin{example}[Poisson multiple change-point model (ctd.)]
 For the coal-mining disaster a transdimensional update may consist
of adding or removing a changepoint and its height, in which case
$u_{m,m+1}=(s_{*},h_{*},j)\in[0,L]\times\mathbb{R}_{+}\times\llbracket m+1\rrbracket$
and $u_{m,m-1}\in\llbracket m\rrbracket$. A possible choice for the
distributions is the uniform distribution for $u_{m,m-1}$ (a randomly
chosen changepoint is removed) the uniform distribution for $s_{*}$,
and the prior distribution for $h_{*}$, in which case $j$ is a deterministic
function of $s_{*}$ and $z_{m}$. This update is referred to as `birth-death'.
\end{example}
Together with an invertible mapping $\phi_{\theta,\vartheta}:\mathsf{Z}_{\theta,\vartheta}\rightarrow\mathsf{Z}_{\vartheta,\theta}$
such that $\phi_{\theta,\vartheta}^{-1}=\phi_{\vartheta,\theta}$
this allows one to define the involution $\varphi(\theta,\vartheta,z_{\theta},u_{\theta,\vartheta})=(\vartheta,\theta,z_{\vartheta},u_{\vartheta,\theta})$
with $(z_{\vartheta},u_{\vartheta,\theta})=\phi_{\theta,\vartheta}(z_{\theta},u_{\theta,\vartheta})$,
and hence a valid MH type update \citet{Green_1995}. While the choice
of $\mathsf{U}_{\theta,\vartheta}$ and $\phi_{\theta,\vartheta}$
are often natural for numerous problems, choosing the distribution
$Q_{\theta,\vartheta}$ can be difficult and result in poor performance.
Our aim here is to show that averaging acceptance ratios of the standard
procedure over multiple matching variables can improve performance
significantly. The MHAAR algorithm in this context, which we call
Reversible-multiple-jump MCMC (RmJ-MCMC), follows along the lines
of Subsection \ref{subsec: Pseudo-marginal algorithm with averaged acceptance ratio estimator}.

The RmJ-MCMC update is described in detail in Algorithm~\ref{alg: Reversible multiple jump MCMC}
where we have taken into account Remark \ref{rem: Choosing Q1 and Q2 with different probabilities}
and changed the distribution of $c$, but also taken into account
that for $c=2$ the nature of the auxiliary variables may differ for
$i=k$ and $i\neq k$. In Algorithm \ref{alg: Reversible multiple jump MCMC},
$r_{u}(\theta,\vartheta,z_{\theta})$ is the acceptance rate of the
standard RJ-MCMC algorithm when the current sample is $(\theta,z_{\theta})$,
$\vartheta$ is proposed from $q(\theta,\cdot)$ and $u_{\theta,\vartheta}$
is the dimension-matching variable sampled from $Q_{\theta,\vartheta,z_{\theta}}(\cdot)$.
One can check that Algorithm is a special case of MHAAR given in Algorithm
\ref{alg: MHAAR for pseudo-marginal ratio in latent variable models-1},
where the space of the latent variable depends on $\theta$, and likewise
the space of auxiliary variables, which are the dimension-matching
variables of RmJ-MCMC, depends on $\theta,\vartheta$ as well as $c$. 

\begin{algorithm}[!h]
\caption{RmJ-MCMC: MHAAR for trans-dimensional models.}
\label{alg: Reversible multiple jump MCMC}

\KwIn{Current sample $(\theta,z_{\theta})$}

\KwOut{New sample}

Sample $\vartheta\sim q(\theta,\cdot)$ and $c\sim\mathcal{P}(\omega(\theta,\vartheta,z_{\theta},1),\omega(\theta,\vartheta,z_{\theta},2))$.

\If{$c=1$}{

\For{$i=1,\ldots,N$}{

Sample $u_{\theta,\vartheta}^{(i)}\sim Q_{\theta,\vartheta,z_{\theta}}(\cdot)$}

Sample $k\sim\mathcal{P}\big(r_{u_{\theta,\vartheta}^{(1)}}(\theta,\vartheta,z_{\theta}),\ldots,r{}_{u_{\theta,\vartheta}^{(N)}}(\theta,\vartheta,z_{\theta})\big)$,
set $z'_{\vartheta}=\phi_{\theta,\vartheta}^{[1]}(z_{\theta},u_{\theta,\vartheta}^{(k)})$.\\
Return $(\vartheta,z_{\vartheta})$ with probability $\min\{1,r_{\mathfrak{u}}^{N}(\theta,\vartheta,z_{\theta})\},$
otherwise return $(\theta,z_{\theta})$.

}\Else{

Sample $k\sim\text{Unif}(\left\llbracket N\right\rrbracket )$ and
$u_{\theta,\vartheta}^{(k)}\sim Q_{\theta,\vartheta,z_{\theta}}(\cdot)$,
set $z'_{\vartheta}=\phi_{\theta,\vartheta}^{[1]}(z_{\theta},u_{\theta,\vartheta}^{(k)})$.
\\
\For{$i=1,\ldots,N,i\neq k$}{

Sample $u_{\vartheta,\theta}^{(i)}\sim Q_{\vartheta,\theta,z_{\vartheta}}(\cdot)$.

}

Return $(\vartheta,z'_{\vartheta})$ with probability $\min\{1,r_{\mathfrak{u}'}^{N}(\vartheta,\theta,z'_{\vartheta})^{-1}\}$,
otherwise return $(\theta,z_{\theta})$.

}
\end{algorithm}

\begin{example}[Poisson multiple change-point model (ctd.)]
\label{ex:RJMCMC2}\emph{ }We now evaluate this approach on the coal-mining
disaster example. In order to improve computational efficiency we
set $\omega(m,m+1,1)=1$ and $\omega(m,m-1,1)=0$. Indeed when attempting
a birth it is preferable to average over the continuous valued $u_{m,m+1}$
rather than the discrete valued $u_{m,m-1}$, in particular when $N\gg m+1$.
The priors chosen are as in \citet{Green_1995} and the specifics
of the MCMC move for updating the latent variables within model are
chosen as in \citet{Karagiannis_and_Andrieu_2013}. To illustrate
the gains in terms of convergence to equilibrium of our scheme we
had $\text{3000}$ independent runs started at the same point $x_{0}$,
estimated the expectations $\mathbb{E}^{N}\big[\mathbb{I}\{M_{t}=m\}\big]$
by an ensemble average, and reported $\big|\hat{\pi}(m)-3000^{-1}\sum_{k=1}^{3000}\mathbb{I}\{M_{t}^{(k)}=m\}\big|$
for $m\in\{1,\ldots,8\}$ and $N=1,10,100$ in Figure \ref{fig:reversible-jump-burnin}
where $\hat{\pi}(m)$ was estimated by a realisation of length $10^{6}$
with $N=90$ and $T=50$, discarding the burn-in. We see that the
approach appears to reduce time to convergence to equilibrium by the
order of $50\%$. We also generated $K=10^{6}$ samples to compute
the IAC for $m$. Figure \ref{fig: reversiblejumpexample} indicates
a variance reduction of the order of $60\%$ at $N=130$. We also
provide results for the scheme used in \citet{Karagiannis_and_Andrieu_2013}
(referred to as AIS for Annealing Importance Sampling) for illustration.
For $T$ (a tuning parameter of the algorithm) large the algorithm
approaches the algorithm which would sample from the model distribution
directly as $T\rightarrow\infty$. Our algorithm achieves similar
performance improvement, but is parallelisable.

\begin{figure}[!h]
\includegraphics[width=1\textwidth]{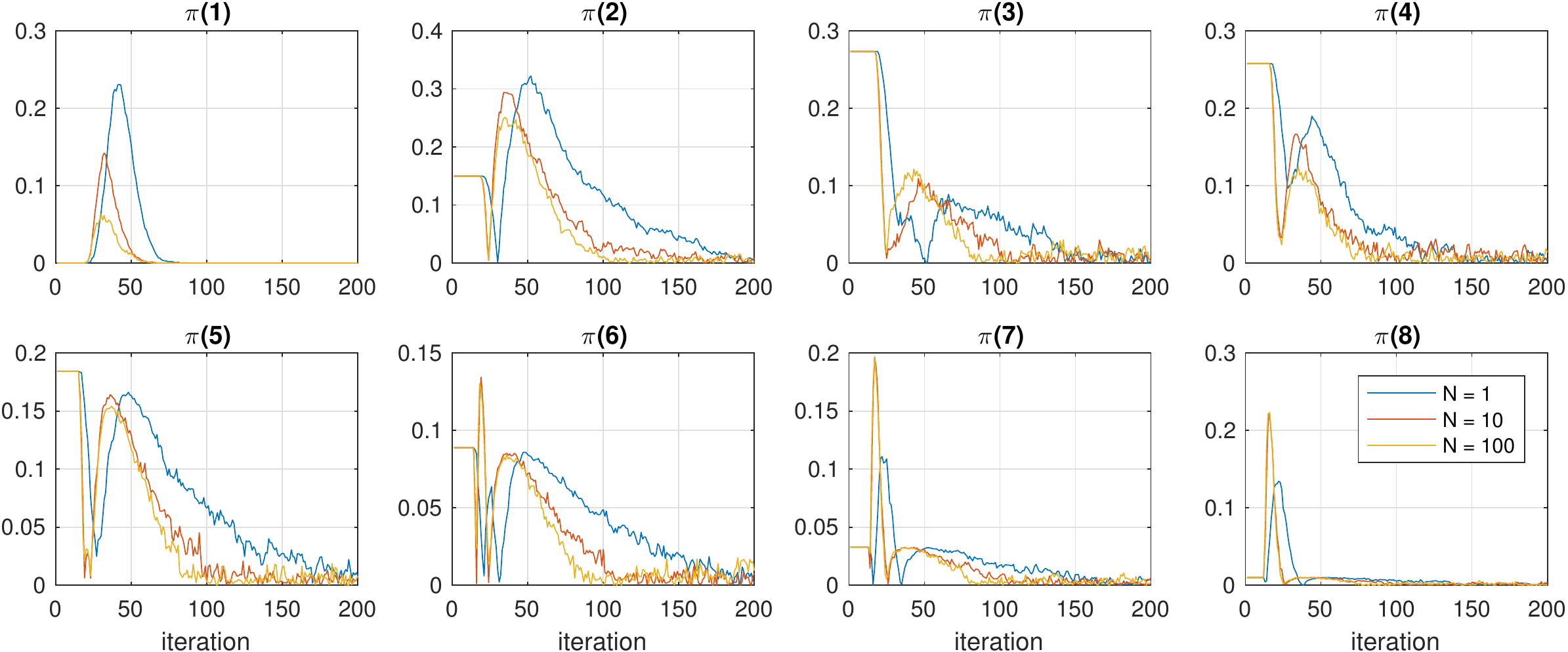}
\caption{Estimates of time to convergence of $\mathbb{E}_{x_{0}}^{N}\big[f_{m}(X_{i})\big]$
to $\pi(m)$ for $N=1,10,100$.}
\label{fig:reversible-jump-burnin}
\end{figure}

\begin{figure}[!h]
\centerline{\includegraphics[scale=0.6]{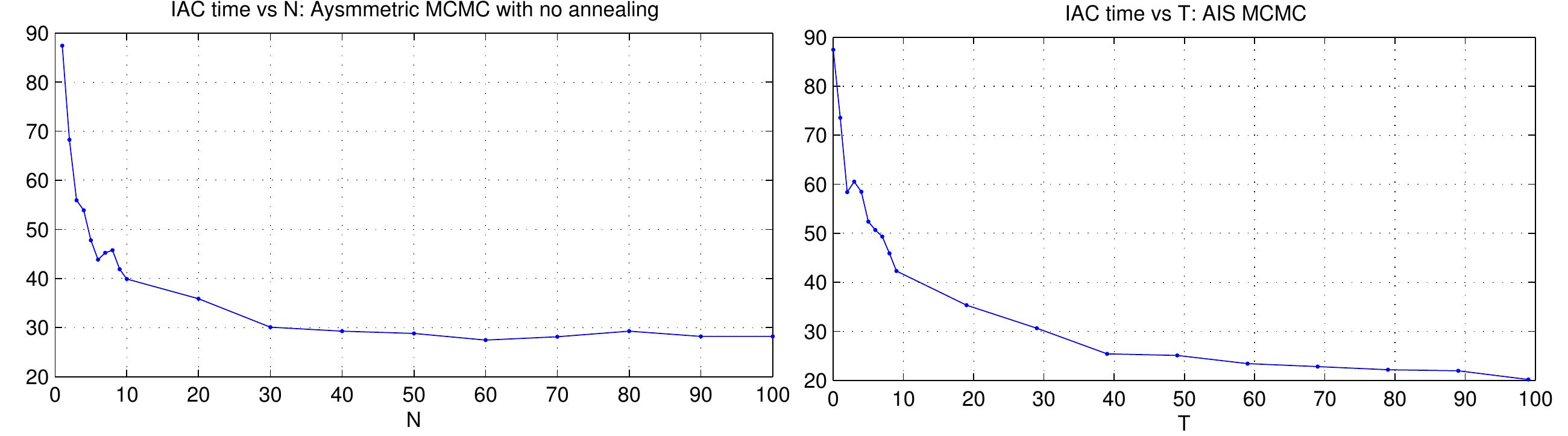}}
\protect\protect\caption{Left: IAC for $m$ vs number of particles $N=1,2,\ldots,10,20,\ldots,100$.
Right: IAC for $m$ vs number of particles $T=0,1,2,\ldots,10,20,\ldots,100$.}
\label{fig: reversiblejumpexample}
\end{figure}
\end{example}

\section{An efficient alternative to pseudo-marginal algorithms\label{sec: MHAAR via Rao-Blackewellisation of PMR}}

  In this section we consider a class of latent variable models
of probability density defined on $\Theta\times\mathsf{Z}^{T}$ for
some $T\geq1$ and of the form
\[
\pi(\theta,z)\propto\eta(\theta)\prod_{t=1}^{T}\gamma_{t,\theta}(z_{t}),
\]
with $z=z_{1:T}$, $\eta(\theta)$ a prior density, and \textbf{$\gamma_{t,\theta}(z_{t})$
}is typically a complete likelihood function depending on some observation
$y_{t}$, see the examples in this section. We drop any such dependencies
from notation for simplicity. It was shown in \citet{Yildirim_et_al_2018}
that it is possible to develop efficient sampling schemes for such
models which in particular scale favourably with $T$ large. We show
here that these algorithms can be further improved at little cost
using the methodology developed in this paper, leading in particular
to alternative to pseudo-marginal algorithms \citet{Andrieu_and_Roberts_2009}
with much better performance. We will show in Section \ref{sec: State-space models: SMC and conditional SMC within MHAAR}
how these ideas can be extended to the context of state-space models.

\subsection{A novel consistent pseudo-marginal estimator\label{sec: A novel consistent pseudo-marginal estimator}}

The algorithm we develop can be thought of as being the numerical
approximation of the scheme in Subsection \ref{subsec: Pseudo-marginal ratio algorithms}
where $u=z'_{1:T}$ the proposed new values of the MH update and $Q_{\theta,\vartheta,z}$
has density proportional to $\prod_{t=1}^{T}\gamma_{t,\theta,\vartheta}(z'_{t})$
. This cannot be achieved in practice and instead replace this update
with a Markov kernel reversible with respect to this distribution,
in the spirit of \citet{Neal_2004}, dependent on a parameter $M\text{\ensuremath{\in\mathbb{N}_{*}}}$
and such that as $M\rightarrow\infty$ the algorithm approaches the
idealised algorithm. An interesting feature is that this kernel produces
multiple samples which can be used in the averaging procedure, at
very little extra cost.

We first introduce the algorithm of \citet{Yildirim_et_al_2018} and
identify computational inefficiencies which can be addressed with
a MHAAR strategy. For $\theta,\vartheta\in\Theta$ and $t\in\llbracket T\rrbracket$,
let $q_{t,\theta,\vartheta}$ be a probability distribution on $(\mathsf{Z},\mathcal{Z})$
and with $M\geq2$ and $\mathfrak{u}:=u_{1:T}^{(1:M-1)}\in\mathsf{Z}^{(M-1)T},$
let
\begin{equation}
\Phi_{\theta,\vartheta}(\mathrm{d}\mathfrak{u}):=\prod_{t=1}^{T}\prod_{i=1}^{M-1}q_{t,\theta,\vartheta}(\mathrm{d}u_{t}^{(i)}).\label{eq: cSMC for MHAAR-RB}
\end{equation}
For notational simplicity we define $\mathbf{v}=v_{1:T}^{(1:M)}\in\mathsf{Z}^{MT}$
such that $(v_{1}^{(1)},\ldots v_{T}^{(1)})=z$ and $v_{1:T}^{(2:M)}=\mathfrak{u}$.
For any index sequence $\mathbf{k}=(k_{1},\ldots,k_{T})\in\left\llbracket M\right\rrbracket ^{T}$,
we define $v^{(\mathbf{k})}:=(v_{1}^{(k_{1})},\ldots,v_{T}^{(k_{T})})$,
so that $z=v^{(\mathbf{1})}$ where $\mathbf{1}:=(1,\ldots,1)$ is
the vector of size $T$ consisting of 1's. The proposal mechanism
of the algorithms considered consists of sampling candidates from
$\Phi_{\theta,\vartheta}(\mathrm{d}\mathfrak{u})$ and then attempting
a swap of $v^{(\mathbf{k})}$ and , where $\mathbf{k}\in\left\llbracket M\right\rrbracket ^{T}$
is sampled conditional on $\mathbf{v}$ according to
\begin{equation}
b_{\theta,\vartheta}(\mathbf{k}|\mathbf{v}):=\prod_{t=1}^{T}\frac{\gamma_{t,\theta,\vartheta}(v_{t}^{(k_{t})})/q_{t,\theta,\vartheta}(v_{t}^{(k_{t})})}{\sum_{j=1}^{M}\gamma_{t,\theta,\vartheta}(v_{t}^{(j)})/q_{t,\theta,\vartheta}(v_{t}^{(j)})}.\label{eq: Selection probabilities of MHAAR-RB-1}
\end{equation}
Here $\gamma_{t,\theta,\vartheta}(v)$ is a user defined probability
density on $(\mathsf{Z},\mathcal{Z})$\textendash possible choices
include $\gamma_{t,(\theta+\vartheta)/2}(v)$ or $\gamma_{t,\theta}(v)$.
Using the framework of Subsection \ref{subsec:A-general-perspective}
we let $\xi=(\theta,\vartheta,\mathbf{v},\mathbf{k})$,
\[
\mathring{\pi}(\mathrm{d}\xi):=\pi(\mathrm{d}(\theta,z))q(\theta,\mathrm{d}\vartheta)\Phi_{\theta,\vartheta}(\mathrm{d}\mathfrak{u})b_{\theta,\vartheta}(\mathbf{k}|\mathbf{v}),
\]
and consider the involution $\varphi(\theta,\vartheta,\mathbf{v},\mathbf{k})=(\vartheta,\theta,\mathfrak{s}_{\mathbf{\mathbf{1}},\mathbf{k}}(\mathbf{v}),\mathbf{\mathbf{k}})$
where $\mathfrak{s}_{\mathbf{\mathbf{1}},\mathbf{k}}:\mathsf{Z}^{MT}\mapsto\mathsf{Z}^{MT}$
is the operator on $\mathbf{v}$ which swaps $v^{(\mathbf{1})}$ and
$v^{(\mathbf{k})}$, that is, if $\mathbf{v}'=\mathfrak{s}_{\mathbf{\mathbf{1}},\mathbf{k}}(\mathbf{v})$,
it satisfies
\begin{equation}
\mathbf{v}_{t}^{\prime(i)}=\begin{cases}
v_{t}^{(1)} & \text{for }i=k_{t},\\
v_{t}^{(k_{t})} & \text{for }i=1,\\
v_{t}^{(i)} & \text{otherwise}.
\end{cases},\quad t=1,\ldots,n;i=1,\ldots,M.\label{eq: swapping operator for MHAAR-RB}
\end{equation}
The corresponding acceptance ratio can be shown to be $r_{v^{(\mathbf{1})},v^{(\mathbf{k})}}(\theta,\vartheta)$,
where, for any $z,z'\in\mathsf{Z}^{T}$,
\begin{equation}
r_{z,z'}(\theta,\vartheta):=\frac{q(\vartheta,\theta)\eta(\vartheta)}{q(\theta,\vartheta)\eta(\theta)}\prod_{t=1}^{T}\frac{\gamma_{t,\theta,\vartheta}(z_{t})}{\gamma_{t,\theta}(z_{t})}\frac{\gamma_{t,\vartheta}(z'_{t})}{\gamma_{t,\theta,\vartheta}(z'_{t})},\label{eq: PMR acceptance ratio multiple latent variable}
\end{equation}
 We will refer to this algorithm as AIS MCMC, since the proposal
mechanism for $z$ can be viewed as a one-step annealing using the
`intermediate' distribution with (unnormalised) density $\gamma_{t,\theta,\vartheta}(\cdot)$,
building on the ideas in \citet{Neal_2004}. While this algorithm
can be shown to be efficient in the regime $T\rightarrow\infty$ by
appropriate scaling of $\vartheta-\theta$, it should be clear that
the use of a single ``path'' $\mathbf{k}$ in $\mathbf{v}$ is wasteful
and the use of the ``Rao-Blackwellised'' acceptance ratio
\begin{equation}
r_{\mathbf{\mathbf{1}},\mathbf{v}}(\theta,\vartheta):=\sum_{\mathbf{k}\in\llbracket M\rrbracket^{T}}b_{\theta,\vartheta}(\mathbf{k}|\mathbf{v})r_{v^{(\mathbf{1})},v^{(\mathbf{k})}}(\theta,\vartheta),\label{eq: MHAAR-RB acceptance ratio-iid model}
\end{equation}
may be preferable. Before showing how this can be achieved within
the MHAAR framework we take a closer look at $r_{\mathbf{\mathbf{1}},\mathbf{v}}(\theta,\vartheta)$,
which further motivates these algorithms. Rearranging terms (see Theorem
\ref{subsec: Pseudo-marginal algorithm with averaged acceptance ratio estimator})
it can be shown that
\begin{equation}
r_{\mathbf{1},\mathbf{v}}(\theta,\vartheta)=\frac{q(\vartheta,\theta)\eta(\vartheta)}{q(\theta,\vartheta)\eta(\theta)}\prod_{t=1}^{T}\frac{\gamma_{t,\theta,\vartheta}(v_{t}^{(1)})}{\gamma_{t,\theta}(v_{t}^{(1)})}\prod_{t=1}^{T}\frac{\sum_{i=1}^{M}\gamma_{t,\vartheta}(v_{t}^{(i)})/q_{t,\theta,\vartheta}(v_{t}^{(i)})}{\sum_{j=1}^{M}\gamma_{t,\theta,\vartheta}(v_{t}^{(j)})/q_{t,\theta,\vartheta}(v_{t}^{(j)})},\label{eq: acceptance ratio-expanded-iid model}
\end{equation}
implying in particular that this can be computed in $\mathcal{O}(MT)$
operations and not $\mathcal{O}(M^{T})$ as suggested by our earlier
expression. It is worth noting that for any $\theta,\vartheta\in\Theta$,
this is an unbiased estimator of $r(\theta,\vartheta)$ when $z=v^{(\mathbf{1})}\sim\pi({\rm d}z\mid\theta)$
\textendash{} this is established in a more general context in Theorem
\ref{thm: SMC unbiased estimator of acceptance ratio} in Section
\ref{subsec: MHAAR with cSMC for SSM}. The choice $\gamma_{t,\theta,\vartheta}=\gamma_{t,\theta}$
leads to
\begin{equation}
r_{\mathbf{1},\mathbf{v}}(\theta,\vartheta)=\frac{q(\vartheta,\theta)\eta(\vartheta)}{q(\theta,\vartheta)\eta(\theta)}\prod_{t=1}^{T}\frac{\sum_{i=1}^{M}\gamma_{t,\vartheta}(v^{(i)})/q_{t,\theta,\vartheta}(v^{(i)})}{\sum_{j=1}^{M}\gamma_{t,\theta}(v^{(j)})/q_{t,\theta,\vartheta}(v^{(j)})},\label{eq: MHAAR-RB acceptance ratio with no annealing-iid model}
\end{equation}
which is reminiscent of the acceptance ratio of a pseudo-marginal
algorithm \citet{Andrieu_and_Roberts_2009} where importance sampling
is used to estimate the likelihood function. However the crucial difference
here is that only one set of auxiliary variables, sampled afresh at
each iteration, is used to estimate the numerator and denominator
of $r(\theta,\vartheta)$ in \eqref{eq:genericMHacceptratio}, leading
to reduced variability and improved performance \textendash{} as pointed
out in Subsection 1.2, for a pseudo-marginal algorithm a poor draw
of the denominator leads to the algorithm getting stuck in the same
state for a large number of iterations. This algorithm can be thought
of as an alternative to the correlated pseudo-marginal algorithm of
\citet{Deligiannidis_et_al_2018}.

The new algorithm, MHAAR-RB (for Rao-Blackwellised) hereafter, is
obtained by alternating between two sampling mechanisms for $\mathbf{k}$.
Let $\xi=(\theta,\vartheta,\mathbf{v},\mathbf{k},c)\in\Theta^{2}\times\mathsf{Z}^{MT}\times\left\llbracket M\right\rrbracket ^{T}\times\{1,2\}$
and
\begin{equation}
\mathring{\pi}(\mathrm{d}\xi):=\frac{1}{2}\pi(\mathrm{d}(\theta,z))Q_{c}^{M}\big((\theta,z);\mathrm{d}(\mathfrak{u},\mathbf{k})\big)\label{eq: joint distribution for MHAAR-RB-iid model}
\end{equation}
with
\begin{align*}
Q_{1}^{M}\big((\theta,z);\mathrm{d}(\mathfrak{u},\mathbf{k})\big) & :=q(\theta,\mathrm{d}\vartheta)\Phi_{\theta,\vartheta}(z,\mathrm{d}\mathfrak{u})b_{\theta,\vartheta}^{\text{(1)}}(\mathbf{k}|\mathbf{v}),\\
Q_{2}^{M}\big((\theta,z);\mathrm{d}(\mathfrak{u},\mathbf{k})\big) & :=q(\theta,\mathrm{d}\vartheta)\Phi_{\vartheta,\theta}(z,\mathrm{d}\mathfrak{u})b_{\theta,\vartheta}^{(2)}(\mathbf{k}|\mathbf{v}).
\end{align*}
where the sampling probabilities are given as

\begin{equation}
b_{\theta,\vartheta}^{\text{(1)}}(\mathbf{k}|\mathbf{v}):=\prod_{t=1}^{T}\frac{\gamma_{t,\vartheta}(v^{(k_{t})})/q_{t,\theta,\vartheta}(v^{(k_{t})})}{\sum_{j=1}^{M}\gamma_{t,\vartheta}(v^{(j)})/q_{t,\theta,\vartheta}(v^{(j)})},\label{eq: Selection probabilities of MHAAR-RB-iid model}
\end{equation}
which can be shown to be obtained by weighting $b_{\theta,\vartheta}(\mathbf{k}|\mathbf{v})$
by the acceptances ratio $\mathring{r}_{\mathbf{\mathbf{1}},\mathbf{k},\mathbf{v}}(\theta,\vartheta)$
corresponding to $\mathbf{k}$, and $b_{\theta,\vartheta}^{\text{(2)}}(\mathbf{k}|\mathbf{v})=b_{\vartheta,\theta}(\mathbf{k}|\mathbf{v})$.
Given the current sample $(\theta,z)\in\Theta\times\mathsf{Z}$, an
update of MHAAR-RB proceeds as follows:
\begin{enumerate}
\begin{singlespace}
\item Sample $c\sim\text{Unif}(\left\{ 1,2\right\} )$, $(\vartheta,\mathfrak{u},\mathbf{k})\sim Q_{c}^{N}\big((\theta,z);\cdot\big)$
and form $\xi=(\theta,\vartheta,\mathbf{v},\mathbf{k},c)$.
\item With $\mathfrak{s}_{\mathbf{1},\mathbf{k}}(\mathbf{v})$ as in \eqref{eq: swapping operator for MHAAR-RB},
let
\begin{equation}
\xi'=\varphi(\theta,\vartheta,\mathbf{v},\mathbf{k},c):=(\vartheta,\theta,\mathfrak{s}_{\mathbf{1},\mathbf{k}}(\mathbf{v}),\mathbf{k},3-c).\label{eq: involution for MHAAR-RB-iid model}
\end{equation}
\item Return $\xi'$ with probability $\min\left\{ 1,\mathring{r}(\xi)\right\} $,
otherwise return $\xi$, where
\begin{equation}
\mathring{r}(\xi):=\begin{cases}
r_{\mathbf{1},\mathbf{v}}(\theta,\vartheta), & c=1\\
1/r_{\mathbf{k},\mathbf{v}}(\vartheta,\theta), & c=2
\end{cases},\label{eq:MHAAR-RB acceptance ratio - general - iid model}
\end{equation}
with, for $\mathbf{l}\in\llbracket M\rrbracket^{T}$,
\begin{align}
r_{\mathbf{l},\mathbf{v}}(\theta,\vartheta):=\frac{q(\vartheta,\theta)\eta(\vartheta)}{q(\theta,\vartheta)\eta(\theta)}\prod_{t=1}^{T}\frac{\gamma_{t,\theta,\vartheta}(v_{t}^{(l_{t})})}{\gamma_{t,\theta}(v_{t}^{(l_{t})})}\prod_{t=1}^{T}\frac{\sum_{i=1}^{M}\gamma_{t,\vartheta}(v_{t}^{(i)})/q_{t,\theta,\vartheta}(v_{t}^{(i)})}{\sum_{j=1}^{M}\gamma_{t,\theta,\vartheta}(v_{t}^{(j)})/q_{t,\theta,\vartheta}(v_{t}^{(j)})} & .\label{eq: MHAAR-RB acceptance ratio-expanded-iid model-repeated}
\end{align}
\end{singlespace}
\end{enumerate}
The following theorem, whose proof is left to Appendix \ref{subsec: Acceptance ratio of MHAAR-RB},
establishes the correctness of the acceptance ratio above.
\begin{thm}
\label{thm: acceptance ratio for multiple latent variable model}The
acceptance ratio resulting from the choices of $\mathring{\pi}$ as
in \eqref{eq: joint distribution for MHAAR-RB-iid model} and the
involution as in \eqref{eq: involution for MHAAR-RB-iid model} is
given by \eqref{eq:MHAAR-RB acceptance ratio - general - iid model}-\eqref{eq: MHAAR-RB acceptance ratio-expanded-iid model-repeated}.
\end{thm}
A detailed pseudo-code of MHAAR-RB is given in Algorithm \ref{alg: MHAAR-RB}.
When $c=1$, the acceptance ratio does not depend on $\mathbf{k}$,
which can be taken advantage of by sampling $\mathbf{k}$ upon acceptance
only. Notice also the optional stage which has not been discussed
yet. These are motivated by the fact that the proposed variables $(\vartheta,v^{(\mathbf{k})})$
are either accepted or rejected jointly and it seems natural, upon
rejection, to attempt to refresh the current latent variable only,
i.e. attempt a transition to $(\theta,v^{(\mathbf{l})})$ for some
$\mathbf{l}\in\llbracket M\rrbracket^{T}$. We show in Appendix \ref{subsec: Delayed rejection step for MHAAR-RB}
that such a delayed rejection strategy is possible in general and
takes the particular form shown in Algorithm \ref{alg: MHAAR-RB},
that is no rejection occurs in this optional stage in the situation
where $\gamma_{t,\theta,\vartheta}=\gamma_{t,\theta}$. The computational
cost of these steps is $\mathcal{O}(MT)$. 

\begin{algorithm}[!h]
\caption{MHAAR-RB for the multiple latent variable model}
\label{alg: MHAAR-RB}

\KwIn{Current sample $(\theta,z)$}

\KwOut{New sample}

Sample $\vartheta\sim q(\theta,\cdot)$ and $c\sim\text{Unif}(\left\{ 1,2\right\} )$
\\
\If{$c=1$}{

Set $v^{(\mathbf{1})}=z_{1:T}$ and sample $v_{t}^{(i)}\overset{}{\sim}q_{t,\theta,\vartheta}(\cdot)$
for $i=2,\ldots,M$, $t=1,\ldots,T$.\\
Sample $\mathbf{k}\sim b_{\theta,\vartheta}^{\text{(1)}}(\mathbf{\cdot}|\mathbf{v})$
and set $z'=v^{(\mathbf{k})}$.\\
Return $(\vartheta,z')$ with probability $\min\left\{ 1,r_{\mathbf{1},\mathbf{v}}(\theta,\vartheta)\right\} $;
otherwise return $(\vartheta,z)$.

}\Else{

Set $v^{(\mathbf{1})}=z_{1:T}$ and sample $v_{t}^{(i)}\overset{}{\sim}q_{t,\vartheta,\theta}(\cdot)$
for $i=2,\ldots,M$, $t=1,\ldots,T$.\\
Sample $\mathbf{k}\sim b_{\theta,\vartheta}^{(2)}(\cdot|\mathbf{v})$
and set $z'=v^{(\mathbf{k})}.$\\
Return $(\vartheta,z')$ with probability $\min\left\{ 1,r_{\mathbf{k},\mathbf{v}}(\vartheta,\theta)^{-1}\right\} $;
otherwise return $(\theta,z)$.

}

\textbf{Optional refreshment of $z_{1:T}$}\\
\If{the move is rejected and $\gamma_{t,\theta,\vartheta}=\gamma_{t,\theta}$
for all $t=1,\ldots,T$,}{

Sample $\mathbf{l}\sim b_{\theta,\vartheta}^{\textup{ref},(c)}(\cdot|\mathbf{v})$,
and set $z_{1:T}=v^{(\mathbf{l})}$, where
\[
b_{\theta,\vartheta}^{\textup{ref},(1)}(\mathbf{l}|\mathbf{v})=\prod_{t=1}^{n}\tfrac{\gamma_{t,\theta}(v_{t}^{(l_{t})})/q_{t,\theta,\vartheta}(v_{t}^{(l_{t})})}{\sum_{i=1}^{M}\gamma_{t,\theta}(v_{t}^{(i)})/q_{t,\theta,\vartheta}(v_{t}^{(i)})},\quad b_{\theta,\vartheta}^{\textup{ref},(2)}(\mathbf{l}|\mathbf{v})=\prod_{t=1}^{n}\tfrac{\gamma_{t,\theta}(v_{t}^{(l_{t})})/q_{t,\vartheta,\theta}(v_{t}^{(l_{t})})}{\sum_{i=1}^{M}\gamma_{t,\theta}(v_{t}^{(i)})/q_{t,\vartheta,\theta}(v_{t}^{(i)})}.
\]

Return $(\theta,z_{1:T})$.

}
\end{algorithm}

\subsection{Examples}
\begin{example}[ABC learning of an $\alpha$-stable distribution]
\label{ex: intractable likelihood} Consider an intractable likelihood
function $\theta\mapsto\ell_{\theta}(y)$ for $y\in\mathsf{Y}$ such
that sampling from the corresponding data generating distribution
is tractable.  Approximate Bayesian computation (ABC) \citep{Pritchard_et_al_1999,Beaumont_et_al_2002,Marjoram_et_al_2003}
is a general methodology to address inference in such scenarios. For
$\epsilon>0$, let $g^{\epsilon}(z,y):=\kappa(y,z;\epsilon)\text{ for }\text{\ensuremath{y,z\in\mathsf{Y}}}$,
where $\kappa(\cdot;\cdot;\epsilon)$ is some kernel and $\epsilon>0$
a bandwidth parameter. An ABC-based approximation to the intractable
posterior $\pi(\theta)\propto\eta(\theta)\prod_{t=1}^{T}\ell_{\theta}(y_{t})$
is obtained by marginalisation of the joint density
\begin{equation}
\pi^{\epsilon}(\theta,z_{1:T}):=\eta(\theta)\prod_{t=1}^{T}\ell_{\theta}(z_{t})g^{\epsilon}(z_{t},y_{t}).\label{eq: ABC posterior}
\end{equation}
For illustration we consider the scenario where the observations are
assumed to arise from an $\alpha$-stable distribution $\mathcal{A}(\alpha,\beta,\mu,\sigma)$,
where $\alpha,\beta,\mu,\sigma\in\mathbb{R}_{+}$ are the shape, skewness,
location, and scale parameters, respectively. Here we take $g^{\epsilon}(z,y)=\kappa(\arctan(y);\arctan(z),\epsilon)$,
where $\kappa(\cdot,\cdot;\epsilon)$ is taken a Gaussian kernel,
as in \citet{Yildirim_et_al_2015}.

We generated a sequence of i.i.d.\ observations of length $T=100$
from $\mathcal{A}(1.8,0,0,2)$. Assuming $\beta$ is known, we consider
estimating $\theta=(\alpha,\mu,\sigma)$ using the ABC posterior distribution
with $\epsilon=0.1$. In order to illustrate the benefit of MHAAR-RB
we compare performance of the RB and non-RB versions of the algorithms
for two choices of $\gamma_{t,\theta,\vartheta}(z)$:
\begin{itemize}
\item $\gamma_{t,\theta,\vartheta}(z)=\ell_{\theta}(z)g^{\epsilon}(z,y_{t})$
and $Q_{t,\theta,\vartheta}(z)=\ell_{\theta}(z)$. We refer to this
version as MHAAR-RB-0 and the corresponding non-RB version is referred
to as MwG (since the algorithm then corresponds to alternating between
an update of $z$ conditional upon $\theta$ and $\theta$ conditional
upon $z$), 
\item $\gamma_{t,\theta,\vartheta}(z)=\ell_{(\theta+\vartheta)/2}(z)g^{\epsilon}(z,y_{t})$
and $Q_{t,\theta,\vartheta}(z)=\ell_{(\theta+\vartheta)/2}(z)$. We
will refer to this version as MHAAR-RB-1. When no Rao-Blackwellisation
is performed we refer to the algorithm as AIS MCMC.
\end{itemize}
Note that for a given $M$ the complexity of these algorithms is comparable.
The computational overhead arising from RB is limited since it consists
of applying simple operations such as additions and multiplications
to the most expensive quantities computed by all the algorithms. We
provide precise details concerning prior choices and proposal distributions
below and focus first on results. 

We ran the algorithms for $2\times10^{5}$ iterations for the values
$M=10,20,50,100$. In Table \ref{fig: ABC MCMC convergence vs time},
we report IAC and IAC $\times$ CPU time per iteration for the MHAAR-RB
algorithms as well as their non-RB counterparts. The difference between
the RB algorithms and their non-RB counterparts is striking: the former
seem to benefit highly from increasing $M$ in contrast with the latter.
MHAAR-RB-0 seems superior to MHAAR-RB-1, which is explained by the
fact the acceptance ratio of MHAAR-RB-0 in \eqref{eq: MHAAR-RB acceptance ratio with no annealing-iid model}
enjoys a full averaging and suffers less from the dependency on$z_{t}$
compared to the acceptance ratio of MHAAR-RB-1 in \eqref{eq: acceptance ratio-expanded-iid model}.
We further observe the following further benefit of better mixing:
for MHAAR-RB-0 the gain of using $M=100$ rather than $M=10$ replicas
is $1043/21\approx50$ while averaging the output from $10$ computers
running MHAAR-RB-0 for $M=10$ would have lead to a gain of $10$.
This advantage persists when (serial) CPU time is taken into account,
even though our implementation uses Matlab, for which for loops can
be particularly slow.

In Figure \ref{fig: ABC MCMC convergence vs time} we report ensemble
averages, over $1000$ independent runs, vs time, for the algorithms
compared in this example. One can observe the benefit of using averaging
with MHAAR-RB, especially with the one without annealing, MHAAR-RB-0,
as well as increasing $M$.

For all algorithms, $\vartheta$ is proposed using a random walk proposal
for all of its components, with standard deviation $0.2$ for each.
For simplicity, we take a flat prior for $\theta$. The first quarter
of the $2\times10^{5}$ iterations are discarded as burn-in time from
the calculations related to IAC time.

\begin{table}[!h]
\begin{centering}
{\small{}}%
\begin{tabular}{|>{\raggedright}m{0.25cm}|>{\centering}p{1.25cm}|>{\centering}p{1.25cm}|>{\centering}p{1.25cm}|>{\centering}p{1.25cm}|>{\centering}p{1.25cm}||>{\centering}p{1.25cm}|>{\centering}p{1.25cm}|>{\centering}p{1.25cm}|>{\centering}p{1.25cm}|}
\hline 
\multirow{2}{0.25cm}{{\small{}$\theta$}} & \multirow{2}{1.25cm}{{\small{}$M$}} & \multicolumn{4}{c||}{{\small{}IAC time}} & \multicolumn{4}{c|}{{\small{}IAC $\times$ CPU time per iteration}}\tabularnewline
\cline{3-10} \cline{4-10} \cline{5-10} \cline{6-10} \cline{7-10} \cline{8-10} \cline{9-10} \cline{10-10} 
 &  & {\footnotesize{}MHAAR-RB-0} & {\footnotesize{}MwG} & {\footnotesize{}MHAAR-RB-1} & {\footnotesize{}AIS MCMC} & {\footnotesize{}MHAAR-RB-0} & {\footnotesize{}MwG} & {\footnotesize{}MHAAR-RB-1} & {\footnotesize{}AIS MCMC}\tabularnewline
\hline 
\multirow{4}{0.25cm}{{\small{}$\alpha$}} & {\small{}10} & {\small{}1043} & {\small{}1679} & {\small{}1504} & {\small{}1550} & {\small{}0.227} & {\small{}0.473} & {\small{}0.570} & {\small{}0.372}\tabularnewline
 & {\small{}20} & {\small{}244} & {\small{}1582} & {\small{}883} & {\small{}1198} & {\small{}0.088} & {\small{}0.635} & {\small{}0.465} & {\small{}0.407}\tabularnewline
 & {\small{}50} & {\small{}46} & {\small{}1748} & {\small{}250} & {\small{}1127} & {\small{}0.039} & {\small{}1.340} & {\small{}0.241} & {\small{}0.686}\tabularnewline
 & {\small{}100} & {\small{}21} & {\small{}1103} & {\small{}251} & {\small{}757} & {\small{}0.030} & {\small{}1.469} & {\small{}0.395} & {\small{}0.795}\tabularnewline
\hline 
\multirow{4}{0.25cm}{{\small{}$\mu$}} & {\small{}10} & {\small{}15952} & {\small{}11399} & {\small{}7519} & {\small{}2620} & {\small{}3.467} & {\small{}3.211} & {\small{}2.850} & {\small{}5.602}\tabularnewline
 & {\small{}20} & {\small{}2909} & {\small{}7706} & {\small{}3575} & {\small{}9235} & {\small{}1.054} & {\small{}3.093} & {\small{}1.880} & {\small{}3.135}\tabularnewline
 & {\small{}50} & {\small{}440} & {\small{}8041} & {\small{}2093} & {\small{}8254} & {\small{}0.375} & {\small{}6.166} & {\small{}2.022} & {\small{}5.024}\tabularnewline
 & {\small{}100} & {\small{}115} & {\small{}17323} & {\small{}2236} & {\small{}2665} & {\small{}0.165} & {\small{}23.068} & {\small{}3.524} & {\small{}2.8}\tabularnewline
\hline 
\multirow{4}{0.25cm}{{\small{}$\sigma$}} & {\small{}10} & {\small{}876} & {\small{}1461} & {\small{}2485} & {\small{}23317} & {\small{}0.190} & {\small{}0.412} & {\small{}0.942} & {\small{}0.629}\tabularnewline
 & {\small{}20} & {\small{}265} & {\small{}1175} & {\small{}489} & {\small{}1355} & {\small{}0.096} & {\small{}0.472} & {\small{}0.257} & {\small{}0.46}\tabularnewline
 & {\small{}50} & {\small{}74} & {\small{}990} & {\small{}243} & {\small{}959} & {\small{}0.063} & {\small{}0.759} & {\small{}0.235} & {\small{}0.584}\tabularnewline
 & {\small{}100} & {\small{}40} & {\small{}1160} & {\small{}257} & {\small{}576} & {\small{}0.057} & {\small{}1.545} & {\small{}0.405} & {\small{}0.605}\tabularnewline
\hline 
\end{tabular}{\small\par}
\par\end{centering}
\caption{Comparison of algorithms in terms of IAC and IAC \texttimes{} CPU
time per iteration }

\label{tbl: alpha-stable model results}
\end{table}

\begin{figure}[!h]
\centerline{\includegraphics[scale=0.7]{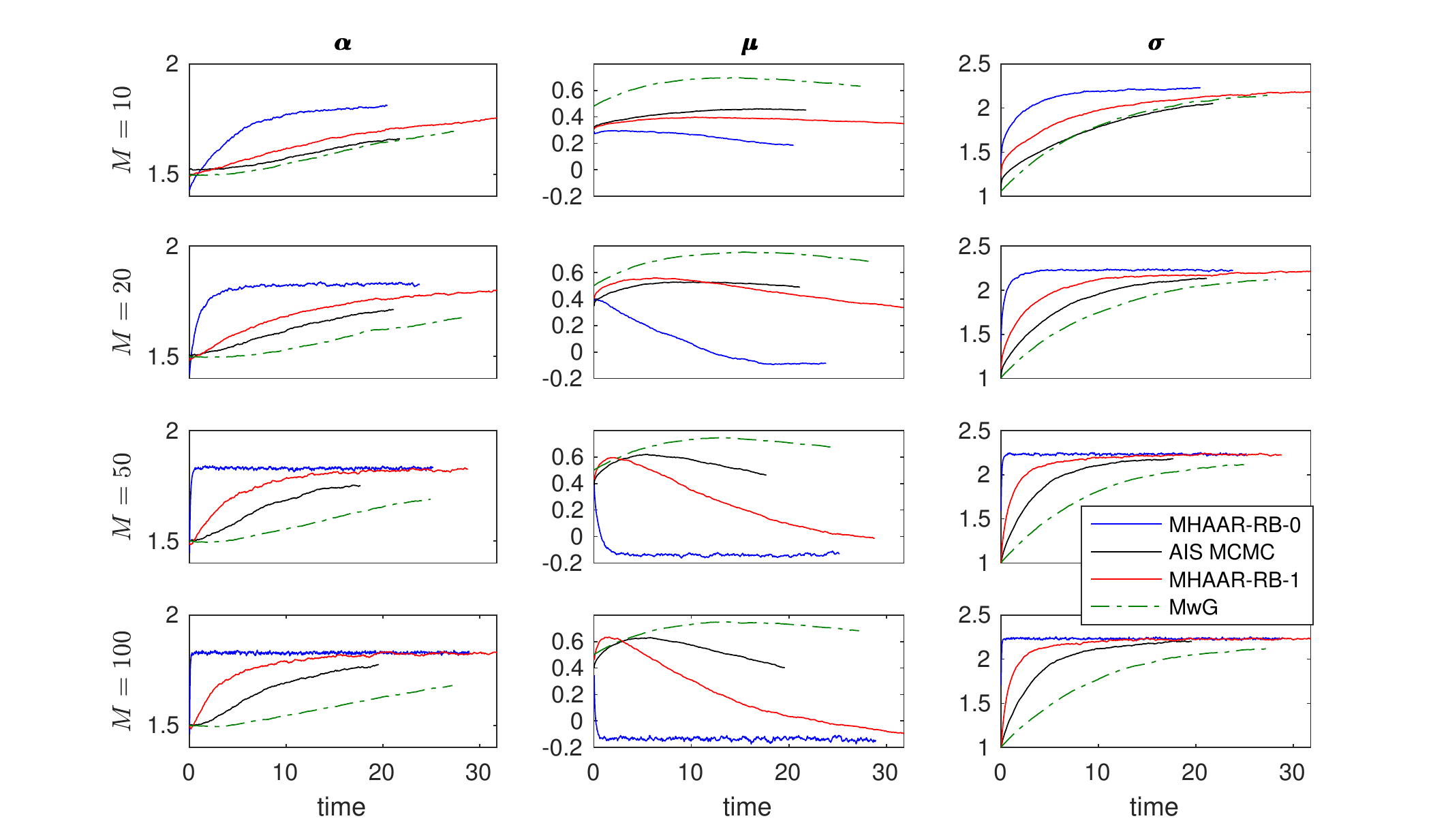}}

\caption{Ensemble averages vs time for for MHAAR-RB-0, AIS MCMC, MHAAR-RB-1,
MwG.}

\label{fig: ABC MCMC convergence vs time}
\end{figure}
\end{example}
\begin{example}[Gaussian process regression model]

The Gaussian process regression model is an example for a single latent
variable model, i.e., $T=1$. We observe pairs $(x_{i},y_{i})$ for
$i=1,\ldots,n$, where $x_{i}$ is a vector of $d$ covariates, and
\[
y_{i}=f(x_{i})+\varepsilon_{i},\quad\varepsilon_{i}\sim\mathcal{N}(0,\sigma^{2}),
\]
where $f$ is an unknown function with a Gaussian process prior with
zero mean and some covariance function $C$, yielding $(f(x_{1}),\ldots,f(x_{n}))\sim\mathcal{N}(0,C_{x_{1:n}})$.
One commonly used covariance function has the form 
\[
C_{x_{1:n}}(i,j)=\tau^{2}\Upsilon(i,j),\quad\Upsilon(i,j)=\left(\upsilon+\exp\left\{ -\sum_{k=1}^{d}\left[\theta_{i}(x_{i}(k)-x_{j}(k))\right]^{2}\right\} +\varsigma\delta_{i,j}\right)
\]
We assume $\upsilon$ and $\varsigma$ are fixed and known, and the
unknown variables $\theta$ and $z=(\tau,\sigma^{2})$ are \emph{a
priori} independent, having Gaussian prior distributions for their
logarithms. The log-likelihood of $y=y_{1:n}$ and $x=x_{1:n}$ given
$\theta$ is 
\[
\ell(\theta,z;x,y)=-0.5\left(\log\left\vert \Sigma\right\vert +y^{T}\Sigma^{-1}y\right).
\]
where $\Sigma=\tau^{2}\Upsilon+\sigma^{2}I_{n}$. Therefore we have
a joint distribution $\pi(\theta,z)$ over the unknown variables.

Following the terminology of \citet{Neal_2004,Neal_2010}, we call
a variable a slow (resp.\ fast) variable if the update of the posterior
density is hard (resp.\ easy) when the variable is changed with the
other parameters fixed. If eigenvalue decomposition is used for $\Sigma$,
then $\theta$ may be viewed as the slow variable, and $z=(\tau,\sigma^{2})$
as fast variables: Suppose $\Upsilon=E\Lambda E^{T}$ and $\hat{y}=E^{T}y$,
where $\Lambda$ has the eigenvalues $\lambda_{1},\ldots,\lambda_{n}$
on its diagonal. Since $\Upsilon e=\lambda e$ implies $\Sigma e=(\tau^{2}\lambda+\sigma^{2})e$,
we can write 
\[
\ell(y;\theta)=-0.5\sum_{i=1}^{n}\log(\tau^{2}\lambda_{i}+\sigma^{2})-0.5\sum_{i=1}^{n}\frac{\hat{y}_{i}^{2}}{\tau^{2}\lambda_{i}+\sigma^{2}}.
\]
Therefore, while $\hat{y}_{i}$'s have to be re-evaluated when $\theta$
changes, changing $z=(\tau,\sigma^{2})$ does not require re-evaluation
of $\hat{y}_{i}$'s, which is the most computationally demanding part
of the likelihood evaluation. The distinction of slow-fast variables
for this model gets clearer for larger $n$. The fact that $\theta$
is the slow variable and $z$ is the fast variable justifies the use
of MHAAR-RB, whose performance increases with the number of auxiliary
variables generated for the latent variable $z$. When $\vartheta$
is proposed, one needs to perform a single eigenvalue decomposition,
which is expensive, which is followed by sampling $M$ auxiliary variables
and calculating quantities depending on them, which is relatively
cheap even for large values of $M$.

We have compared AIS MCMC with MHAAR-RB for the Gaussian process regression
model on the same data set used in \citet{Neal_2010}, with $p=12$
covariates and $n=100$ points. (The software in \texttt{https://www.cs.toronto.edu/\textasciitilde radford/ensmcmc.software.html}
can be used to generate the data.) We ran MHAAR-RB with $\gamma_{\theta,\vartheta}=\gamma_{\theta}$
with values $M=10$, $50$, $100$ and AIS MCMC with $L=10$, $50$,
$100$ annealing steps with a geometric annealing schedule. All algorithms
are started from the same initial point and run for $10^{6}$ iterations.
To demonstrate performance, we report IAC and IAC $\times$ CPU times
for the average $\frac{1}{d}\sum_{i=1}^{d}\theta_{i}$ in Table \ref{tbl: IAC Gaussian regression process}.
Although the difference between the performances of the two algorithms
is not spectacular in terms of IAC, the MHAAR-RB algorithm benefits
from its simplicity in generating the auxiliary variables, hence beating
AIS MCMC significantly in terms of IAC $\times$ CPU time. The table
also shows the poorer performance of a MwG algorithm, where both slow
and fast variables are updated in an alternating fashion by MH moves
with random walk proposals. This indicates the usefulness of algorithms
such as AIS MCMC and MHAAR-RB that exploit the existence of slow vs
fast variables.

\begin{table}[h]
\caption{IAC and IAC $\times$ CPU times}
\label{tbl: IAC Gaussian regression process}
\centerline{
\begin{tabular}{| c | c | c | c | c | c | c |}
\hline
\multirow{2}{*}{$M$ or $L$} & \multicolumn{3}{|c|}{IAC time ($\times 10^{3}$)} & \multicolumn{3}{c|}{IAC $\times$ CPU per iteration} \\
\cline{2-7}
 &  MHAAR-RB  & MwG & AIS MCMC & MHAAR-RB & MwG & AIS MCMC  \\
\hline
10 & 2.33 & 5.12 & 2.24 &15.98 & 32.32  & 15.89\\
50 &  1.06 & 4.21  & 1.37&  8.54& 30.6965 & 16.09\\
100 &  1.05 & 4.30 &1.22 & 9.74 & 33.6285  &18.04 \\
\hline
\end{tabular}
}
\end{table}

We also report some results pertaining the converge of the algorithms
for this example. Figure \ref{fig: GPR ensemble averages two parameters}
shows ensemble averages of the compared algorithms, out of $1000$
runs starting from the same initial point, versus both iteration (left)
and time (right). We ran the algorithms with all the parameter choices
appearing in Table \ref{tbl: IAC Gaussian regression process}. The
parameter choices appearing in the figure, $M=50$ for MHAAR-RB, $M=10$
for MwG and $L=10$ for AIS MCMC, correspond to the best choices in
terms of convergence vs time. The figure justifies the use of both
annealing (via AIS MCMC ) and averaging (via MHAAR-RB), especially
the latter proves more useful owing to the relative ease of implementing
the averaging compared to annealing. Also, we provide the averages
for the two parameters where the difference is most visible; for the
other parameters the algorithms showed similar performance.

\begin{figure}
\centerline{\includegraphics[scale=0.55]{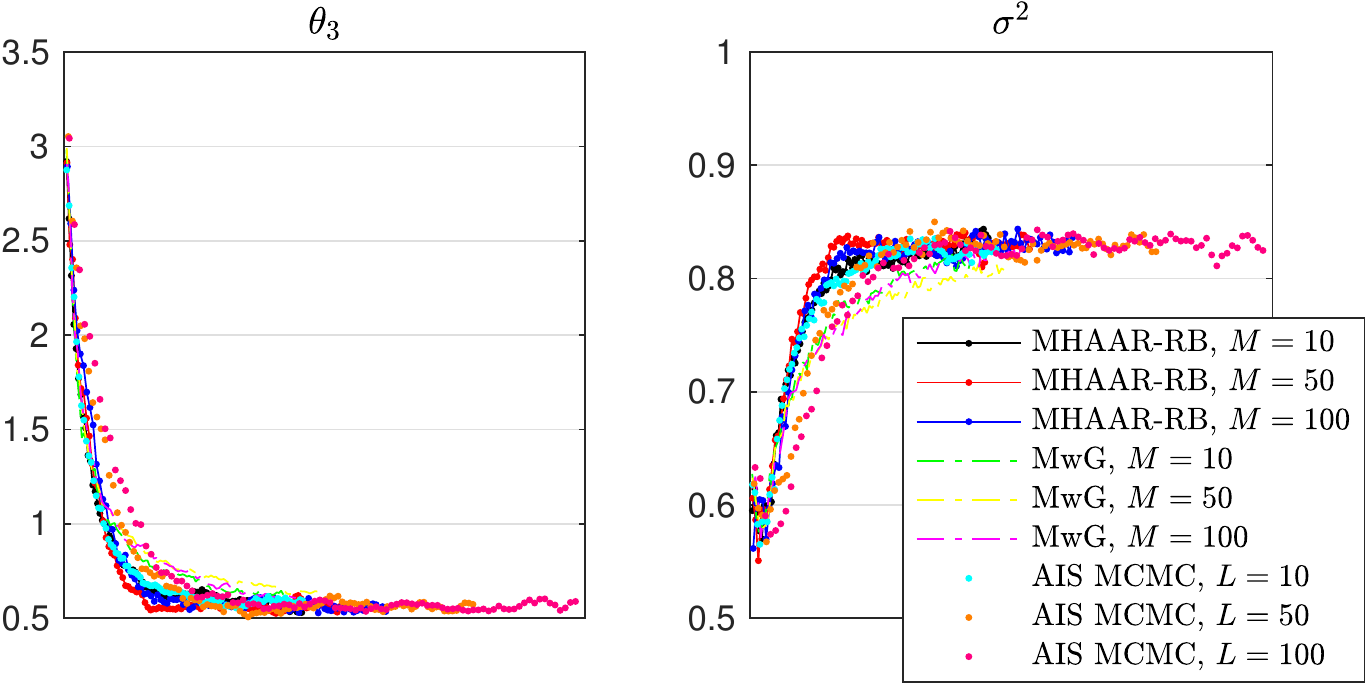}\hspace{1cm}\includegraphics[scale=0.55]{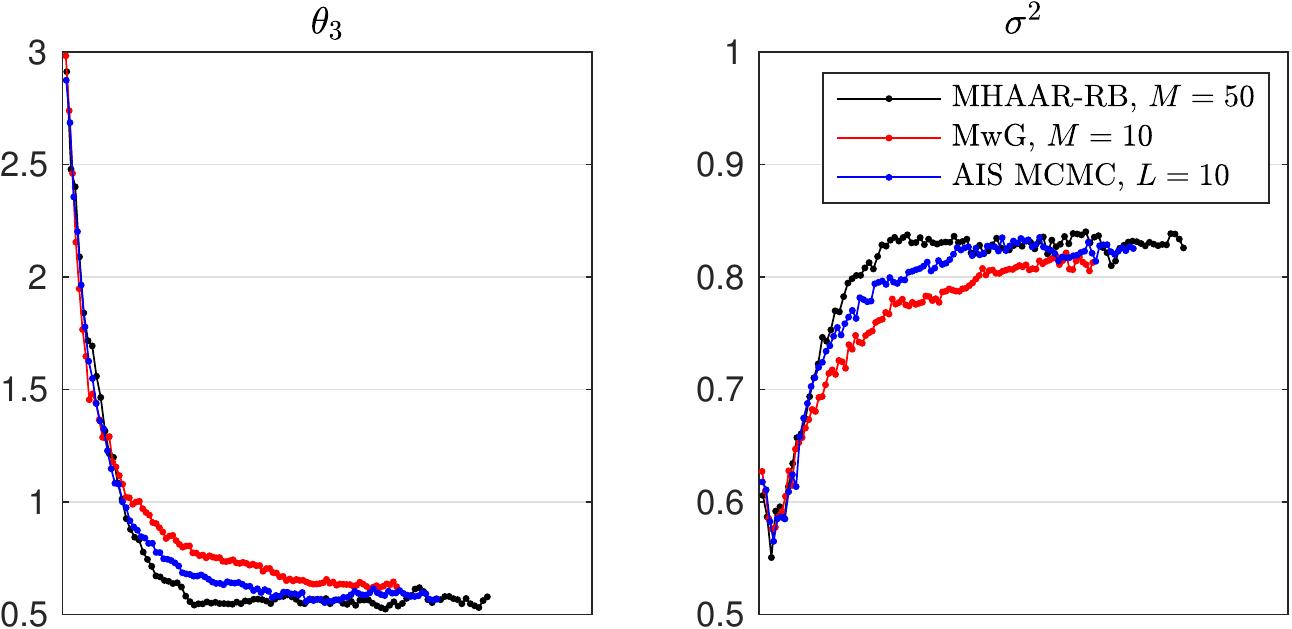}}\caption{Ensemble averages vs time of MHAAR-RB, MwG, and AIS MCMC. Left: All
the settings, Right: Best settings}

\label{fig: GPR ensemble averages two parameters}
\end{figure}

The other details of our experiment are as follows. The model parameters
are selected in parallel with \citet{Neal_2010}: We take $\varsigma=0.01$,
$\upsilon=1$, and the prior distribution the vector $\log\theta$
is taken a normal distribution with mean $\log0.5$ and unit variance
for each component, with a correlation of $0.69$ for any pair of
components. The other parameters $\tau,\sigma^{2}$ are apriori independent
from $\theta$ and among themselves, with $\log\tau\sim\mathcal{N}(0,2.25)$
and $\log\sigma\sim\mathcal{N}(\log0.5,2.25)$. AIS MCMC and MHAAR-RB
attempt to update one component of $\theta$ at a time with the same
proposal mechanism. For each component, a normal random walk proposal
is used for $\log\theta_{i}$ with mean $0$ and standard deviation
$2$. At the intermediate steps of AIS MCMC, the fast variables are
updated with an MH kernel with random walk proposals on $\log\tau$
and $\log\sigma$ with zero mean and standard deviations $0.6$ for
both. We run MHAAR-RB with no annealing, i.e., $\gamma_{\theta,\vartheta}=\gamma_{\theta}$,
ending up with the acceptance ratio in \eqref{eq: MHAAR-RB acceptance ratio with no annealing-iid model}
with $T=1$ (Superiority of no annealing in general is shown in the
previous example). Moreover, $q_{\theta,\vartheta}(z)$ is taken as
density of the prior distribution of $z$, therefore, we have $Q_{1}^{M}=Q_{2}^{M}$.
\end{example}

\section{State-space models: SMC and cSMC within MHAAR\label{sec: State-space models: SMC and conditional SMC within MHAAR}}

In Sections \ref{sec: Pseudo-marginal ratio algorithms using averaged acceptance ratio estimators}
and \ref{sec: MHAAR via Rao-Blackewellisation of PMR}, we have shown
how two different generic MHAAR strategies which consist of averaging
estimates of the acceptance ratio could be helpful. Here we extend
the methodology in Section \ref{sec: MHAAR via Rao-Blackewellisation of PMR}
to state-space models. Specifically we present methods where dependent
acceptance ratios arising from a single conditional SMC algorithm
can be averaged in order to improve performance.

\subsection{State-space models and cSMC \label{subsec: State-space models and conditional SMC} }

In its simplest form, a state-space model (SSM) is comprised of a
latent Markov chain $\{Z_{t};t\geq1\}$ taking its values in some
measurable space $(\mathsf{Z},\mathcal{Z})$ and observations $\{Y_{t};t\geq1\}$
taking values in $(\mathsf{Y},\mathcal{Y})$. The latent process has
initial probability of density $f_{\theta}(z_{1})$ and transition
density $f_{\theta}(z_{t-1},z_{t})$, dependent on a parameter $\theta\in\Theta\subset\mathbb{R}^{d_{\theta}}$.
An observation at time $t\geq1$ is assumed conditionally independent
of all other random variables given $Z_{t}=z_{t}$ and its conditional
observation density is $g_{\theta}(z_{t},y_{t})$. The corresponding
joint density of the latent and observed variables up to time $T\geq1$
is 
\begin{equation}
p_{\theta}(z_{1:T},y_{1:T})=f_{\theta}(z_{1})\prod_{t=2}^{T}f_{\theta}(z_{t-1},z_{t})\prod_{t=1}^{T}g_{\theta}(z_{t},y_{t}).\label{eq: HMM joint density}
\end{equation}
The densities $f_{\theta}$ and $g_{\theta}$ could also depend on
$t$, at the expense of notational complications. In order to alleviate
notation and ensure consistency we let $z:=z_{1:T}$ and $y:=y_{1:T}$.
The likelihood function associated to the observations $y$ can be
obtained 
\begin{equation}
\ell_{\theta}(y):=\int_{\mathsf{Z}^{T}}p_{\theta}(z,y){\rm d}z.\label{eq:likelihoodSSM}
\end{equation}
With a prior $\eta(\mathrm{d}\theta)$ on $\theta$ with density $\eta(\theta)$,
the joint posterior $\pi(\mathrm{d}(\theta,z))$ has the density 
\[
\pi(\theta,z)\propto\eta(\theta)p_{\theta}(z,y)
\]
so that $\pi(\theta)\propto\eta(\theta)\ell_{\theta}(y)$ and $\pi_{\theta}(z):=p_{\theta}(z\mid y)=p_{\theta}(z,y)/\ell_{\theta}(y)$.
Therefore, the acceptance ratio of the marginal MCMC algorithm for
SSM can be written as
\begin{equation}
r(\theta,\vartheta)=\frac{q(\vartheta,\theta)}{q(\theta,\vartheta)}\frac{\eta(\vartheta)\ell_{\vartheta}(y)}{\eta(\theta)\ell_{\theta}(y)}.\label{eq: marginal MCMC acceptance ratio - SSM}
\end{equation}

Conditional sequential Monte Carlo (cSMC) introduced in \citet{Andrieu_et_al_2010}
is an MCMC transition kernel akin to particle filters that is particularly
well suited to sampling from $\pi_{\theta}(\mathrm{d}z)$. It was
shown in \citet{Lindsten_and_Schon_2012} that cSMC with backward
sampling \citep{whiteley2010discussion} can be used efficiently as
part of a more elaborate Metropolis-within-Particle Gibbs algorithm
in order to sample from the posterior distribution $\pi(\mathrm{d}(\theta,z))$;
see Algorithm \ref{alg: Metropolis within particle Gibbs}. In \citet{gunawan2020scalable}
it is shown how this can be combined with ideas of \citet{Deligiannidis_et_al_2018}
to improve performance in specific scenarios.

\begin{algorithm}[!h]
\caption{Metropolis-within-particle Gibbs}
\label{alg: Metropolis within particle Gibbs} 

\KwIn{Current sample $(\theta,z)$} 

\KwOut{New sample} 

Sample $\vartheta\sim q(\theta,\cdot)$ and $z'\sim\mathrm{cSMC}\big(M,\theta,z\big)$.
\\
Return $(\vartheta,z')$ with probability 
\begin{equation}
\min\left\{ 1,\frac{q(\vartheta,\theta)\eta(\vartheta)p_{\vartheta}(z',y)}{q(\theta,\vartheta)\eta(\theta)p_{\theta}(z',y)}\right\} ;\label{eq: noisy acceptance ratio}
\end{equation}
otherwise return $(\theta,z')$.
\end{algorithm}

The cSMC algorithm with backward sampling for state-space models used
to present our results is given in Algorithm \ref{alg: Conditional SMC}
in Appendix~\ref{sec: Auxiliary results and proofs for cSMC based algorithms }.
To simplify exposition we consider the bootstrap particle filter where
the particles are initialised according to $f_{\theta}(z_{1})$ and
propagated according to the state transition $f_{\theta}(z_{t-1},z_{t})$;
our results can be extended straightforwardly to other choices. The
cSMC produces $T\times M$ samples from which $M^{T}$ paths can be
sampled using the backward recursion of \citep{whiteley2010discussion}.
The cSMC returns only one such path when used in Algorithm~\ref{alg: Metropolis within particle Gibbs},
which may seem to be wasteful. A natural idea is to make use of multiple\textendash or
even all $M^{T}$ possible\textendash trajectories and average out
the corresponding acceptance ratios \eqref{eq: noisy acceptance ratio}
before accepting or rejecting. We show that this is indeed possible
theoretically with Algorithms~\ref{alg: MHAAR-RB for SSM} and \ref{alg: MHAAR-S for SSM}
in the next section. We then show that these schemes are highly advantageous
on parallel computing architecture, but also on serial machines in
some difficult scenarios. The justification of the algorithms is postponed
to Appendix~\ref{sec: Auxiliary results and proofs for cSMC based algorithms };
while this can be thought of as extensions of the results of Section
\ref{sec: MHAAR via Rao-Blackewellisation of PMR} the dependence
structure implied by the cSMC leads to significant conceptual and
notational complications.

\subsection{MHAAR with cSMC for state-space models \label{subsec: MHAAR with cSMC for SSM}}

We will first present an unbiased estimator of the marginal acceptance
ratio in \eqref{eq: marginal MCMC acceptance ratio - SSM} for SSM
using particles produced by a cSMC iteration. Building on this we
present our MHAAR algorithm for SSM.

\subsubsection{Unbiased estimator of the acceptance ratio using particles of cSMC}

The particles $\mathbf{v}=v_{1:T}^{(1:M)}$ outputted by the cSMC
update can be partitioned as $\mathbf{v}=(z,\mathfrak{u})$, where
$z:=v^{(\mathbf{1})}$ is the path conditional upon which the cSMC
is run, and $\mathfrak{u}:=v^{(\bar{\mathbf{1}})}$ consists of the
rest of the variables in $\mathbf{v}$. It can be shown that the conditional
distribution of $\mathfrak{u}$ given $(\theta,z)\in\Theta\times\mathsf{Z}$
is given by
\[
\Phi_{\theta}(z,\mathrm{d}\mathfrak{u})=\prod_{i=2}^{M}f_{\theta}({\rm d}v_{1}^{(i)})\prod_{t=2}^{T}\left\{ \prod_{i=2}^{M}\frac{\sum_{j=1}^{M}w_{t-1,\theta}(v_{t-1}^{(j)})f_{\theta}(v_{t-1}^{(j)},{\rm d}v_{t}^{(i)})}{\sum_{j=1}^{M}w_{t-1,\theta}(v_{t-1}^{(j)})}\right\} .
\]
The law of the indices $\mathbf{k}:=(k_{1},\ldots,k_{T})$ drawn in
the backward sampling step in Algorithm \ref{alg: Conditional SMC}
(lines \ref{line:beginBS}-\ref{line:BSend}) conditional upon $\theta$
and $\mathbf{v}$ is given by
\[
b_{\theta}(\mathbf{k}|\mathbf{v}):=\frac{w_{T,\theta}(v_{T}^{(k_{T})})}{\sum_{i=1}^{M}w_{T,\theta}(v_{T}^{(i)})}\prod_{t=1}^{T-1}\frac{w_{t,\theta}(v_{t}^{(k_{t})})f_{\theta}(v_{t}^{(k_{t})},v_{t+1}^{(k_{t+1})})}{\sum_{i=1}^{M}w_{t,\theta}(v_{t}^{(i)})f_{\theta}(v_{t}^{(i)},v_{t+1}^{(k_{t+1})})}.
\]
Further, for any $\theta,\vartheta,\zeta\in\Theta$, and $z,z'\in\mathsf{Z}^{T}$,
define 
\begin{equation}
r_{z,z'}(\theta,\vartheta;\zeta)=\frac{q(\vartheta,\theta)\eta(\vartheta)p_{\vartheta}(z',y)p_{\zeta}(z,y)}{q(\theta,\vartheta)\eta(\theta)p_{\zeta}(z',y)p_{\theta}(z,y)}.\label{eq: AIS acceptance ratio for SSM}
\end{equation}
In the following, we show that it is possible to construct unbiased
estimators of $r(\theta,\vartheta)$ in \eqref{eq: marginal MCMC acceptance ratio - SSM}
using cSMC, provided we have a random sample $z\sim\pi_{\theta}(\cdot)$.
Specifically, this is obtained as the expected value of $r_{z,v^{(\mathbf{k})}}(\theta,\vartheta;\zeta)$
with respect to the backward sampling distribution of $\mathbf{k}$,
$b_{\theta}(\mathbf{k}|\mathbf{v})$.
\begin{thm}
\label{thm: SMC unbiased estimator of acceptance ratio} For $\theta,\vartheta,\zeta\in\Theta$
and any $M\geq1$, let $z\sim\pi_{\theta}(\cdot)$, $\mathbf{v}|z\sim\mathrm{cSMC}(M,\zeta,z)$
be the generated particles from the cSMC algorithm targeting $\pi_{\zeta}(\cdot)$
with $M$ particles, conditioned on $z$. Then, \textup{$r_{\mathbf{1},\mathbf{v}}(\theta,\vartheta;\zeta)$}
is an unbiased estimator of $r(\theta,\vartheta)$ in \eqref{eq: marginal MCMC acceptance ratio - SSM},
where for $\mathbf{v}\in\mathsf{Z}^{TM}$, $\mathbf{l}\in\left\llbracket M\right\rrbracket ^{T}$,
\textup{$r_{\mathbf{l},\mathbf{v}}(\theta,\vartheta;\zeta)$ is defined
as}
\begin{equation}
r_{\mathbf{l},\mathbf{v}}(\theta,\vartheta;\zeta):=\sum_{\mathbf{k}\in\left\llbracket M\right\rrbracket ^{T}}r_{v^{(\mathbf{l})},v^{(\mathbf{k})}}(\theta,\vartheta;\zeta)b_{\zeta}(\mathbf{k}|\mathbf{v}).\label{eq: SMC acceptance ratio estimator all paths}
\end{equation}
\end{thm}
The proof of Theorem \ref{thm: SMC unbiased estimator of acceptance ratio}
is left to Appendix \ref{subsec: Proof of unbiasedness for the Rao-Blackwellised estimator}.
Theorem \ref{thm: SMC unbiased estimator of acceptance ratio} is
original to the best of our knowledge and we find it interesting in
several aspects. Firstly, unlike the estimator in Metropolis-within-Particle
Gibbs (Algorithm \ref{alg: Metropolis within particle Gibbs}), the
estimator in Theorem \ref{thm: SMC unbiased estimator of acceptance ratio}
uses all possible paths from the particles generated by the cSMC.
Also, with a slight modification one can similarly obtain unbiased
estimators for $\pi(\vartheta)/\pi(\theta)$ which is of primary interest
in some applications. The theorem is derived from \citet[Theorem 5.2]{del2010backward}
and the results in \citet{Andrieu_et_al_2010} relating the laws of
cSMC and SMC. 

\subsubsection{MHAAR-RB for SSM}

Theorem \ref{thm: SMC unbiased estimator of acceptance ratio} motivates
the design of a MHAAR algorithm using the unbiased estimator \eqref{eq: SMC acceptance ratio estimator all paths}
as its acceptance ratios. We describe the algorithm, MHAAR-RB for
SSM, in detail below. The procedure requires a pair of functions $\zeta_{1}:\Theta^{2}\rightarrow\Theta$
and $\zeta_{2}:\Theta^{2}\rightarrow\Theta$ satisfying $\zeta_{1}(\theta,\vartheta)=\zeta_{2}(\vartheta,\theta)$
for $\theta,\vartheta\in\Theta$, in order to determine the intermediate
parameter value for which the cSMC is run. MHAAR-RB for SSM targets
the joint distribution for the variable $\xi=(\theta,\vartheta,\mathbf{v},\mathbf{k},c)\in\Theta^{2}\times\mathsf{Z}^{MT}\times\left\llbracket M\right\rrbracket {}^{T}\times\{1,2\}$
defined as 
\begin{equation}
\mathring{\pi}(\mathrm{d}(\theta,\vartheta,\mathbf{v},\mathbf{k},c))=\frac{1}{2}\pi(\mathrm{d}(\theta,z))Q_{c}^{M}(\theta,z;\mathrm{d}(\mathfrak{u},\mathbf{k})).\label{eq: joint distribution for MHAAR-RB-SSM}
\end{equation}
where we have used $z:=v^{(\mathbf{1})}$. Clearly, the marginal distribution
for $(\theta,z)$ is $\pi(\mathrm{d}(\theta,z))$, as desired. The
proposal mechanisms are
\begin{align*}
Q_{c}^{M}(\theta,z;\mathrm{d}(\mathfrak{u},\mathbf{k}))= & q(\theta,\mathrm{d}\vartheta)\Phi_{\zeta_{c}(\theta,\vartheta)}(z,\mathrm{d}\mathfrak{u})b_{\theta,\vartheta}^{(c)}(\mathbf{k}|\mathbf{v}),\quad c=1,2,
\end{align*}
where the sampling probabilities are given as $b_{\theta,\vartheta}^{(2)}(\mathbf{k}|\mathbf{v})=b_{\zeta_{2}(\theta,\vartheta)}(\mathbf{k}|\mathbf{v})$
and 
\[
b_{\theta,\vartheta}^{(1)}(\mathbf{k}|\mathbf{v})=\text{\ensuremath{\frac{r_{v^{(\mathbf{1})},v^{(\mathbf{k})}}(\theta,\vartheta;\zeta_{1}(\theta,\vartheta))b_{\zeta_{1}(\theta,\vartheta)}(\mathbf{k}|\mathbf{v})}{r_{\mathbf{1},\mathbf{v}}(\theta,\vartheta;\zeta_{1}(\theta,\vartheta))}},}
\]
which is obtained by weighting the backward sampling probabilities
of the cSMC by the acceptance ratios they correspond to, yielding
the normalising constant $r_{\mathbf{1},\mathbf{v}}(\theta,\vartheta;\zeta_{1}(\theta,\vartheta))$
defined in \eqref{eq: SMC acceptance ratio estimator all paths}.
One iteration of MHAAR-RB for SSM consists of the following main steps:
\begin{enumerate}
\item Sample $c\sim\text{Unif}(\left\{ 1,2\right\} )$, then sample $(\vartheta,\mathfrak{u},\mathbf{k})\sim Q_{c}^{N}(\theta,z;\cdot)$,
and form $\xi=(\theta,\vartheta,\mathbf{v},\mathbf{k},c)$.
\item Propose an MH update of $\xi$ via the involution 
\begin{equation}
\xi'=\varphi(\theta,\vartheta,\mathbf{v},\mathbf{k},c):=(\vartheta,\theta,\mathfrak{s}_{\mathbf{1},\mathbf{k}}(\mathbf{v}),\mathbf{k},3-c),\label{eq: involution for MHAAR-RB-SSM}
\end{equation}
where $\mathfrak{s}_{\mathbf{\mathbf{1}},\mathbf{k}}(\mathbf{v})$
is an operator on $\mathbf{v}$ that swaps $v^{(\mathbf{1})}$ and
$v^{(\mathbf{k})}$. 
\item Accept $\xi'$ with acceptance probability $\min\left\{ 1,\mathring{r}(\xi)\right\} $,
otherwise reject and keep $\xi$.
\end{enumerate}
We prove in Appendix \ref{subsec: Acceptance ratio of MHAAR-RB-SSM}
that this proposed involution leads to the averaged acceptance ratio
$r_{\mathbf{1},\mathbf{v}}(\theta,\vartheta;\zeta_{1}(\theta,\vartheta))$
in its acceptance probabilities, as stated in the theorem below.
\begin{thm}
\label{thm: acceptance ratio of MHAAR-RB-SSM}With the joint distribution
$\mathring{\pi}$ defined in \eqref{eq: joint distribution for MHAAR-RB-SSM},
the acceptance ratio for the proposed involution defined in \eqref{eq: involution for MHAAR-RB-SSM}
is given by 
\[
\mathring{r}(\xi):=\begin{cases}
r_{\mathbf{1},\mathbf{v}}(\theta,\vartheta;\zeta_{1}(\theta,\vartheta)), & c=1,\\
1/r_{\mathbf{k},\mathbf{v}}(\vartheta,\theta;\zeta_{2}(\theta,\vartheta)), & c=2.
\end{cases}
\]
\end{thm}
The proof has two interesting by-products: (i) An alternative proof
of Theorem \ref{thm: SMC unbiased estimator of acceptance ratio},
and (ii) another unbiased estimator of $r(\theta,\vartheta)$ which
uses all $M^{T}$ possible paths formed from the particles generated
by the cSMC, which is we state precisely in the following corollary.
\begin{cor}
\label{cor: SMC unbiased estimator of acceptance ratio} For $\theta,\vartheta,\zeta\in\Theta$
and any $M\geq1$, let $z\sim\pi_{\theta}(\cdot)$, $\mathbf{v}|z\sim{\rm cSMC}(M,\zeta,z)$
be the generated particles from the cSMC algorithm with $M$ particles
conditional on $\zeta,z$ and $\mathbf{k}|\mathbf{v}\sim b_{\zeta}(\cdot|\mathbf{v})$.
Then, $1/r_{\mathbf{k},\mathbf{v}}(\vartheta,\theta;\zeta)$ is an
unbiased estimator of $r(\theta,\vartheta).$
\end{cor}
We present MHAAR-RB for SSM in Algorithm \ref{alg: MHAAR-RB for SSM}.
 The per iteration computational complexity of Algorithm~\ref{alg: MHAAR-RB for SSM}
is $\mathcal{O}(M^{2}T)$. This follows upon observing that the unnormalised
probability in \eqref{eq: SMC acceptance ratio estimator all paths}
can be written as
\[
r_{v^{(\mathbf{1})},v^{(\mathbf{k})}}(\theta,\vartheta;\zeta)b_{\zeta}(\mathbf{k}|\mathbf{v})=:\varrho_{\mathbf{1},\mathbf{v}}(\mathbf{k})=\varrho_{\mathbf{1},\mathbf{v},1}(k_{1})\prod_{t=2}^{T}\varrho_{\mathbf{1},\mathbf{v},t}(k_{t-1},k_{t})
\]
for an appropriate choice of the functions $\varrho_{\mathbf{1},\mathbf{v},t}$
and that $r_{\mathbf{1},\mathbf{v}}(\theta,\vartheta;\zeta)=\sum_{\mathbf{k}\in\llbracket M\rrbracket^{T}}\varrho_{\mathbf{1},\mathbf{v}}(\mathbf{k})$
can be computed using a sum-product algorithm, while sampling $\mathbf{k}$
with probability proportional to $\varrho_{\mathbf{1},\mathbf{v}}(\mathbf{k})$,
required when $c=1$, can be performed with a forward-filtering backward-sampling
algorithm \citep{Zucchini_et_al_2016}. We note that: (a) while complexity
is $\mathcal{O}(M^{2}T)$ the operations involved are often much cheaper
than for the ${\rm cSMC}$ since, for example, likelihood terms involved
need not re-evaluation, (b) recent work investigates the implementation
of such recursions on GPUs e.g. \citet{8638034}, although this is
far beyond the scope of the present methodological paper.

\begin{algorithm}[!h]
\caption{MHAAR-RB for SSM}
\label{alg: MHAAR-RB for SSM}

\KwIn{Current sample $(\theta,z)$}

\KwOut{New sample} 

Sample $\vartheta\sim q(\theta,\cdot)$ and $c\sim\text{Unif}(\left\{ 1,2\right\} )$,
and set $\zeta=\zeta_{c}(\theta,\vartheta)$.\\
\If {$c=1$ }{ 

Run a cSMC$(M,\zeta,z)$ targeting $\pi_{\zeta}$ conditional on $z$
to obtain $\mathbf{v}$.\\
Sample $\mathbf{k}\sim b_{\theta,\vartheta}^{(1)}(\mathbf{\cdot}|\mathbf{v})$
and set $z'=v^{(\mathbf{k})}$\\
Return $(\vartheta,z')$ with probability $\min\{1,r_{\mathbf{1},\mathbf{v}}(\theta,\vartheta;\zeta)\}$;
otherwise return $(\theta,z)$.

}

\Else{ 

Run a ${\rm cSMC}(M,\zeta,z)$ targeting $\pi_{\zeta}$ conditional
upon $z$ to obtain $\mathbf{v}$.\\
Sample $\mathbf{k}\sim b_{\theta,\vartheta}^{(2)}(\cdot|\mathbf{v})$
and set $z'=v^{(\mathbf{k})}$.\\
Return $(\vartheta,z')$ with probability $\min\{1,1/r_{\mathbf{k},\mathbf{v}}(\vartheta,\theta;\zeta)\}$;
otherwise return $(\vartheta,z')$.

}

\textbf{Optional refreshment of $z$}\\
\If{the move is rejected, $c=1$, and $\zeta=\theta$}{

Sample $\mathbf{l}\sim b_{\theta}(\cdot|\mathbf{v})$ and return $(\theta,v^{(\mathbf{l})})$.

}
\end{algorithm}

\paragraph{Refreshing $z$ via delayed rejection:}

In Section \ref{sec: MHAAR via Rao-Blackewellisation of PMR}, in
the particular scenario where the latent variable sequence consists
of ${\rm iid}$ states, we have already discussed how a delayed rejection
step can be included to refresh the variable $z$ upon a `stage 1'
rejection, at an minimal computational cost. Delayed rejection is
also possible for SSM and is particularly attractive when $c=1$ and
$\zeta_{1}(\theta,\vartheta)=\theta$. In this case $z$ can be refreshed
upon rejection by simply performing another backward sampling iteration
on the already sampled particles $\mathbf{v}$. Otherwise a second
accept/reject step is required. The proof of validity for all scenarios
is left to Appendix \ref{sec: Delayed rejection step for MHAAR-RB-SSM}.
The delayed rejection step is included in Algorithm~\ref{alg: MHAAR-RB for SSM}
as an `optional' step and its cost is $\mathcal{O}(MT)$.
\begin{example}
\label{ex: GLSSM example} We consider the following linear Gaussian
SSM 
\begin{align*}
Z_{t} & =\phi(Z_{t-1}-(1-a)\theta)+(1-a)\theta+V_{t},\quad t\geq2\\
Y_{t} & =Z_{t}+a\theta+W_{t},\quad t\geq1.
\end{align*}
where $\phi>0$ is a coefficient, $a\in[0,1]$, $Z_{1}\sim\mathcal{N}(0,\sigma_{z}^{2})$,
$V_{t}\overset{\mathrm{iid}}{\sim}\mathcal{N}(0,(1-\phi^{2})\sigma_{z}^{2})$,
and $W_{t}\overset{\mathrm{iid}}{\sim}\mathcal{N}(0,\sigma_{y}^{2})$.
Naturally a Kalman filter can be used here to compute the likelihood
function efficiently and no Monte Carlo methods are needed. However
this model offers a fully controllable testbed useful to illustrate
the type of situations where MHAAR is of interest. Importantly the
likelihood function $\theta\mapsto\ell_{\theta}(y,a)$ does not depend
on the choice of $a$ but, assuming a prior distribution on $\theta$,
the posterior dependency between $\theta$ and $Z_{1:T}$ does. As
a result the mixing properties of a Gibbs sampler sampling alternately
from $\pi(\theta|z_{1:T})$ and $\pi(z_{1:T}|\theta)$ are highly
dependent on the choice of $a$. For example for $\phi=0$, \citet{papaspiliopoulos_et_al_2003}
showed that for $\sigma_{z}^{2}/\sigma_{y}^{2}\ll1$ (resp. $\sigma_{z}^{2}/\sigma_{y}^{2}\gg1$)
the choice $a\approx1$ ($a\approx0$) leads to strong posterior dependence. 

We generated a dataset of size $T=100$ from this SSM with $\theta^{\ast}=1$,
$\phi=0.95$ and noise parameters $\sigma_{z}^{2}=1$ and $\sigma_{y}^{2}=0.1$,
the regime where $a\approx1$ leads to strong dependence, hence our
choice of $a=1$. We compared MHAAR-RB, MHAAR-RB-R (with refreshment
of $z$ upon rejection) for SSM as in Algorithm~\ref{alg: MHAAR-RB for SSM}
and MwPG in terms of IAC time and IAC $\times$ CPU time per iteration
for $\theta$ for different values of $M$. For MHAAR-RB and MHAAR-RB-R,
we used $\zeta_{1}(\theta,\vartheta)=\theta$. Each run is performed
for $10^{6}$ iterations, except that we run MwPG for $5\times10^{6}$
iterations to overcome the variability in the estimates for the IAC
times. The prior for $\theta$ is taken as $\mathcal{N}(0,10^{4})$.
For all the algorithms, a random walk proposal is used with a proposal
standard deviation of $\sigma_{q}=0.3$. The results are displayed
in Table~\ref{tbl: Lin-Gauss HMM IAC and IAC times CPU per iteration}.
We observe that MHAAR-RB and MHAAR-RB-R's response to increasing $M$
is substantial and should be contrasted with the standard MwPG's underwhelming
performance. Further we note the superiority of MHAAR-RB and MHAAR-RB-R
on MwPG even when the IAC time is rescaled with the computation time,
that is MHAAR-RB and MHAAR-RB-R outperform MwPG even on a serial machine
for this example. 

Figure~\ref{fig: ensemble averages vs time for MHAAR-RB Lin-Gauss HMM}
shows the ensemble averages over $100$ runs (see Example~\ref{ex:RJMCMC2})
for the posterior expectation of $\theta$ versus iteration number
and time for the three algorithms, illustrating burn-in length. The
results mirror those of Table~\ref{tbl: Lin-Gauss HMM IAC and IAC times CPU per iteration}
concerned with IAC times with MHAAR-RB and MHAAR-RB-R vastly superior
to MwPG in terms of burn in length, with much better reactivity to
increasing $M$. 

\begin{table}[!h]
\begin{singlespace}
\begin{centering}
{\small{}}%
\begin{tabular}[b]{|c|c|c|c|c|c|c|}
\hline 
\multirow{2}{*}{{\small{}$M$}} & \multicolumn{3}{c|}{{\small{}IAC time ($\times10^{3}$)}} & \multicolumn{3}{c}{{\small{}IAC $\times$ CPU time per iteration}}\tabularnewline
\cline{2-7} \cline{3-7} \cline{4-7} \cline{5-7} \cline{6-7} \cline{7-7} 
 & {\small{}MHAAR-RB} & {\small{}MHAAR-RB-R} & {\small{}MwPG} & {\small{}MHAAR-RB} & {\small{}MHAAR-RB-R} & {\small{}MwPG}\tabularnewline
\hline 
{\small{}5} & {\small{}4.1801} & {\small{}1.7666} & {\small{}4.2519} & {\small{}13.8732} & {\small{}5.7242} & {\small{}13.2968}\tabularnewline
{\small{}10} & {\small{}1.4971} & {\small{}1.1556} & {\small{}3.8536} & {\small{}5.8461} & {\small{}4.4266} & {\small{}12.6457}\tabularnewline
{\small{}20} & {\small{}0.4713} & {\small{}0.4332} & {\small{}3.5337} & {\small{}2.7598} & {\small{}2.5566} & {\small{}12.1545}\tabularnewline
{\small{}50} & {\small{}0.1579} & {\small{}0.1516} & {\small{}3.2501} & {\small{}1.7587} & {\small{}1.7935} & {\small{}14.0562}\tabularnewline
\hline 
\end{tabular}{\small\par}
\par\end{centering}
\end{singlespace}
\caption{Comparison of MHAAR-RB, MHAAR-RB-R, and MwPG in terms of IAC and IAC
\texttimes{} CPU time per iteration for $\theta$.}

\label{tbl: Lin-Gauss HMM IAC and IAC times CPU per iteration}
\end{table}

\begin{figure}[!h]
\begin{centering}
\includegraphics{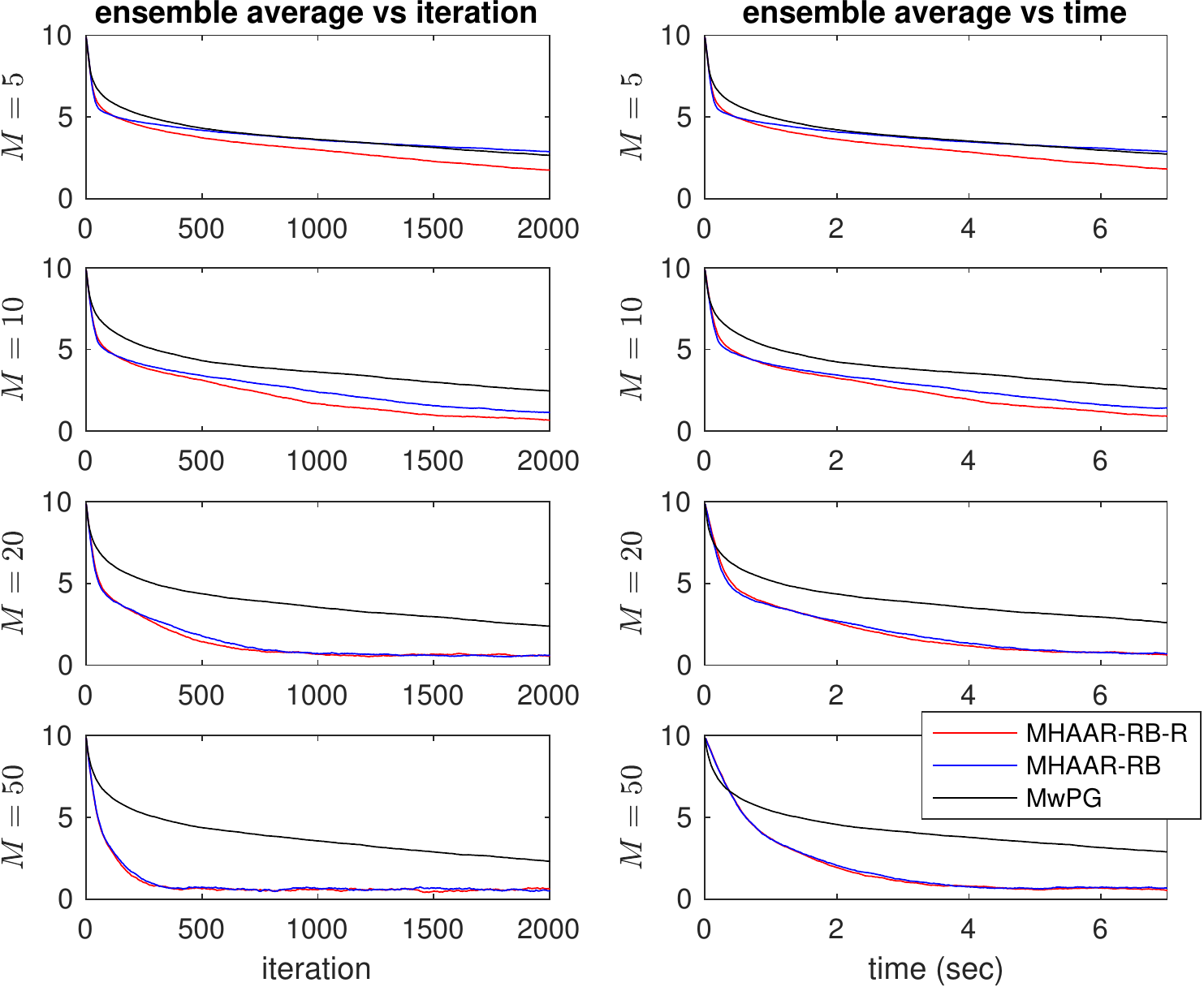}
\par\end{centering}
\caption{Ensemble averages for the posterior expectation of $\theta$ vs iteration
number and time for the algorithms compared in Table \ref{tbl: Lin-Gauss HMM IAC and IAC times CPU per iteration}.}

\label{fig: ensemble averages vs time for MHAAR-RB Lin-Gauss HMM}
\end{figure}
\end{example}

\subsubsection{Reduced computational cost via subsampling \label{subsec: Easing computational burden with subsampling: Multiple paths BS-SMC}}

The $\mathcal{O}(M^{2}T)$ cost per iteration of MHAAR-RB for SSM
precludes its application as $M$ becomes large, as required in some
applications. A computationally less demanding and intuitive version
of Algorithm \ref{alg: MHAAR-RB for SSM} could use a subsampled version
of the large sum in \eqref{eq: SMC acceptance ratio estimator all paths}
applying the backward sampling procedure $N$ times to recover $N$
paths. That is, letting $\mathfrak{u}=(u^{(1)},\ldots,u^{(N)})\in\mathsf{Z}^{TN}$,
a natural idea is to use the unbiased estimator of \eqref{eq: SMC acceptance ratio estimator all paths}
\begin{equation}
r_{z,\mathfrak{u}}^{N}(\theta,\vartheta;\zeta)=\frac{1}{N}\sum_{i=1}^{N}r_{z,u^{(i)}}(\theta,\vartheta;\zeta),\label{eq: acceptance ratio for SSM - subsampling}
\end{equation}
where
\[
u^{(1)},\ldots,u^{(N)}\overset{{\rm iid}}{\sim}\sum_{\mathbf{k}\in\left\llbracket M\right\rrbracket ^{T}}b_{\zeta}(\mathbf{k}|\mathbf{v})\delta_{v^{(\mathbf{k})}}(\cdot).
\]
Designing an algorithm using this acceptance ratio \eqref{eq: acceptance ratio for SSM - subsampling}
while preserving the correct invariant distribution $\pi(\mathrm{d}(\theta,z))$
is possible in the MHAAR framework. The resulting algorithm, which
we name MHAAR-S(ubsample) for SSM, is presented in Algorithm~\ref{alg: MHAAR-S for SSM}
in Appendix~\ref{subsec: Proof of reversibility for Algorithms-1}.
The computational complexity of MHAAR-S for SSM is $\mathcal{\mathcal{O}}(NMT)$
per iteration instead of $\mathcal{O}(M^{2}T)$ for Algorithm~\ref{alg: MHAAR-S for SSM}.
We note again that sampling $N$ paths using backward sampling is
an embarrassingly parallelisable operation. Details and correctness
of MHAAR-S for SSM as well as additional numerical results are provided
in Appendix~\ref{subsec: Reversibility of MHAAR-S-SSM}.
\begin{example}[\textbf{Example \ref{ex: GLSSM example}, ctd}]
 \label{ex: GLSSM example-cont}We run MHAAR-S for HMM for the dataset
used in Example~\ref{ex: GLSSM example} with $M=20$ particles and
several values of $N$. Table~\ref{tbl: MHAAR-S and MwPG in terms of IAC-GLSSM}
shows the IAC times for MHAAR-S for SSM, estimated from $2\times10^{6}$
iterations, in comparison with IAC times of MHAAR-RB-R and MwPG with
the same number of particles. We also show the ensemble averages of
those algorithms, obtained from 100 independent runs, in Figure~\ref{fig: ensemble averages vs time for MHAAR-SLin-Gauss HMM-1}.
Note that, using all the $M^{T}$ possible paths, the MHAAR-RB and
MHAAR-RB-R algorithms set a limit on the performance of MHAAR-S for
SSM. Both the table and the figure show that using multiple paths
results in gains in terms of convergence to equilibrium compared to
MwPG, illustrating the potential of the MAHHR approach to leverage
massively parallel architectures and reduce wall-clock time.

\begin{figure}[!h]
\begin{centering}
\includegraphics[scale=0.8]{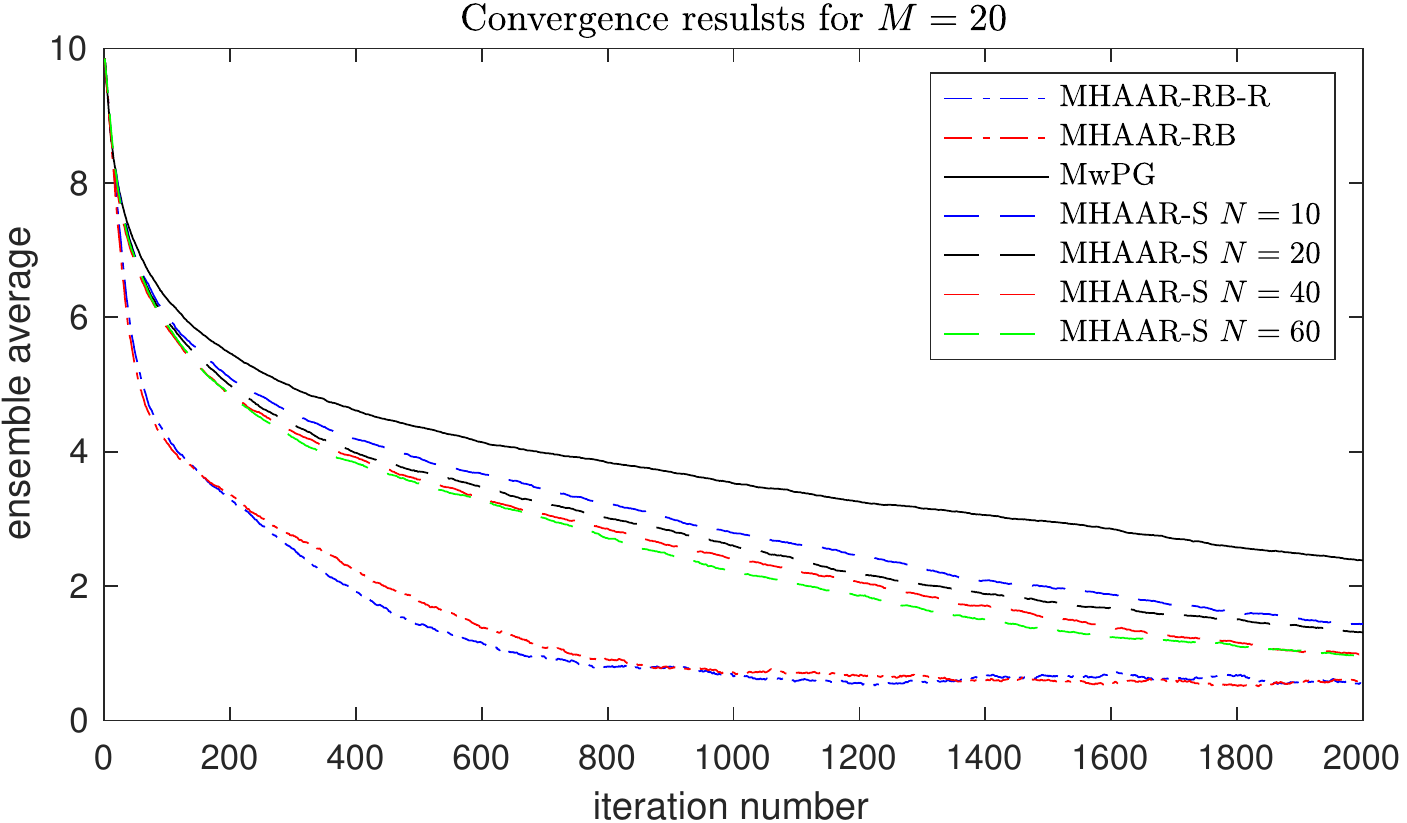}
\par\end{centering}
\caption{Ensemble averages for the posterior expectation of $\theta$ vs iteration
number and time for MHAAR-S, in comparison with  MHAAR-RB and MwPG.}

\label{fig: ensemble averages vs time for MHAAR-SLin-Gauss HMM-1}
\end{figure}

\begin{table}[!h]
\begin{centering}
{\small{}}%
\begin{tabular}[b]{|c|c|c|c|c|c|c|}
\hline 
\multicolumn{4}{|c|}{{\small{}MHAAR-S}} & \multirow{2}{*}{{\small{}MHAAR-RB-R}} & \multirow{2}{*}{{\small{}MHAAR-RB}} & \multirow{2}{*}{{\small{}MwPG}}\tabularnewline
\cline{1-4} \cline{2-4} \cline{3-4} \cline{4-4} 
{\small{}$N=10$} & {\small{}$N=20$} & {\small{}$N=40$} & {\small{}$N=60$} &  &  & \tabularnewline
\hline 
{\small{}2.0378} & {\small{}1.5770} & {\small{}1.5507} & {\small{}1.4047} & {\small{}0.4332} & {\small{}0.4713} & {\small{}3.5337}\tabularnewline
\hline 
\end{tabular}{\small\par}
\par\end{centering}
\caption{Comparison of MHAAR-S and MwPG in terms of IAC time ($\times10^{3}$).
Each run is performed for 500000 iterations. $M=20$ is taken for
all runs.}

\label{tbl: MHAAR-S and MwPG in terms of IAC-GLSSM}
\end{table}
\end{example}

\section{Discussion \label{sec: Discussion}}

In this paper, we exploit the ability to use more than one proposal
schemes within a MH update. We derive several useful MHAAR algorithms
that enable averaging multiple estimates of acceptance ratios, which
would not be valid by using a standard single proposal MH update.
The framework of MHAAR is rather general and provides a generic way
of improving performance of MH update based algorithm for a wide range
of problems. This is illustrated with doubly intractable models, general
latent variable models, trans-dimensional models, and general state-space
models. Although relevant in specific scenarios involving computations
on serial machines, MHAAR algorithms are particularly useful when
implemented on a parallel architecture since the computation required
to have an average acceptance ratio estimate can largely be parallelised.
In particular our experiments demonstrate significant reduction of
the burn in period required to reach equilibrium, an issue for which
very few generic approaches exist currently.

\section{Acknowledgements}

CA and SY acknowledge support from EPSRC ``Intractable Likelihood:
New Challenges from Modern Applications (ILike)'' (EP/K014463/1)
and the Isaac Newton Institute for Mathematical Sciences, Cambridge,
for support and hospitality during the programme ``Scalable inference;
statistical, algorithmic, computational aspects'' during which some
the work was carried out (EPSRC grant EP/K032208/1). CA and AD acknowledge
support of EPSRC grants Bayes4Health (EP/R018561/1) and CoSInES (EP/R034710/1).
NC is partially supported by a grant from the French National Research
Agency (ANR) as part of program ANR-11-LABEX-0047. The authors would
also like to thank Nick Whiteley for useful discussions.

\bibliographystyle{plainnat}
\bibliography{myrefs_thesis}

\begin{thebibliography}{37}
\providecommand{\natexlab}[1]{#1}
\providecommand{\url}[1]{\texttt{#1}}
\expandafter\ifx\csname urlstyle\endcsname\relax
  \providecommand{\doi}[1]{doi: #1}\else
  \providecommand{\doi}{doi: \begingroup \urlstyle{rm}\Url}\fi

\bibitem[Andrieu and Roberts(2009)]{Andrieu_and_Roberts_2009}
Christophe Andrieu and Gareth~O. Roberts.
\newblock The pseudo-marginal approach for efficient {M}onte {C}arlo
  computations.
\newblock \emph{Annals of Statistics}, 37\penalty0 (2):\penalty0 569--1078,
  2009.

\bibitem[Andrieu and Vihola(2016)]{Andrieu_and_Vihola_2014}
Christophe Andrieu and Matti Vihola.
\newblock Establishing some order amongst exact approximations of
  {M}{C}{M}{C}s.
\newblock \emph{Annals of Applied Probability}, 26\penalty0 (5):\penalty0
  2661--2696, 10 2016.

\bibitem[Andrieu et~al.(2010)Andrieu, Doucet, and
  Holenstein]{Andrieu_et_al_2010}
Christophe Andrieu, Arnaud Doucet, and Roman Holenstein.
\newblock Particle {M}arkov chain {M}onte {C}arlo methods.
\newblock \emph{Journal of the Royal Statistical Society: Series B (Statistical
  Methodology)}, 72:\penalty0 269--342, 2010.
\newblock \doi{10.1111/j.1467-9868.2009.00736.x}.

\bibitem[Andrieu et~al.(2018)Andrieu, Doucet, Y{\i}ld{\i}r{\i}m, and
  Chopin]{andrieu2018utility}
Christophe Andrieu, Arnaud Doucet, Sinan Y{\i}ld{\i}r{\i}m, and Nicolas Chopin.
\newblock On the utility of {M}etropolis-{H}astings with asymmetric acceptance
  ratio.
\newblock \emph{ArXiv e-prints}, \penalty0 (1803.09527), 2018.

\bibitem[Andrieu et~al.(2020)Andrieu, Lee, and
  Livingstone]{andrieu:lee:livingstone:2020}
Christophe Andrieu, Anthony Lee, and Sam Livingstone.
\newblock A general perspective on the {M}etropolis-{H}astings kernel.
\newblock \emph{ArXiv e-prints}, 2020.

\bibitem[Beaumont(2003)]{Beaumont_2003}
M.~Beaumont.
\newblock Estimation of population growth of decline in genetically monitored
  populations.
\newblock \emph{Genetics}, 164:\penalty0 1139--1160, 2003.

\bibitem[Beaumont et~al.(2002)Beaumont, Zhang, and
  Balding]{Beaumont_et_al_2002}
Mark~A. Beaumont, Wenyang Zhang, and David~J. Balding.
\newblock Approximate {B}ayesian computation in population genetics.
\newblock \emph{Genetics}, 162\penalty0 (4):\penalty0 2025--2035, December
  2002.
\newblock ISSN 0016-6731.
\newblock URL \url{http://www.genetics.org/content/162/4/2025.abstract}.

\bibitem[Bornn et~al.(2017)Bornn, Pillai, Smith, and Woodard]{bornn2017use}
Luke Bornn, Natesh~S Pillai, Aaron Smith, and Dawn Woodard.
\newblock The use of a single pseudo-sample in approximate {B}ayesian
  computation.
\newblock \emph{Statistics and Computing}, 27\penalty0 (3):\penalty0 583--590,
  2017.

\bibitem[Cainey(2013)]{cainey}
Joe Cainey.
\newblock \emph{Contributions to Exact Approximation Methodology}.
\newblock PhD thesis, University of Bristol, School of Mathematics, University
  of Bristol, 2013.

\bibitem[Del~Moral et~al.(2010)Del~Moral, Doucet, and Singh]{del2010backward}
Pierre Del~Moral, Arnaud Doucet, and Sumeetpal~S Singh.
\newblock A backward particle interpretation of {F}eynman-{K}ac formulae.
\newblock \emph{ESAIM: Mathematical Modelling and Numerical Analysis},
  44\penalty0 (5):\penalty0 947--975, 2010.

\bibitem[Deligiannidis et~al.(2018)Deligiannidis, Doucet, and
  Pitt]{Deligiannidis_et_al_2018}
Georgios Deligiannidis, Arnaud Doucet, and {Michael K.} Pitt.
\newblock The correlated pseudo-marginal method.
\newblock \emph{Journal of the Royal Statistical Society. Series B: Statistical
  Methodology}, 7 2018.
\newblock ISSN 1369-7412.
\newblock \doi{10.1111/rssb.12280}.

\bibitem[Dellaportas and Kontoyiannis(2012)]{dellaportas2012control}
Petros Dellaportas and Ioannis Kontoyiannis.
\newblock Control variates for estimation based on reversible {M}arkov chain
  {M}onte {C}arlo samplers.
\newblock \emph{Journal of the Royal Statistical Society: Series B (Statistical
  Methodology)}, 74\penalty0 (1):\penalty0 133--161, 2012.

\bibitem[Delmas and Jourdain(2009)]{delmas2009does}
Jean-Fran{\c{c}}cois Delmas and Benjamin Jourdain.
\newblock Does waste recycling really improve the multi-proposal
  {M}etropolis--{H}astings algorithm? {A}n analysis based on control variates.
\newblock \emph{Journal of Applied Probability}, 46\penalty0 (4):\penalty0
  938--959, 2009.

\bibitem[Green(1995)]{Green_1995}
P.~Green.
\newblock Reversible jump {M}arkov chain {M}onte {C}arlo for {B}ayesian model
  determination.
\newblock \emph{Biometrika}, 82\penalty0 (4):\penalty0 711--732, 1995.

\bibitem[Gunawan et~al.(2020)Gunawan, Carter, and Kohn]{gunawan2020scalable}
David Gunawan, Chris Carter, and Robert Kohn.
\newblock On scalable particle {M}arkov chain {M}onte {C}arlo, 2020.

\bibitem[Karagiannis and Andrieu(2013)]{Karagiannis_and_Andrieu_2013}
G.~Karagiannis and C.~Andrieu.
\newblock Annealed importance sampling for reversible jump {M}{C}{M}{C}
  algorithms.
\newblock \emph{Journal of Computational and Graphical Statistics}, 22\penalty0
  (3):\penalty0 623--648, 2013.

\bibitem[Lee et~al.(2010)Lee, Yau, Giles, Doucet, and Holmes]{lee2010utility}
Anthony Lee, Christopher Yau, Michael~B Giles, Arnaud Doucet, and Christopher~C
  Holmes.
\newblock On the utility of graphics cards to perform massively parallel
  simulation of advanced {M}onte {C}arlo methods.
\newblock \emph{Journal of computational and graphical statistics}, 19\penalty0
  (4):\penalty0 769--789, 2010.

\bibitem[Levin and Peres(2017)]{levin2017markov}
David~A Levin and Yuval Peres.
\newblock \emph{Markov chains and mixing times}, volume 107.
\newblock American Mathematical Soc., 2017.

\bibitem[Lindsten and Sch\"{o}n(2012)]{Lindsten_and_Schon_2012}
F.~Lindsten and T.~B. Sch\"{o}n.
\newblock On the use of backward simulation in the particle {G}ibbs sampler.
\newblock In \emph{2012 IEEE International Conference on Acoustics, Speech and
  Signal Processing (ICASSP)}, pages 3845--3848, March 2012.
\newblock \doi{10.1109/ICASSP.2012.6288756}.

\bibitem[Marjoram et~al.(2003)Marjoram, Molitor, Plagnol, and
  Tavar\'{e}]{Marjoram_et_al_2003}
Paul Marjoram, John Molitor, Vincent Plagnol, and Simon Tavar\'{e}.
\newblock Markov chain {M}onte {C}arlo without likelihoods.
\newblock \emph{Proceedings of the National Academy of Sciences of the United
  States of America}, 100\penalty0 (26):\penalty0 15324--15328, 2003.
\newblock ISSN 00278424.
\newblock URL \url{http://www.jstor.org/stable/3149004}.

\bibitem[M\o{l}ler et~al.(2006)M\o{l}ler, Pettitt, Reeves, and
  Berthelsen]{Muller_et_al_2006}
J.~M\o{l}ler, A.~N. Pettitt, R.~Reeves, and K.~K. Berthelsen.
\newblock An efficient {M}arkov chain {M}onte {C}arlo method for distributions
  with intractable normalising constants.
\newblock \emph{Biometrika}, 93\penalty0 (2):\penalty0 451--458, 2006.
\newblock \doi{10.1093/biomet/93.2.451}.
\newblock URL \url{http://biomet.oxfordjournals.org/content/93/2/451.abstract}.

\bibitem[M{\"u}ller and Stoyan(2002)]{mullercomparison}
A~M{\"u}ller and D~Stoyan.
\newblock Comparison methods for stochastic models and risks.
\newblock \emph{John Wiley\&Sons Ltd., Chichester}, 2002.

\bibitem[Murray et~al.(2006)Murray, Ghahramani, and MacKay]{Murray_et_al_2006}
I.~Murray, Z.~Ghahramani, and D.~J.~C. MacKay.
\newblock {MCMC for doubly-intractable distributions}.
\newblock In \emph{Proceedings of the 22nd Annual Conference on Uncertainty in
  Artificial Intelligence (UAI-06)}, pages 359--366, 2006.

\bibitem[{Natarajan} and {Chandrachoodan}(2018)]{8638034}
K.~{Natarajan} and N.~{Chandrachoodan}.
\newblock Lossless parallel implementation of a turbo decoder on {GPU}.
\newblock In \emph{2018 IEEE 25th International Conference on High Performance
  Computing (HiPC)}, pages 133--142, 2018.
\newblock \doi{10.1109/HiPC.2018.00023}.

\bibitem[Neal(2004)]{Neal_2004}
Radford~M. Neal.
\newblock Taking bigger {M}etropolis steps by dragging fast variables.
\newblock Technical report, University of Toronto, 2004.

\bibitem[Neal(2010)]{Neal_2010}
Radford~M. Neal.
\newblock {MCMC} using ensembles of states for problems with fast and slow
  variables such as {G}aussian process regression.
\newblock Technical report, University of Toronto, 2010.

\bibitem[Papaspiliopoulos et~al.(2003)Papaspiliopoulos, Roberts, and
  Skold]{papaspiliopoulos_et_al_2003}
O.~Papaspiliopoulos, G.O. Roberts, and M.~Skold.
\newblock Non-centred parameterisations for hierarchical models and data
  augmentation.
\newblock In {J. M.} Bernardo, {M. J.} Bayarri, {J. O.} Berger, {A. P.} Dawid,
  D.~Heckerman, {A. F. M.} Smith, and M.~West, editors, \emph{Bayesian
  Statistics VII}, pages 307--327. 2003.

\bibitem[Pritchard et~al.(1999)Pritchard, Seielstad, Perez-Lezaun, and
  Feldman]{Pritchard_et_al_1999}
J.~Pritchard, M.~Seielstad, A.~Perez-Lezaun, and M.~Feldman.
\newblock Population growth of human {Y} chromosomes: a study of {Y} chromosome
  microsatellites.
\newblock \emph{Molecular Biology and Evolution}, 16:\penalty0 1791--1798,
  1999.

\bibitem[Sherlock et~al.(2017)Sherlock, Thiery, and
  Lee]{doi:10.1093/biomet/asx031}
Chris Sherlock, Alexandre~H. Thiery, and Anthony Lee.
\newblock Pseudo-marginal {M}etropolis-{H}astings sampling using averages of
  unbiased estimators.
\newblock \emph{Biometrika}, 104\penalty0 (3):\penalty0 727--734, 2017.
\newblock \doi{10.1093/biomet/asx031}.
\newblock URL \url{+ http://dx.doi.org/10.1093/biomet/asx031}.

\bibitem[Sohn(1995)]{sohn1995parallel}
Andrew Sohn.
\newblock Parallel n-ary speculative computation of simulated annealing.
\newblock \emph{IEEE Transactions on Parallel and Distributed systems},
  6\penalty0 (10):\penalty0 997--1005, 1995.

\bibitem[Suchard et~al.(2010)Suchard, Wang, Chan, Frelinger, Cron, and
  West]{suchard2010understanding}
Marc~A Suchard, Quanli Wang, Cliburn Chan, Jacob Frelinger, Andrew Cron, and
  Mike West.
\newblock Understanding {GPU} programming for statistical computation:
  {S}tudies in massively parallel massive mixtures.
\newblock \emph{Journal of computational and graphical statistics}, 19\penalty0
  (2):\penalty0 419--438, 2010.

\bibitem[Tierney(1998)]{Tierney_1998}
Luke Tierney.
\newblock A note on {M}etropolis {H}astings kernels for general state spaces.
\newblock \emph{Annals of Applied Probability}, 8\penalty0 (1):\penalty0 1--9,
  1998.

\bibitem[Whiteley(2010)]{whiteley2010discussion}
Nick Whiteley.
\newblock Discussion on particle {M}arkov chain {M}onte {C}arlo methods.
\newblock \emph{Journal of the Royal Statistical Society: Series B},
  72\penalty0 (3):\penalty0 306--307, 2010.

\bibitem[Y{\i}ld{\i}r{\i}m et~al.(2015)Y{\i}ld{\i}r{\i}m, Singh, Dean, and
  Jasra]{Yildirim_et_al_2015}
Sinan Y{\i}ld{\i}r{\i}m, Sumeetpal~S. Singh, Thomas Dean, and Ajay Jasra.
\newblock Parameter estimation in hidden {M}arkov models with intractable
  likelihoods using sequential {M}onte {C}arlo.
\newblock \emph{Journal of Computational and Graphical Statistics}, 24\penalty0
  (3):\penalty0 846--865, 2015.
\newblock \doi{10.1080/10618600.2014.938811}.
\newblock URL \url{https://doi.org/10.1080/10618600.2014.938811}.

\bibitem[Y{\i}ld{\i}r{\i}m et~al.(2018)Y{\i}ld{\i}r{\i}m, Andrieu, and
  Doucet]{Yildirim_et_al_2018}
Sinan Y{\i}ld{\i}r{\i}m, Christophe Andrieu, and Arnaud Doucet.
\newblock Scalable {M}onte {C}arlo inference for state-space models, 2018.

\bibitem[Zanella(2020)]{zanella:2020}
Giacomo Zanella.
\newblock Informed proposals for local {MCMC} in discrete spaces.
\newblock \emph{Journal of the American Statistical Association}, 115\penalty0
  (530):\penalty0 852--865, 2020.
\newblock \doi{10.1080/01621459.2019.1585255}.
\newblock URL \url{https://doi.org/10.1080/01621459.2019.1585255}.

\bibitem[Zucchini et~al.(2016)Zucchini, MacDonald, and
  Langrock]{Zucchini_et_al_2016}
W.~Zucchini, I.~MacDonald, and R.~Langrock.
\newblock \emph{Hidden {M}arkov models for time series: {A}n introduction using
  {R}}, volume~6.
\newblock Chapman and Hall/CRC, 2016.
\newblock \doi{https://doi.org/10.1201/b20790}.

\end{thebibliography}

\appendix

\section{Proofs for the theorems in Section \ref{sec: Pseudo-marginal ratio algorithms using averaged acceptance ratio estimators}}

\subsection{Acceptance ratio of Algorithm \ref{alg: MHAAR for pseudo-marginal ratio in latent variable models-1}\label{subsec: Acceptance ratio of MHAAR for latent variable models}}
\begin{proof}[Proof of Theorem \ref{thm:expression-accept-ration-MHAAR}]
 Let $\mathring{\pi}_{0}(\mathrm{d}(\theta,z,u))=\pi(\mathrm{d}(\theta,z))q(\theta,\mathrm{d}\vartheta)Q_{\theta,\vartheta,z}(\mathrm{d}u)$.
Then, $r_{u^{(i)}}(\theta,\vartheta,z)$ is the acceptance ratio for
$\mathring{\pi}_{0}$ corresponding to the involution $\varphi_{0}(\theta,\vartheta,z,u)=(\vartheta,\theta,\phi_{\theta,\vartheta}(z,u))$.
Also, observe that, when $c=1$, we have 
\begin{align*}
\pi(\mathrm{d}\xi)= & \mathring{\pi}_{0}(\mathrm{d}(\theta,\vartheta,z,u^{(k)}))\prod_{i\neq k}Q_{\theta,\vartheta,z}({\rm d}u^{(i)})\frac{r_{u^{(k)}}(\theta,\vartheta,z)}{\sum_{i=1}^{N}r_{u^{(i)}}(\theta,\vartheta,z)}\\
\pi^{\varphi}(\mathrm{d}\xi)= & \mathring{\pi}_{0}^{\varphi_{0}}(\mathrm{d}(\theta,\vartheta,z,u^{(k)}))\prod_{i\neq k}Q_{\theta,\vartheta,z}({\rm d}u^{(i)})\frac{1}{N}.
\end{align*}
Therefore, for $c=1$, we have
\begin{align*}
\mathring{r}(\xi)=\frac{\pi^{\varphi}(\mathrm{d}\xi)}{\pi(\mathrm{d}\xi)}= & \frac{\mathring{\pi}_{0}^{\varphi_{0}}(\mathrm{d}(\theta,\vartheta,z,u^{(k)}))}{\mathring{\pi}_{0}(\mathrm{d}(\theta,\vartheta,z,u^{(k)}))}\frac{\prod_{i\neq k}Q_{\theta,\vartheta,z}({\rm d}u^{(i)})}{\prod_{i\neq k}Q_{\theta,\vartheta,z}({\rm d}u^{(i)})}\frac{1}{\frac{r_{u^{(k)}}(\theta,\vartheta,z)}{\frac{1}{N}\sum_{i=1}^{N}r_{u^{(i)}}(\theta,\vartheta,z)}}.\\
= & r_{\mathfrak{u}}^{N}(\theta,\vartheta,z),
\end{align*}
When $c=2$, we use the relation in \eqref{eq:prop-r-circ-phi-inverse-r}
to obtain

\begin{align*}
\mathring{r}(\theta,\vartheta,z,\mathfrak{u},k,2)= & \mathring{r}(\vartheta,\theta,z',\mathfrak{u}',k,1)^{-1}\\
= & \left[r_{\mathfrak{u}'}^{N}(\vartheta,\theta,z')\right]^{-1}.
\end{align*}
\end{proof}

\subsection{Proof of Theorem \ref{thm:theoreticaljustification}\label{subsec:Proof-of-Theorem justification}}
\begin{proof}[Proof of Theorem \ref{thm:theoreticaljustification}]
We start by noticing that the expression for the Dirichlet form associated
with $\mathring{P}^{N}$ can be rewritten in either of the following
simplified forms 
\begin{align*}
\mathcal{E}_{\mathring{P}^{N}}(f)= & \frac{1}{2}\int\pi({\rm d}\theta)\int_{\mathsf{\mathfrak{U}}\times\llbracket N\rrbracket}Q_{1}^{N}\big(\theta,{\rm d}(\vartheta,\mathfrak{u},k)\big)\min\{1,r_{\mathfrak{u}}^{N}(x,y)\}\left[f(\theta)-f(\vartheta)\right]^{2}\\
= & \frac{1}{2}\int\pi({\rm d}\theta)\int_{\mathfrak{U}\times\llbracket N\rrbracket}Q_{2}^{N}\big(\theta,{\rm d}(\vartheta,\mathfrak{u},k)\big)\min\{1,1/r_{\mathfrak{u}}^{N}(\vartheta,\theta)\}\left[f(\theta)-f(\vartheta)\right]^{2}.
\end{align*}
The expression on the first line turns out to be particularly convenient.
A well known result from the convex order literature states that for
any $n\geq2$ exchangeable random variables $Z_{1},\ldots,Z_{n}$
and any convex function $f$ we have $\mathbb{E}\left[f\left(n^{-1}\sum_{i=1}^{n}Z_{i}\right)\right]\leq\mathbb{E}\left[f\left((n-1)^{-1}\sum_{i=1}^{n-1}Z_{i}\right)\right]$
whenever the expectations exist \citep[Corollary 1.5.24]{mullercomparison}.
The two sums are said to be convex ordered. Now since $a\mapsto-\min\{1,a\}$
is convex we deduce that for any $N\geq1$, $\theta,\vartheta\in\mathsf{\Theta}$,
\begin{equation}
\int_{\mathsf{U}^{N}}Q_{\theta,\vartheta}^{N}({\rm d}\mathfrak{u})\min\{1,r_{\mathfrak{u}}^{N}(\theta,\vartheta)\}\leq\int_{\mathsf{U}^{N+1}}Q_{\theta,\vartheta}^{N+1}({\rm d}\mathfrak{u})\min\{1,r_{\mathfrak{\mathfrak{u}}}^{N+1}(\theta,\vartheta)\}\label{eq:convexorderingratio}
\end{equation}
where $Q_{\theta,\vartheta}^{N}(\mathrm{d}\mathfrak{u}):=\prod_{i=1}^{N}Q_{\theta,\vartheta}(\mathrm{d}u^{(i)})$,
and consequently for any $f\in L^{2}(\mathsf{\Theta},\pi)$ and $N\geq1$
\[
\mathcal{E}_{\mathring{P}^{N+1}}(f)\leq\mathcal{E}_{\mathring{P}^{N}}(f).
\]
All the monotonicity properties follow from \citet{Tierney_1998}
since $\mathring{P}^{N}$ and $\mathring{P}^{N+1}$ are $\pi-$reversible.
The comparisons to $P$ follow from the application of Jensen's inequality
to $a\mapsto\min\{1,a\}$, which leads for any $\theta,\vartheta\in\Theta$
to 
\[
\int_{\mathfrak{U}}Q_{\theta,\vartheta}^{N}({\rm d}\mathfrak{u})\min\{1,r_{\mathfrak{u}}^{N}(\theta,\vartheta)\}\leq\min\{1,r\big(\theta,\vartheta\big)\},
\]
and again using the results of \citet{Tierney_1998}.
\end{proof}

\section{Proofs for Section \ref{sec: MHAAR via Rao-Blackewellisation of PMR}\label{sec: Proofs for MHAAR-RB}}

We first prove Theorem \ref{thm: acceptance ratio for multiple latent variable model}
which establishes the expression for the acceptance ratio of MHAAR-RB
for the multiple latent variable model. Then, we prove the correctness
of the delayed rejection algorithm given in Section \ref{sec: A novel consistent pseudo-marginal estimator}.

\subsection{Acceptance ratio of Algorithm \ref{alg: MHAAR-RB}\label{subsec: Acceptance ratio of MHAAR-RB}}

For the multiple latent variable model in Section \ref{sec: A novel consistent pseudo-marginal estimator},
recall the joint density
\begin{equation}
\pi(\theta,z)\propto\eta(\theta)\prod_{t=1}^{T}\gamma_{t,\theta}(z_{t}).\label{eq: multiple latent variable model joint density re-stated}
\end{equation}
Define $C_{\theta}=\prod_{t=1}^{T}\int_{\mathsf{Z}}\gamma_{t,\theta}(z)\mathrm{d}z$,
and $C=\int_{\theta}\eta(\theta)C_{\theta}\mathrm{d}\theta$ so that
the marginal density is $\pi(\theta)=\eta(\theta)C_{\theta}/C$ and
the conditional density of the latent variables is 
\begin{equation}
\pi_{\theta}(z):=\frac{\pi(\theta,z)}{\pi(\theta)}=\frac{\prod_{t=1}^{T}\gamma_{t,\theta}(z_{t})}{C_{\theta}}.\label{eq: multiple latent variable model conditional density re-stated}
\end{equation}
Furthermore, for any $\theta,\vartheta\in\Theta^{2}$, $t\geq1$,
and $z,v\in\mathsf{Z}^{2}$, define
\[
w_{t,\theta,\vartheta}(v):=\frac{\gamma_{t,\theta,\vartheta}(v)}{q_{t,\theta,\vartheta}(v)},\quad\text{and}\quad\lambda_{t,\theta,\vartheta}(z,u):=\frac{\gamma_{t,\theta,\vartheta}(z)}{\gamma_{t,\theta}(z)}\frac{\gamma_{t,\vartheta}(u)}{\gamma_{t,\theta,\vartheta}(u)}.
\]
We need the following preparatory lemmas for the proofs in this section.
\begin{lem}
\label{lem: cSMC kernel ratio MHAAR-RB}For any $\theta,\vartheta\in\Theta$,
$\mathbf{v}\in\mathsf{Z}^{MT}$, $\mathbf{k}\in\left\llbracket M\right\rrbracket {}^{T}$,
and , we have
\[
\frac{(\pi_{\theta}\otimes\Phi_{\theta,\vartheta})^{\mathfrak{s}_{\mathbf{1},\mathbf{k}}}(\mathrm{d}\mathbf{v})b_{\theta,\vartheta}(\mathbf{k}|\mathfrak{s}_{\mathbf{1},\mathbf{k}}(\mathbf{v}))}{(\pi_{\theta}\otimes\Phi_{\theta,\vartheta})(\mathrm{d}\mathbf{v})b_{\theta,\vartheta}(\mathbf{k}|\mathbf{v})}=\prod_{t=1}^{T}\frac{\gamma_{t,\theta}(v_{t}^{(k_{t})})}{\gamma_{t,\theta,\vartheta}(v_{t}^{(k_{t})})}\frac{\gamma_{t,\theta,\vartheta}(v_{t}^{(1)})}{\gamma_{t,\theta}(v_{t}^{(1)})},
\]
where $\mathfrak{s}_{\mathbf{1},\mathbf{k}}$, $\Phi_{\theta,\vartheta}$,
$b_{\theta,\vartheta}$ are defined in \eqref{eq: swapping operator for MHAAR-RB},
\eqref{eq: cSMC for MHAAR-RB}, and in \eqref{eq: Selection probabilities of MHAAR-RB-iid model},
respectively.
\end{lem}
\begin{proof}[Proof of Lemma \ref{lem: cSMC kernel ratio MHAAR-RB}]
 The denominator is equal to 
\begin{align*}
(\pi_{\theta}\otimes\Phi_{\theta,\vartheta})(\mathrm{d}\mathbf{v})b_{\theta,\vartheta}(\mathbf{k}|\mathbf{v})= & \pi_{\theta}(v^{(\mathbf{1})})\left[\prod_{t=1}^{T}\prod_{i=2}^{M}q_{t,\theta,\vartheta}(v_{t}^{(i)})\right]\prod_{t=1}^{T}\frac{w_{t,\theta,\vartheta}(v_{t}^{(k_{t})})}{\sum_{i=1}^{M}w_{t,\theta,\vartheta}(v_{t}^{(i)})}\\
= & \frac{1}{C_{\theta}}\left[\prod_{t=1}^{T}\frac{\gamma_{t,\theta}(v_{t}^{(1)})}{q_{t,\theta,\vartheta}(v_{t}^{(1)})}\prod_{i=1}^{M}q_{t,\theta,\vartheta}(v_{t}^{(i)})\right]\prod_{t=1}^{T}\frac{w_{t,\theta,\vartheta}(v_{t}^{(k_{t})})}{\sum_{i=1}^{M}w_{t,\theta,\vartheta}(v_{t}^{(i)})}\\
= & \frac{1}{C_{\theta}}\prod_{t=1}^{T}\frac{\gamma_{t,\theta}(v_{t}^{(1)})\gamma_{t,\theta}(v_{t}^{(k_{t})})}{q_{t,\theta,\vartheta}(v_{t}^{(1)})q_{t,\theta,\vartheta}(v_{t}^{(k_{t})})}\frac{\prod_{i=1}^{M}q_{t,\theta,\vartheta}(v_{t}^{(i)})}{\sum_{i=1}^{M}w_{t,\theta,\vartheta}(v_{t}^{(i)})}\prod_{t=1}^{T}\frac{\gamma_{t,\theta,\vartheta}(v_{t}^{(k_{t})})}{\gamma_{t,\theta}(v_{t}^{(k_{t})})}.
\end{align*}
The numerator is obtained by swapping $v^{(\mathbf{1})}$ and $v^{(\mathbf{k})}$
in the denominator, therefore it is equal to
\[
(\pi_{\theta}\otimes\Phi_{\theta,\vartheta})^{\mathfrak{s}_{\mathbf{1},\mathbf{k}}}(\mathrm{d}\mathbf{v})b_{\theta,\vartheta}(\mathbf{k}|\mathfrak{s}_{\mathbf{1},\mathbf{k}}(\mathbf{v}))=\frac{1}{C_{\theta}}\prod_{t=1}^{T}\frac{\gamma_{t,\theta}(v_{t}^{(1)})\gamma_{t,\theta}(v_{t}^{(k_{t})})}{q_{t,\theta,\vartheta}(v_{t}^{(1)})q_{t,\theta,\vartheta}(v_{t}^{(k_{t})})}\frac{\prod_{i=1}^{M}q_{t,\theta,\vartheta}(v_{t}^{(i)})}{\sum_{i=1}^{M}w_{t,\theta,\vartheta}(v_{t}^{(i)})}\prod_{t=1}^{T}\frac{\gamma_{t,\theta,\vartheta}(v_{t}^{(1)})}{\gamma_{t,\theta}(v_{t}^{(1)})}
\]
Taking the ratio yields the result.
\end{proof}
\begin{lem}
\label{lem: relation between selection probabilities MHAAR-RB}For
any $\theta,\vartheta\in\Theta$, $\mathbf{v}\in\mathsf{Z}^{MT}$,
and $\mathbf{k}\in\left\llbracket M\right\rrbracket {}^{T}$, we have
\[
b_{\theta,\vartheta}^{(1)}(\mathbf{k}|\mathbf{v})=b_{\theta,\vartheta}(\mathbf{k}|\mathbf{v})\frac{r_{v^{(\mathbf{1})},v^{(\mathbf{k})}}(\theta,\vartheta)}{r_{\mathbf{1},\mathbf{v}}(\theta,\vartheta)}
\]
where $b_{\theta,\vartheta}$ and $b_{\theta,\vartheta}^{(1)}$ are
defined in \eqref{eq: Selection probabilities of MHAAR-RB-1} and
\eqref{eq: Selection probabilities of MHAAR-RB-iid model}, \textup{and
$r_{v^{(\mathbf{1})},v^{(\mathbf{k})}}(\theta,\vartheta)$} and $r_{\mathbf{1},\mathbf{v}}(\theta,\vartheta)$
are defined in \eqref{eq: PMR acceptance ratio multiple latent variable}
and \eqref{eq: MHAAR-RB acceptance ratio-iid model}, respectively.
\end{lem}
\begin{proof}[Proof of Lemma \ref{lem: relation between selection probabilities MHAAR-RB}]
We prove simply by showing that the ratio
\begin{align*}
\frac{b_{\theta,\vartheta}^{(1)}(\mathbf{k}|\mathbf{v})}{b_{\theta,\vartheta}(\mathbf{k}|\mathbf{v})}= & \prod_{t=1}^{T}\frac{\frac{w_{t,\theta,\vartheta}(v_{t}^{(k_{t})})\lambda_{t,\theta,\vartheta}(v_{t}^{(1)},v_{t}^{(k_{t})})}{\sum_{j=1}^{M}w_{t,\theta,\vartheta}(v_{t}^{(j)})\lambda_{t,\theta,\vartheta}(v_{t}^{(1)},v_{t}^{(j)})}}{\frac{w_{t,\theta,\vartheta}(v_{t}^{(k_{t})})}{\sum_{j=1}^{M}w_{t,\theta,\vartheta}(v_{t}^{(j)})}}\\
= & \prod_{t=1}^{T}\frac{\sum_{j=1}^{M}w_{t,\theta,\vartheta}(v_{t}^{(j)})}{\sum_{j=1}^{M}w_{t,\theta,\vartheta}(v_{t}^{(j)})\lambda_{t,\theta,\vartheta}(v_{t}^{(1)},v_{t}^{(j)})}\prod_{t=1}^{T}\frac{w_{t,\theta,\vartheta}(v_{t}^{(k_{t})})\lambda_{t,\theta,\vartheta}(v_{t}^{(1)},v_{t}^{(k_{t})})}{w_{t,\theta,\vartheta}(v_{t}^{(k_{t})})}\\
= & \frac{\eta(\vartheta)q(\vartheta,\theta)}{\eta(\theta)q(\theta,\vartheta)}\prod_{t=1}^{T}\lambda_{t,\theta,\vartheta}(v_{t}^{(1)},v_{t}^{(k)})r_{\mathbf{1},\mathbf{v}}(\theta,\vartheta)^{-1}\\
= & \frac{r_{v^{(\mathbf{1})},v^{(\mathbf{k})}}(\theta,\vartheta)}{r_{\mathbf{1},\mathbf{v}}(\theta,\vartheta)}
\end{align*}
as claimed.
\end{proof}
The next lemma can be verified by inspection and therefore will be
left without a proof.
\begin{lem}
\label{lem: Invariance of r to permutations}For any $\theta,\vartheta\in\Theta$,
$\mathbf{v}\in\mathsf{Z}^{MT}$, and $\mathbf{k}\in\left\llbracket M\right\rrbracket {}^{T}$,
we have \textup{$r_{\mathbf{k},\mathbf{v}}(\theta,\vartheta)=r_{\mathbf{1},\mathfrak{s}_{\mathbf{1},\mathbf{k}}(\mathbf{v})}(\theta,\vartheta)$}.
\end{lem}
Now we prove Theorem \ref{thm: acceptance ratio for multiple latent variable model}
using the lemmas above.
\begin{proof}[Proof Theorem \ref{thm: acceptance ratio for multiple latent variable model}]
Recalling the notation in Section \ref{sec: A novel consistent pseudo-marginal estimator},
the joint distribution $\mathring{\pi}$ for $\xi=(\theta,\vartheta,\mathbf{v},\mathbf{k},c)$
can be written in compact form as 
\[
\mathring{\pi}(\mathrm{d}\xi)=\frac{1}{2}\pi(\mathrm{d}\theta)q(\theta,\mathrm{d}\vartheta)\left[\mathbb{I}_{1}(c)(\pi_{\theta}\otimes\Phi_{\theta,\vartheta})(\mathrm{d}\mathbf{v})b_{\theta,\vartheta}^{(1)}(\mathbf{k}|\mathbf{v})+\mathbb{I}_{2}(c)(\pi_{\theta}\otimes\Phi_{\vartheta,\theta})(\mathrm{d}\mathbf{v})b_{\vartheta,\theta}(\mathbf{k}|\mathbf{v})\right],
\]
and the proposed involution is $\varphi(\theta,\vartheta,\mathbf{v},\mathbf{k},c)=(\vartheta,\theta,\mathfrak{s}_{\mathbf{1},\mathbf{k}}(\mathbf{v}),\mathbf{k},3-c)$.
When $c=1$, $\mathring{r}(\xi)$ is
\begin{align*}
\frac{\mathring{\pi}^{\varphi}(\mathrm{d}\xi)}{\mathring{\pi}(\mathrm{d}\xi)}= & \frac{q(\vartheta,\theta)\pi(\vartheta)}{q(\theta,\vartheta)\pi(\theta)}\frac{(\pi_{\vartheta}\otimes\Phi_{\theta,\vartheta})^{\mathfrak{s}_{\mathbf{1},\mathbf{k}}}(\mathrm{d}\mathbf{v})b_{\theta,\vartheta}(\mathbf{k}|\mathfrak{s}_{\mathbf{1},\mathbf{k}}(\mathbf{v}))}{(\pi_{\theta}\otimes\Phi_{\theta,\vartheta})(\mathrm{d}\mathbf{v})b_{\theta,\vartheta}^{(1)}(\mathbf{k}|\mathbf{v})}\\
= & \frac{q(\vartheta,\theta)\pi(\vartheta)}{q(\theta,\vartheta)\pi(\theta)}\frac{\pi_{\vartheta}(\mathrm{d}v^{(\mathbf{k})})}{\pi_{\theta}(\mathrm{d}v^{(\mathbf{k})})}\frac{(\pi_{\theta}\otimes\Phi_{\theta,\vartheta})^{\mathfrak{s}_{\mathbf{1},\mathbf{k}}}(\mathrm{d}\mathbf{v})b_{\theta,\vartheta}(\mathbf{k}|\mathfrak{s}_{\mathbf{1},\mathbf{k}}(\mathbf{v}))}{(\pi_{\theta}\otimes\Phi_{\theta,\vartheta})(\mathrm{d}\mathbf{v})b_{\theta,\vartheta}(\mathbf{k}|\mathbf{v})}\frac{r_{\mathbf{1},\mathbf{v}}(\theta,\vartheta)}{r_{v^{(\mathbf{1})},v^{(\mathbf{k})}}(\theta,\vartheta)}\\
= & \frac{q(\vartheta,\theta)\pi(\vartheta)}{q(\theta,\vartheta)\pi(\theta)}\frac{\pi_{\vartheta}(v^{(\mathbf{k})})}{\pi_{\theta}(v^{(\mathbf{k})})}\prod_{t=1}^{T}\frac{\gamma_{t,\theta}(v_{t}^{(k_{t})})}{\gamma_{t,\theta,\vartheta}(v_{t}^{(k_{t})})}\frac{\gamma_{t,\theta,\vartheta}(v_{t}^{(1)})}{\gamma_{t,\theta}(v_{t}^{(1)})}\frac{1}{r_{v^{(\mathbf{1})},v^{(\mathbf{k})}}(\theta,\vartheta)}r_{\mathbf{1},\mathbf{v}}(\theta,\vartheta)\\
= & \frac{q(\vartheta,\theta)\eta(\vartheta)}{q(\theta,\vartheta)\eta(\theta)}\prod_{t=1}^{T}\frac{\gamma_{t,\vartheta}(v^{(k_{t})})}{\gamma_{t,\theta}(v^{(k_{t})})}\frac{\gamma_{t,\theta}(v_{t}^{(k_{t})})}{\gamma_{t,\theta,\vartheta}(v_{t}^{(k_{t})})}\frac{\gamma_{t,\theta,\vartheta}(v_{t}^{(1)})}{\gamma_{t,\theta}(v_{t}^{(1)})}\frac{1}{r_{v^{(\mathbf{1})},v^{(\mathbf{k})}}(\theta,\vartheta)}r_{\mathbf{1},\mathbf{v}}(\theta,\vartheta)\\
= & r_{\mathbf{1},\mathbf{v}}(\theta,\vartheta),
\end{align*}
where we have used Lemma \ref{lem: relation between selection probabilities MHAAR-RB}
for the second line, Lemma \ref{lem: cSMC kernel ratio MHAAR-RB}
for the third line, \eqref{eq: multiple latent variable model joint density re-stated}
for the fourth line, and the definition of acceptance ratio $r_{v^{(\mathbf{1})},v^{(\mathbf{k})}}(\theta,\vartheta)$
in \eqref{eq: PMR acceptance ratio multiple latent variable} for
the last line.

For $c=2$, upon using \eqref{eq:prop-r-circ-phi-inverse-r}, we write
\begin{align*}
\mathring{r}(\theta,\vartheta,\mathbf{v},\mathbf{k},2)= & \mathring{r}(\vartheta,\theta,\mathfrak{s}_{\mathbf{1},\mathbf{k}}(\mathbf{v}),\mathbf{k},1)^{-1}\\
= & r_{\mathbf{1},\mathfrak{s}_{\mathbf{1},\mathbf{k}}(\mathbf{\mathbf{v}})}(\vartheta,\theta)^{-1}\\
= & r_{\mathbf{k},\mathbf{\mathbf{v}}}(\vartheta,\theta)^{-1},
\end{align*}
 where the last line is due to $\mathfrak{s}_{\mathbf{1},\mathbf{k}}(\mathbf{v})^{\mathbf{(1)}}=v^{(\mathbf{k})}$
and Lemma \ref{lem: Invariance of r to permutations}.
\end{proof}

\subsection{Delayed rejection step for Algorithm \ref{alg: MHAAR-RB}\label{subsec: Delayed rejection step for MHAAR-RB}}

When the delayed rejection step is included in Algorithm \ref{alg: MHAAR-RB},
the algorithm targets the modified joint distribution for $\check{\xi}:=(\xi,\mathbf{l},\mathbf{l}')$
defined as

\[
\check{\pi}(\mathrm{d}\check{\xi})=\mathring{\pi}(\mathrm{d}\xi)\left[\mathbb{I}_{1}(c)b^{\textup{ref}}(\mathbf{l},\mathbf{l}'|\xi)+\mathbb{I}_{2}(c)b^{\textup{ref}}(\mathbf{l}',\mathbf{l}|\varphi(\xi))\right]
\]
where $\xi=(\theta,\vartheta,\mathbf{v},\mathbf{k},c)$ is as in Section
\ref{subsec: Acceptance ratio of MHAAR-RB}, and, conditional on $\xi$,
the joint probability distribution of $\mathbf{l},\mathbf{l}'\in\left\llbracket M\right\rrbracket {}^{T}$
is given by 
\[
b^{\textup{ref}}(\mathbf{l},\mathbf{l}'|\xi):=b_{\theta,\vartheta}^{\textup{ref},(1)}(\mathbf{l}|\mathbf{v})b_{\vartheta,\theta}^{\textup{ref},(2)}(\mathbf{l}'|\mathbf{\mathfrak{s}_{\mathbf{1},\mathbf{k}}(\mathbf{v})})
\]
with the individual probabilities defined as
\[
b_{\theta,\vartheta}^{\textup{ref},(1)}(\mathbf{l}|\mathbf{v}):=\prod_{t=1}^{T}\frac{\gamma_{t,\theta}(v_{t}^{(l_{t})})/q_{t,\theta,\vartheta}(v_{t}^{(l_{t})})}{\sum_{i=1}^{N}\gamma_{t,\theta}(v_{t}^{(i)})/q_{t,\theta,\vartheta}(v_{t}^{(i)})},\quad b_{\theta,\vartheta}^{\textup{ref},(2)}(\mathbf{l}|\mathbf{v}):=\prod_{t=1}^{T}\frac{\gamma_{t,\theta}(v_{t}^{(l_{t})})/q_{t,\vartheta,\theta}(v_{t}^{(l_{t})})}{\sum_{i=1}^{N}\gamma_{t,\theta}(v_{t}^{(i)})/q_{t,\vartheta,\theta}(v_{t}^{(i)})}.
\]
These are simply the selection probabilities of $\gamma_{t,\theta}$-invariant
cSMC kernels when the proposed values are sampled from $q_{t,\theta,\vartheta}$
or $q_{t,\vartheta,\theta}$, respectively. 

The algorithm can be thought of as a two-stage delayed rejection algorithm,
where the stage one move corresponds to the regular MHAAR update and
stage two move is executed only if the move in stage one is rejected.
We note that, in practice, the pair $\mathbf{l},\mathbf{l}'$ do not
play any role in the implementation of the first stage. Moreover,
in the implementation of the second stage, one only needs to sample
$\mathbf{l}$ to propose $v^{(\mathbf{l})}$ for the latent variable;
$\mathbf{l}'$ is, again, not needed.

The mentioned two stages of the delayed rejection algorithm are given
below.
\begin{enumerate}
\item In the first stage, MHAAR attempts a transition for the joint variable
$\check{\xi}=(\xi,\mathbf{l},\mathbf{l}')$ as 
\[
\check{\varphi}_{1}(\xi,\mathbf{l},\mathbf{l}'):=(\varphi(\xi),\mathbf{l}',\mathbf{l}).
\]
As $\check{\varphi}_{1}$ is an involution, it yields the acceptance
ratio 
\begin{align}
\check{r}_{1}(\check{\xi}):= & \frac{\check{\pi}^{\check{\varphi}_{1}}(\mathrm{d}\check{\xi})}{\check{\pi}(\mathrm{d}\check{\xi})}\nonumber \\
= & \frac{\mathring{\pi}^{\varphi}(\mathrm{d}\xi)}{\mathring{\pi}(\mathrm{d}\xi)}\frac{b^{\textup{ref}}(\mathbf{l},\mathbf{l}'|\varphi\circ\varphi(\xi))}{b^{\textup{ref}}(\mathbf{l},\mathbf{l}'|\xi)}\nonumber \\
= & \frac{\mathring{\pi}^{\varphi}(\mathrm{d}\xi)}{\mathring{\pi}(\mathrm{d}\xi)}=\mathring{r}(\xi),\label{eq: equality of acceptance ratio of delayed rejection - multiple latent variable}
\end{align}
which is exactly the same acceptance ratio as we would have without
the delayed rejection step. Note that, neither the acceptance ratio
nor the variables carried on to the next iteration depend on the additional
variables $\mathbf{l}$ or $\mathbf{l}'$. Therefore, the variables
$\mathbf{l}$ and $\mathbf{l}'$ need not be sampled prior to the
delayed rejection step.
\item The second stage corresponds to proposing a transformation of $\check{\xi}=(\xi,\mathbf{l},\mathbf{l}')$,
recalling that $\xi=(\theta,\vartheta,\mathbf{v},\mathbf{k},c)$,
with the following involution:
\[
\check{\varphi}_{2}(\xi,\mathbf{l},\mathbf{l}'):=(\theta,\vartheta,\mathfrak{s}_{\mathbf{1},\mathbf{l}}(\mathbf{v}),\mathfrak{r}_{\mathbf{\mathbf{l}}}(\mathbf{k}),c,\mathbf{l},\mathfrak{r}_{\mathbf{\mathbf{l}}}(\mathbf{l}'))
\]
where, for any $\mathbf{k},\mathbf{l}\in\left\llbracket M\right\rrbracket {}^{T}$,
we define $\mathfrak{r}_{\mathbf{l}}(\mathbf{k}):\left\llbracket M\right\rrbracket ^{T}\mapsto\left\llbracket M\right\rrbracket ^{T}$
as,
\begin{equation}
[\mathfrak{r}_{\mathbf{l}}(\mathbf{k})]_{i}=\begin{cases}
l_{i}, & k_{i}=1\\
1, & k_{i}=l_{i}\\
k_{i}, & \text{otherwise}
\end{cases},\quad i=1,\ldots,T,\label{eq: change index according to swap}
\end{equation}
That is, $\mathfrak{r}_{\mathbf{l}}(\mathbf{k})$ is the set of indices
of the elements of $v^{(\mathbf{k})}$ once $v^{(\mathbf{1})}$ and
$v^{(\mathbf{l})}$ have been swapped in $\mathbf{v}$, so that $[\mathfrak{s}_{\mathbf{1},\mathbf{l}}(\mathbf{v})]^{(\mathbf{k})}=v^{(\mathfrak{r}_{\mathbf{l}}(\mathbf{k}))}$.
We note that, the operator $\mathfrak{r}_{\mathbf{l}}$ is merely
introduced to establish the correctness of the algorithm and in practice
does not need to be implemented. Crucially for our analysis, it can
be checked that, for any $\mathbf{l}\in\left\llbracket M\right\rrbracket {}^{T}$,
the operator $\mathfrak{r}_{\mathbf{l}}(\cdot)$ is an involution,
resulting in $\check{\varphi}_{2}$ also being an involution. This
enables us to cast the delayed rejection scheme in our framework.
the acceptance ratio in the second stage can be written as 
\begin{equation}
\check{r}_{2}(\check{\xi})=\frac{\check{\pi}^{\check{\varphi}_{2}}(\mathrm{d}\check{\xi})}{\check{\pi}(\mathrm{d}\check{\xi})}\frac{1-\min\left\{ 1,\check{r}_{1}\circ\check{\varphi}_{2}(\check{\xi})\right\} }{1-\min\left\{ 1,\check{r}_{1}(\check{\xi})\right\} }\label{eq: acceptance probability of delayed rejection - latent variable}
\end{equation}
\end{enumerate}
\begin{thm}
\label{thm: acceptance ratio of delayed rejection}The acceptance
ratio in \eqref{eq: acceptance probability of delayed rejection - latent variable}
is equal to
\[
\check{r}_{2}(\check{\xi})=\begin{cases}
\frac{1-\min\left\{ 0,r_{\mathbf{l},\mathbf{v}}(\theta,\vartheta)\right\} }{1-\min\left\{ 0,r_{\mathbf{1},\mathbf{v}}(\theta,\vartheta)\right\} }, & c=1;\\
\frac{1-\min\left\{ 0,1/r_{\mathbf{k},\mathbf{v}}(\vartheta,\theta)\right\} }{1-\min\left\{ 0,1/r_{\mathbf{1},\mathbf{v}}(\vartheta,\theta)\right\} }, & c=2.
\end{cases}
\]
Moreover, when $\gamma_{t,\theta,\vartheta}=\gamma_{t,\theta}$ for
all $t,\theta,\vartheta$, the acceptance ratio simplifies to $\check{r}_{2}(\check{\xi})=1$.
\end{thm}
In the proof of Theorem \ref{thm: acceptance ratio of delayed rejection},
we will make use of the following lemmas.
\begin{lem}
\label{lem: facts for MHAAR-RB - 1}For any $\theta,\vartheta\in\Theta$,
$\mathbf{v}\in\mathsf{Z}^{MT}$, and $\mathbf{k},\mathbf{l}\in\left\llbracket M\right\rrbracket {}^{T}$,
we have the following facts.
\begin{itemize}
\item $[\mathfrak{s}_{\mathbf{1},\mathbf{l}}(\mathbf{v})]^{(\mathfrak{r}_{\mathbf{\mathbf{l}}}(\mathbf{k}))}=v^{(\mathbf{k})}.$
\item The following equalities hold \textup{
\begin{align}
(\pi_{\theta}\otimes\Phi_{\theta,\vartheta})(\mathrm{d}\mathbf{v})b_{\theta,\vartheta}^{\textup{ref,}(1)}(\mathbf{k}|\mathbf{v})= & (\pi_{\theta}\otimes\Phi_{\theta,\vartheta})^{\mathfrak{s}_{\mathbf{1},\mathbf{k}}}(\mathrm{d}\mathbf{v})b_{\theta,\vartheta}^{\textup{ref,}(1)}(\mathbf{k}|\mathfrak{s}_{\mathbf{1},\mathbf{k}}(\mathbf{v}))\label{eq: cSMC with eta1}\\
(\pi_{\theta}\otimes\Phi_{\vartheta,\theta})(\mathrm{d}\mathbf{v})b_{\theta,\vartheta}^{\textup{ref,}(2)}(\mathbf{k}|\mathbf{v})= & (\pi_{\theta}\otimes\Phi_{\vartheta,\theta})^{\mathfrak{s}_{\mathbf{1},\mathbf{k}}}(\mathrm{d}\mathbf{v})b_{\theta,\vartheta}^{\textup{ref,}(2)}(\mathbf{k}|\mathfrak{s}_{\mathbf{1},\mathbf{k}}(\mathbf{v}))\label{eq: cSMC with eta2}
\end{align}
}
\end{itemize}
\end{lem}
\begin{proof}[Proof Lemma \ref{lem: facts for MHAAR-RB - 1}]
 The the first part of Lemma \ref{lem: facts for MHAAR-RB - 1} can
be verified by inspection. For \eqref{eq: cSMC with eta1}, we write
\begin{align*}
(\pi_{\theta}\otimes\Phi_{\theta,\vartheta})(\mathrm{d}\mathbf{v}) & b_{\theta,\vartheta}^{\textup{ref,}(1)}(\mathbf{k}|\mathbf{v})=\pi_{\theta}(v^{(\mathbf{1})})\left[\prod_{t=1}^{T}\prod_{i=2}^{M}q_{t,\theta,\vartheta}(v_{t}^{(i)})\right]\prod_{t=1}^{T}\frac{\gamma_{t,\theta}(v_{t}^{(k_{t})})/q_{t,\theta,\vartheta}(v_{t}^{(k_{t})})}{\sum_{i=1}^{N}\gamma_{t,\theta}(v_{t}^{(i)})/q_{t,\theta,\vartheta}(v_{t}^{(i)})}\\
= & \pi_{\theta}(v^{(\mathbf{k})})\left[\prod_{t=1}^{T}\frac{\gamma_{t,\theta}(v_{t}^{(1)})}{\gamma_{t,\theta}(v_{t}^{(k_{t})})}\prod_{i=2}^{M}q_{t,\theta,\vartheta}(v_{t}^{(i)})\right]\prod_{t=1}^{T}\frac{\gamma_{t,\theta}(v_{t}^{(k_{t})})/q_{t,\theta,\vartheta}(v_{t}^{(k_{t})})}{\sum_{i=1}^{N}\gamma_{t,\theta}(v_{t}^{(i)})/q_{t,\theta,\vartheta}(v_{t}^{(i)})}\\
= & \pi_{\theta}(v^{(\mathbf{k})})\left[\prod_{t=1}^{T}\frac{\gamma_{t,\theta}(v_{t}^{(1)})}{\gamma_{t,\theta}(v_{t}^{(k_{t})})q_{t,\theta,\vartheta}(v_{t}^{(1)})}\prod_{i=1}^{M}q_{t,\theta,\vartheta}(v_{t}^{(i)})\right]\prod_{t=1}^{T}\frac{\gamma_{t,\theta}(v_{t}^{(k_{t})})/q_{t,\theta,\vartheta}(v_{t}^{(k_{t})})}{\sum_{i=1}^{N}\gamma_{t,\theta}(v_{t}^{(i)})/q_{t,\theta,\vartheta}(v_{t}^{(i)})}\\
= & \pi_{\theta}(v^{(\mathbf{k})})\left[\prod_{t=1}^{T}\frac{\gamma_{t,\theta}(v_{t}^{(k_{t})})}{\gamma_{t,\theta}(v_{t}^{(k_{t})})q_{t,\theta,\vartheta}(v_{t}^{(k_{t})})}\prod_{i=1}^{M}q_{t,\theta,\vartheta}(v_{t}^{(i)})\right]\prod_{t=1}^{T}\frac{\gamma_{t,\theta}(v_{t}^{(1)})/q_{t,\theta,\vartheta}(v_{t}^{(1)})}{\sum_{i=1}^{N}\gamma_{t,\theta}(v_{t}^{(i)})/q_{t,\theta,\vartheta}(v_{t}^{(i)})}\\
= & \pi_{\theta}(v^{(\mathbf{k})})\left[\prod_{t=1}^{T}\prod_{i\neq k_{t}}^{M}q_{t,\theta,\vartheta}(v_{t}^{(i)})\right]\prod_{t=1}^{T}\frac{\gamma_{t,\theta}(v_{t}^{(1)})/q_{t,\theta,\vartheta}(v_{t}^{(1)})}{\sum_{i=1}^{N}\gamma_{t,\theta}(v_{t}^{(i)})/q_{t,\theta,\vartheta}(v_{t}^{(i)})}\\
= & (\pi_{\theta}\otimes\Phi_{\theta,\vartheta})^{\mathfrak{s}_{\mathbf{1},\mathbf{k}}}(\mathrm{d}\mathbf{v})b_{\theta,\vartheta}^{\textup{ref,}(1)}(\mathbf{1}|\mathbf{v})\\
= & (\pi_{\theta}\otimes\Phi_{\theta,\vartheta})^{\mathfrak{s}_{\mathbf{1},\mathbf{k}}}(\mathrm{d}\mathbf{v})b_{\theta,\vartheta}^{\textup{ref,}(1)}(\mathbf{k}|\mathfrak{s}_{\mathbf{1},\mathbf{k}}(\mathbf{v})),
\end{align*}
hence \eqref{eq: cSMC with eta1} is shown. We prove \eqref{eq: cSMC with eta2}
using the same steps above, replacing $\Phi_{\theta,\vartheta}(\cdot)$,
$q_{t,\theta,\vartheta}(\cdot)$, and $b_{\theta,\vartheta}^{\textup{ref,}(1)}(\cdot)$
by $\Phi_{\vartheta,\theta}(\cdot),q_{t,\vartheta,\theta}(\cdot)$,
and $b_{\theta,\vartheta}^{\textup{ref,}(2)}$, respectively.
\end{proof}
The following lemma can be verified by inspection, hence we skip a
formal proof.
\begin{lem}
\label{lem: facts for MHAAR-RB - 2}Suppose $\gamma_{t,\theta,\vartheta}=\gamma_{t,\theta}$,
for all $t\geq1$, and $\theta,\vartheta\in\Theta$. Then, for any
$\mathbf{v}\in\mathsf{Z}^{MT}$, and $\mathbf{k},\mathbf{l}\in\left\llbracket M\right\rrbracket {}^{T}$,
we have $r_{\mathbf{k},\mathbf{v}}(\theta,\vartheta)=r_{\mathbf{l},\mathbf{v}}(\theta,\vartheta)$.
\end{lem}
We proceed to the proof of Theorem \ref{thm: acceptance ratio of delayed rejection}.
\begin{proof}[Proof of Theorem \ref{thm: acceptance ratio of delayed rejection}]
 First, we show that the first ratio in \eqref{eq: acceptance probability of delayed rejection - latent variable}
is equal to $1$. Indeed, for $c=1$,
\begin{align*}
\frac{\check{\pi}^{\check{\varphi}_{2}}(\mathrm{d}\check{\xi})}{\check{\pi}(\mathrm{d}\check{\xi})}= & r(\theta,\vartheta)\frac{(\pi_{\theta}\otimes\Phi_{\theta,\vartheta})^{\mathfrak{s}_{\mathbf{1},\mathbf{l}}}(\mathrm{d}\mathbf{v})b_{\theta,\vartheta}^{(1)}(\mathfrak{r}_{\mathbf{\mathbf{l}}}(\mathbf{k})|\mathfrak{s}_{\mathbf{1},\mathbf{l}}(\mathbf{v}))b_{\theta,\vartheta}^{\textup{ref,}(1)}(\mathbf{l}|\mathfrak{s}_{\mathbf{1},\mathbf{l}}(\mathbf{v}))b_{\vartheta,\theta}^{\textup{ref,}(2)}(\mathfrak{r}_{\mathbf{\mathbf{l}}}(\mathbf{l}')|\mathfrak{s}_{\mathbf{1},\mathbf{l}}(\mathbf{v}))}{(\pi_{\theta}\otimes\Phi_{\theta,\vartheta})(\mathrm{d}\mathbf{v})b_{\theta,\vartheta}^{(1)}(\mathbf{k}|\mathbf{v})b_{\theta,\vartheta}^{\textup{ref,}(1)}(\mathbf{l}|\mathbf{v})b_{\vartheta,\theta}^{\textup{ref,}(2)}(\mathbf{l}'|\mathbf{v})}\\
= & r(\theta,\vartheta)\frac{(\pi_{\theta}\otimes\Phi_{\theta,\vartheta})^{\mathfrak{s}_{\mathbf{1},\mathbf{l}}}(\mathrm{d}\mathbf{v})b_{\theta,\vartheta}^{\textup{ref,}(1)}(\mathbf{l}|\mathfrak{s}_{\mathbf{1},\mathbf{l}}(\mathbf{v}))}{(\pi_{\theta}\otimes\Phi_{\theta,\vartheta})(\mathrm{d}\mathbf{v})b_{\theta,\vartheta}^{\textup{ref,}(1)}(\mathbf{l}|\mathbf{v})}\frac{b_{\theta,\vartheta}^{(1)}(\mathfrak{r}_{\mathbf{\mathbf{l}}}(\mathbf{k})|\mathfrak{s}_{\mathbf{1},\mathbf{l}}(\mathbf{v}))}{b_{\theta,\vartheta}^{(1)}(\mathbf{k}|\mathbf{v})}\frac{b_{\vartheta,\theta}^{\textup{ref,}(2)}(\mathfrak{r}_{\mathbf{\mathbf{l}}}(\mathbf{l}')|\mathfrak{s}_{\mathbf{1},\mathbf{l}}(\mathbf{v}))}{b_{\vartheta,\theta}^{\textup{ref,}(2)}(\mathbf{l}'|\mathbf{v})},
\end{align*}
and all of the ratios are equal to 1, due to Lemma \ref{lem: facts for MHAAR-RB - 1}.
For $c=2$, 
\begin{align*}
\frac{\check{\pi}^{\check{\varphi}_{2}}(\mathrm{d}\check{\xi})}{\check{\pi}(\mathrm{d}\check{\xi})}= & r(\theta,\vartheta)\frac{(\pi_{\theta}\otimes\Phi_{\vartheta,\theta})^{\mathfrak{s}_{\mathbf{1},\mathbf{l}}}(\mathrm{d}\mathbf{v})b_{\vartheta,\theta}(\mathfrak{r}_{\mathbf{\mathbf{l}}}(\mathbf{k})|\mathfrak{s}_{\mathbf{1},\mathbf{l}}(\mathbf{v}))b_{\theta,\vartheta}^{\textup{ref,}(2)}(\mathbf{l}|\mathfrak{s}_{\mathbf{1},\mathbf{l}}(\mathbf{v}))b_{\vartheta,\theta}^{\textup{ref,}(1)}(\mathfrak{r}_{\mathbf{\mathbf{l}}}(\mathbf{l}')|\mathfrak{s}_{\mathbf{1},\mathbf{l}}(\mathbf{v}))}{(\pi_{\theta}\otimes\Phi_{\vartheta,\theta})(\mathrm{d}\mathbf{v})b_{\theta,\vartheta}(\mathbf{k}|\mathbf{v})b_{\theta,\vartheta}^{\textup{ref,}(2)}(\mathbf{l}|\mathbf{v})b_{\vartheta,\theta}^{\textup{ref,}(1)}(\mathbf{l}'|\mathbf{v})}\\
= & r(\theta,\vartheta)\frac{(\pi_{\theta}\otimes\Phi_{\vartheta,\theta})^{\mathfrak{s}_{\mathbf{1},\mathbf{l}}}(\mathrm{d}\mathbf{v})b_{\theta,\vartheta}^{\textup{ref,}(2)}(\mathbf{l}|\mathfrak{s}_{\mathbf{1},\mathbf{l}}(\mathbf{v}))}{(\pi_{\theta}\otimes\Phi_{\vartheta,\theta})(\mathrm{d}\mathbf{v})b_{\theta,\vartheta}^{\textup{ref,}(2)}(\mathbf{l}|\mathbf{v})}\frac{b_{\vartheta,\theta}(\mathfrak{r}_{\mathbf{\mathbf{l}}}(\mathbf{k})|\mathfrak{s}_{\mathbf{1},\mathbf{l}}(\mathbf{v}))}{b_{\vartheta,\theta}(\mathbf{k}|\mathbf{v})}\frac{b_{\vartheta,\theta}^{\textup{ref,}(1)}(\mathfrak{r}_{\mathbf{\mathbf{l}}}(\mathbf{l}')|\mathfrak{s}_{\mathbf{1},\mathbf{l}}(\mathbf{v}))}{b_{\vartheta,\theta}^{\textup{ref,}(1)}(\mathbf{l}'|\mathbf{v})},
\end{align*}
and all of the ratios are equal to $1$, again, due to Lemma \ref{lem: facts for MHAAR-RB - 1}.

The second ratio in \eqref{eq: acceptance probability of delayed rejection - latent variable}
for any choice of $\gamma_{t,\theta,\vartheta}$ is equal to as is
equal to $1$ when $\gamma_{t,\theta,\vartheta}=\gamma_{t,\theta}$.
we can write the ratio of rejection probabilities ,
\begin{align*}
\frac{1-\min\left\{ 1,\check{r}_{1}\circ\check{\varphi}_{2}(\check{\xi})\right\} }{1-\min\left\{ 1,\check{r}_{1}(\check{\xi})\right\} }= & \frac{1-\min\left\{ 1,\mathring{r}(\theta,\vartheta,\mathfrak{s}_{\mathbf{1},\mathbf{l}}(\mathbf{v}),\mathfrak{r}_{\mathbf{\mathbf{l}}}(\mathbf{k}),c)\right\} }{1-\min\left\{ 1,\mathring{r}(\theta,\vartheta,\mathbf{v},\mathbf{k},c)\right\} }.\\
= & \begin{cases}
\frac{1-\min\left\{ 0,r_{\mathbf{l},\mathbf{v}}(\theta,\vartheta)\right\} }{1-\min\left\{ 0,r_{\mathbf{1},\mathbf{v}}(\theta,\vartheta)\right\} }, & c=1;\\
\frac{1-\min\left\{ 0,1/r_{\mathbf{k},\mathbf{v}}(\vartheta,\theta)\right\} }{1-\min\left\{ 0,1/r_{\mathbf{1},\mathbf{v}}(\vartheta,\theta)\right\} }, & c=2.
\end{cases}
\end{align*}
where we use \eqref{eq: equality of acceptance ratio of delayed rejection - multiple latent variable}
in the first line and the second line is due to Lemma \ref{lem: Invariance of r to permutations}.
When $\gamma_{t,\theta,\vartheta}=\gamma_{t,\theta}$, we use Lemma
\ref{lem: facts for MHAAR-RB - 2} to conclude that both ratios for
$c=1$ and $c=2$ simplify to $1$.
\end{proof}

\section{Auxiliary results and proofs Section \ref{sec: State-space models: SMC and conditional SMC within MHAAR}
\label{sec: Auxiliary results and proofs for cSMC based algorithms }}

First, we lay out some useful results on SMC, cSMC, for the state-space
model defined in Section \ref{subsec: State-space models and conditional SMC}. 

It is standard that the law of a particle filter with $M$ particles
and multinomial resampling for $\theta\in\Theta$, $\mathbf{v}\in\mathsf{Z}^{MT}$
and ancestral indices $a\in\left\llbracket M\right\rrbracket {}^{M(T-1)}$
is \citep{Andrieu_et_al_2010}
\begin{align*}
\psi_{\theta}\big({\rm d}(\mathbf{v},a)\big)= & \prod_{i=1}^{M}f_{\theta}({\rm d}v_{1}^{(i)})\prod_{t=2}^{T}\left\{ \prod_{i=1}^{M}\frac{w_{t-1,\theta}(v_{t-1}^{(a_{t-1}^{(i)})})}{\sum_{j=1}^{M}w_{t-1,\theta}(v_{t-1}^{(j)})}f_{\theta}(v_{t-1}^{(a_{t-1}^{(i)})},{\rm d}v_{t}^{(i)})\right\} .
\end{align*}
What is important for us is that the marginal distribution $\psi_{\theta}({\rm d}\mathbf{v})$
has a simple form 
\[
\psi_{\theta}({\rm d}\mathbf{v})=\prod_{i=1}^{M}f_{\theta}({\rm d}v_{1}^{(i)})\prod_{t=2}^{T}\left\{ \prod_{i=1}^{M}\frac{\sum_{j=1}^{M}w_{t-1,\theta}(v_{t-1}^{(j)})f_{\theta}(v_{t-1}^{(j)},{\rm d}v_{t}^{(i)})}{\sum_{j=1}^{M}w_{t-1,\theta}(v_{t-1}^{(j)})}\right\} .
\]
Now, letting $C_{\theta}:=\ell_{\theta}(y)$ in \eqref{eq:likelihoodSSM}
(recall $y=y_{1:T}$), and its estimator $\hat{C}_{\theta}(\mathbf{v}):=\prod_{t=1}^{T}\frac{1}{M}\sum_{i=1}^{M}w_{t,\theta}(v_{t}^{(i)})$,
we introduce 
\begin{equation}
\bar{\psi}_{\theta}({\rm d}\mathbf{v}):=\psi_{\theta}\big({\rm d}\mathbf{v}\big)\frac{\hat{C}_{\theta}(\mathbf{v})}{C_{\theta}}.\label{eq: SMC to cSMC probability law}
\end{equation}
We know from \citet{Andrieu_et_al_2010} that this is a probability
distribution, and is a way of justifying that $\hat{C}_{\theta}(v)$
is an unbiased estimator of $C_{\theta}$\textendash note that the
ancestral history is here integrated out. 

\begin{algorithm}[!h]
\caption{$\mathrm{cSMC}\big(M,\theta,z\big)$}
\label{alg: Conditional SMC}

\KwIn{Number of particles $M$, parameter $\theta$, current sample
$z$}

\KwOut{Particles $\mathbf{v}=v_{1:T}^{(1:M)}$, new sample $z'$}
Set $v_{1}^{(1)}=z_{1}$.\\
 \For{ $i=2,\ldots,M$}{ Sample $v_{1}^{(i)}\sim f_{\theta}(\cdot)$.\\
 Compute $w_{1}^{(i)}=g_{\theta}\big(v_{1}^{(i)},y_{1}\big)$. }
\For{ $t=2,\ldots,T$}{ Set $v_{t}^{(1)}=z_{t}$.\\
 \For{ $i=2,\ldots,M$}{ Sample $a_{t-1}^{(i)}\sim\mathcal{P}\big(w_{t-1}^{(1)},\ldots,w_{t-1}^{(M)}\big)$
and $v_{t}^{(i)}\sim f_{\theta}\big(v_{t-1}^{(a_{t-1}^{(i)})},\cdot\big)$.\\
 Compute $w_{t}^{(i)}=g_{\theta}\big(v_{t}^{(i)},y_{t}\big)$. }
} Sample $k_{T}\sim\mathcal{P}\big(w_{T}^{(1)},\ldots,w_{T}^{(M)}\big)$
and set $z'_{T}=v_{T}^{(k_{T})}$. \label{line:beginBS} \\
 \For{ $t=T-1,\ldots,1$}{ \For{ $i=1,\ldots,M$}{ Compute
$\tilde{w}_{t}^{(i)}=w_{t}^{(i)}f_{\theta}\big(v_{t}^{(i)},v_{t+1}^{(k_{t+1})}\big)$.
} Sample $k_{t}\sim\mathcal{P}\big(\tilde{w}_{t}^{(1)},\ldots,\tilde{w}_{t}^{(M)}\big)$
and set $z'_{t}=v_{t}^{(k_{t})}$. } \label{line:BSend} \Return$\mathbf{v}=v_{1:T}^{(1:N)}$
and $z'=z'_{1:T}$.
\end{algorithm}

The cSMC algorithm is given in Algorithm \ref{alg: Conditional SMC}.
The joint distribution of $\mathbf{v}\in\mathsf{Z}^{MT}$ when $v^{(\mathbf{1})}\sim\pi_{\theta}(\cdot)$
and $v^{(\bar{\mathbf{1}})}$ is sampled by the cSMC kernel targetting
$\pi_{\vartheta}$ can be written as

\begin{equation}
(\pi_{\theta}\otimes\Phi_{\vartheta})({\rm d}\mathbf{v}):=\pi_{\theta}({\rm d}v^{(\mathbf{1})})\prod_{i=2}^{M}f_{\vartheta}({\rm d}v_{1}^{(i)})\prod_{t=2}^{T}\left\{ \prod_{i=2}^{M}\frac{\sum_{j=1}^{M}w_{t-1,\vartheta}(v_{t-1}^{(j)})f_{\vartheta}(v_{t-1}^{(j)},{\rm d}v_{t}^{(i)})}{\sum_{j=1}^{M}w_{t-1,\vartheta}(v_{t-1}^{(j)})}\right\} .\label{eq: cSMC joint distribution}
\end{equation}
Recall the law of the indices used in the backward-sampling procedure
in order to draw a path $v^{(\mathbf{k})}$, 
\[
b_{\theta}(\mathbf{k}|\mathbf{v}):=\frac{w_{T,\theta}(v_{T}^{(k_{T})})}{\sum_{i=1}^{M}w_{T,\theta}(v_{T}^{(i)})}\prod_{t=2}^{T}\frac{w_{t-1,\theta}(v_{t-1}^{(k_{t-1})})f_{\theta}(v_{t-1}^{(k_{t-1})},v_{t}^{(k_{t})})}{\sum_{i=1}^{M}w_{t-1,\theta}(v_{t-1}^{(i)})f_{\theta}(v_{t-1}^{(i)},v_{t}^{(k_{t})})}.
\]

\begin{lem}
\label{lem: cSMC semi-reversibility} For any $\theta\in\Theta$ and
$\mathbf{v}\in\mathsf{Z}^{MT}$,
\[
(\pi_{\theta}\otimes\Phi_{\theta})(\mathrm{d}\mathbf{v})=M^{T}\bar{\psi}_{\theta}\big({\rm d}\mathbf{v}\big)b_{\theta}(\mathbf{1}|\mathbf{v}).
\]
\end{lem}
\begin{proof}[Proof of Lemma \ref{lem: cSMC semi-reversibility}]
 We check that the ratio
\begin{align*}
\frac{(\pi_{\theta}\otimes\Phi_{\theta})({\rm d}\mathbf{v})}{\bar{\psi}_{\theta}({\rm d}\mathbf{v})}= & \frac{C_{\theta}}{\hat{C}_{\theta}(\mathbf{v})}\frac{\pi_{\theta}({\rm d}v^{(\mathbf{1})})\prod_{i=2}^{M}f_{\theta}({\rm d}v_{1}^{(i)})\prod_{t=2}^{T}\left\{ \prod_{i=2}^{M}\frac{\sum_{j=1}^{M}w_{t-1,\theta}(v_{t-1}^{(j)})f_{\theta}(v_{t-1}^{(j)},{\rm d}v_{t}^{(i)})}{\sum_{j=1}^{M}w_{t-1,\theta}(v_{t-1}^{(j)})}\right\} }{\prod_{i=1}^{M}f_{\theta}({\rm d}v_{1}^{(i)})\prod_{t=2}^{T}\left\{ \prod_{i=1}^{M}\frac{\sum_{j=1}^{M}w_{t-1,\theta}(v_{t-1}^{(j)})f_{\theta}(v_{t-1}^{(j)},{\rm d}v_{t}^{(i)})}{\sum_{j=1}^{M}w_{t-1,\theta}(v_{t-1}^{(j)})}\right\} }\\
= & \frac{C_{\theta}}{\prod_{t=1}^{T}\frac{1}{M}\sum_{i=1}^{M}w_{\theta,t}(v_{t}^{(i)})}\frac{\pi_{\theta}({\rm d}v^{(\mathbf{1})})}{f_{\theta}(\mathrm{d}v_{1}^{(1)})}\frac{\prod_{t=2}^{T}\sum_{j=1}^{M}w_{t-1,\theta}(v_{t-1}^{(j)})}{\prod_{t=2}^{T}\sum_{j=1}^{M}w_{t-1,\theta}(v_{t-1}^{(j)})f_{\theta}(v_{t-1}^{(j)},{\rm d}v_{t}^{(1)})}\\
= & M^{T}\frac{w_{1,\theta}(v_{1}^{(1)})\prod_{t=2}^{T}f_{\theta}(v_{t-1}^{(1)},\mathrm{d}v_{t}^{(1)})w_{t,\theta}(v_{t}^{(1)})}{\sum_{i=1}^{M}w_{1,\theta}(v_{1}^{(1)})\prod_{t=2}^{T}\sum_{j=1}^{M}w_{t-1,\theta}(v_{t-1}^{(j)})f_{\theta}(v_{t-1}^{(j)},{\rm d}v_{t}^{(1)})}\\
= & M^{T}\frac{w_{1,\theta}(v_{1}^{(1)})}{\sum_{i=1}^{M}w_{1,\theta}(v_{1}^{(1)})}\prod_{t=2}^{T}\frac{f_{\theta}(v_{t-1}^{(1)},\mathrm{d}v_{t}^{(1)})w_{t,\theta}(v_{t}^{(1)})}{\sum_{j=1}^{M}w_{t-1,\theta}(v_{t-1}^{(j)})f_{\theta}(v_{t-1}^{(j)},{\rm d}v_{t}^{(1)})}\\
= & M^{T}b_{\theta}(\mathbf{1}|\mathbf{v}),
\end{align*}
as claimed.
\end{proof}
The constant $M^{T}$ on the right hand side arises from deterministic
assignment of indices $\mathbf{1}=(1,\ldots,1)$ for the conditioned
path $v^{(\mathbf{1})}$ in the cSMC algorithm. Lemmas \ref{lem: invariance of joint dist of v - SMC}
and \ref{lem: manupilation of selection probability} can be verified
by inspection.
\begin{lem}
\label{lem: invariance of joint dist of v - SMC}For any $\theta\in\Theta$,
$\mathbf{v}\in\mathsf{Z}^{MT}$, and $\mathbf{k}\in\left\llbracket M\right\rrbracket {}^{T}$,
and with $\mathfrak{s}_{\mathbf{1},\mathbf{k}}$ defined in \eqref{eq: swapping operator for MHAAR-RB},
we have $\bar{\psi}_{\theta}^{\mathfrak{s}_{\mathbf{1},\mathbf{k}}}\big({\rm d}\mathbf{v}\big)=\bar{\psi}_{\theta}\big({\rm d}\mathbf{v}\big)$.
\end{lem}
\begin{lem}
\label{lem: manupilation of selection probability}For any $\theta\in\Theta$,
$\mathbf{v}\in\mathsf{Z}^{MT}$, and $\mathbf{k}\in\left\llbracket M\right\rrbracket {}^{T}$,
we have\textup{ $b_{\theta}(\mathbf{k}|\mathfrak{s}_{\mathbf{1},\mathbf{k}}(\mathbf{v}))=b_{\theta}(\mathbf{1}|\mathbf{v})$.} 
\end{lem}
Lemmas \ref{lem: cSMC semi-reversibility}, \ref{lem: invariance of joint dist of v - SMC},
and \ref{lem: manupilation of selection probability} lead to the
following corollaries which will be useful in the subsequent proofs. 
\begin{cor}
\label{cor: cSMC semi-reversibility}For any $\theta\in\Theta$, $\mathbf{v}\in\mathsf{Z}^{MT}$,
and $\mathbf{k}\in\left\llbracket M\right\rrbracket {}^{T}$,
\[
(\pi_{\theta}\otimes\Phi_{\theta})(\mathrm{d}\mathbf{v})=M^{T}\bar{\psi}_{\theta}^{\mathfrak{s}_{\mathbf{1},\mathbf{k}}}\big({\rm d}\mathbf{v}\big)b_{\theta}(\mathbf{k}|\mathfrak{s}_{\mathbf{1},\mathbf{k}}(\mathbf{v}))
\]
\end{cor}
\begin{proof}[Proof of Corollary \ref{cor: cSMC semi-reversibility}]
 The corollary is a direct consequence of Lemmas \ref{lem: cSMC semi-reversibility},
\ref{lem: invariance of joint dist of v - SMC}, and \ref{lem: manupilation of selection probability}.
\end{proof}
\begin{cor}
\label{cor: swapping and cSMC}For any $\theta\in\Theta$, $\mathbf{k}\in\left\llbracket M\right\rrbracket {}^{T}$
and $\mathbf{v}\in\mathsf{Z}^{MT}$,
\[
(\pi_{\theta}\otimes\Phi_{\theta})(\mathrm{d}\mathbf{v})b_{\theta}(\mathbf{k}|\mathbf{v})=(\pi_{\theta}\otimes\Phi_{\theta})^{\mathfrak{s}_{\mathbf{1},\mathbf{k}}}(\mathrm{d}\mathbf{v})b_{\theta}(\mathbf{k}|\mathfrak{s}_{\mathbf{1},\mathbf{k}}(\mathbf{v})).
\]
\end{cor}
\begin{proof}[Proof of Corollary \ref{cor: swapping and cSMC}]
 By Corollary \ref{cor: cSMC semi-reversibility}, we have 
\begin{align*}
(\pi_{\theta}\otimes\Phi_{\theta})(\mathrm{d}\mathbf{v})b_{\theta}(\mathbf{k}|\mathbf{v}) & =M^{T}\bar{\psi}_{\theta}^{\mathfrak{s}_{\mathbf{1},\mathbf{k}}}\big({\rm d}\mathbf{v}\big)b_{\theta}(\mathbf{k}|\mathfrak{s}_{\mathbf{1},\mathbf{k}}(\mathbf{v}))b_{\theta}(\mathbf{k}|\mathbf{v})
\end{align*}
Also, note that, we have 
\begin{align*}
\bar{\psi}_{\theta}^{\mathfrak{s}_{\mathbf{1},\mathbf{k}}}\big({\rm d}\mathbf{v}\big)b_{\theta}(\mathbf{k}|\mathbf{v})= & \bar{\psi}_{\theta}^{\mathfrak{s}_{\mathbf{1},\mathbf{k}}}\big({\rm d}\mathbf{v}\big)b_{\theta}(\mathbf{1}|\mathfrak{s}_{\mathbf{1},\mathbf{k}}(\mathbf{v}))\\
= & (\pi_{\theta}\otimes\Phi_{\theta})^{\mathfrak{s}_{\mathbf{1},\mathbf{k}}}(\mathrm{d}\mathbf{v}),
\end{align*}
where the second line follows from Lemma \ref{lem: cSMC semi-reversibility}.
Substituting the latter equation into the former, we conclude.
\end{proof}
In addition to the results above, the following lemma will be useful
in Section \ref{subsec: Proof of unbiasedness for the Rao-Blackwellised estimator}.
\begin{lem}
\label{lem: change of measures}Let $F:\mathsf{Z}\rightarrow\mathbb{R}$
be a real-valued function. Then, for any $\theta\in\Theta$ we have
\[
\sum_{\mathbf{k}\in\left\llbracket M\right\rrbracket ^{T}}\int_{\mathsf{Z}^{MT}}b_{\theta}(\mathbf{k}|\mathbf{v})F(v^{\mathbf{(k)}})(\pi_{\theta}\otimes\Phi_{\theta})(\mathrm{d}\mathbf{v})=\sum_{\mathbf{k}\in\left\llbracket M\right\rrbracket ^{T}}\int_{\mathsf{Z}^{MT}}F(v^{\mathbf{(k)}})b_{\theta}(\mathbf{k}|\mathbf{v})\bar{\psi}_{\theta}(\mathrm{d}\mathbf{v}).
\]
\end{lem}
\begin{proof}[Proof of Lemma \ref{lem: change of measures}]
 First, notice that for any $\mathbf{l}\in\left\llbracket M\right\rrbracket {}^{T}$,
we can write 
\begin{equation}
\sum_{\mathbf{k}\in\left\llbracket M\right\rrbracket ^{T}}b_{\theta}(\mathbf{k}|\mathbf{v})F(v^{\mathbf{(k)}})=\sum_{\mathbf{k}\in\left\llbracket M\right\rrbracket ^{T}}b_{\theta}(\mathbf{k}|\mathfrak{s_{\mathbf{1},\mathbf{l}}(\mathbf{v})})F(\mathfrak{s}_{\mathbf{1},\mathbf{l}}(\mathbf{v)}^{(\mathbf{k})}).\label{eq: identity under permutation}
\end{equation}
due to one-to-one correspondence of the paths in $\mathbf{v}$ and
$\mathfrak{s}_{\mathbf{1},\mathbf{l}}(\mathbf{v})$. Combining this
with Corollary \ref{cor: cSMC semi-reversibility} applied with path
$\mathbf{l}$, we have
\begin{align*}
(\pi_{\theta}\otimes\Phi_{\theta})(\mathrm{d}\mathbf{v})\sum_{\mathbf{k}\in\left\llbracket M\right\rrbracket ^{T}}F(v^{\mathbf{(k)}})b_{\theta}(\mathbf{k}|\mathbf{v})\\
=M^{T}\bar{\psi}_{\theta}^{\mathfrak{s}_{\mathbf{1},\mathbf{l}}}(\mathrm{d}\mathbf{v}) & b_{\theta}(\mathbf{\mathbf{l}}|\mathfrak{s}_{\mathbf{1},\mathbf{l}}(\mathbf{v}))\sum_{\mathbf{k}\in\left\llbracket M\right\rrbracket ^{T}}F(\mathfrak{s}_{\mathbf{1},\mathbf{l}}(\mathbf{v)}^{(\mathbf{k})})b_{\theta}(\mathbf{k}|\mathfrak{s_{\mathbf{1},\mathbf{l}}(\mathbf{v})}).
\end{align*}
Since the above holds for every $\mathbf{l}\in\left\llbracket M\right\rrbracket ^{T}$,
summing over $\mathbf{l}\in\left\llbracket M\right\rrbracket ^{T}$
and dividing by $M^{T}$ results in
\begin{align*}
(\pi_{\theta}\otimes\Phi_{\theta})(\mathrm{d}\mathbf{v})\sum_{\mathbf{k}\in\left\llbracket M\right\rrbracket ^{T}}b_{\theta}(\mathbf{k}|\mathbf{v})F(v^{\mathbf{(k)}})\\
=\sum_{\mathbf{l}\in\left\llbracket M\right\rrbracket ^{T}} & \bar{\psi}_{\theta}^{\mathfrak{s}_{\mathbf{1},\mathbf{l}}}(\mathrm{d}\mathbf{v})b_{\theta}(\mathbf{\mathbf{l}}|\mathfrak{s}_{\mathbf{1},\mathbf{l}}(\mathbf{v}))\sum_{\mathbf{k}\in\left\llbracket M\right\rrbracket ^{T}}F(\mathfrak{s}_{\mathbf{1},\mathbf{l}}(\mathbf{v)}^{(\mathbf{k})})b_{\theta}(\mathbf{k}|\mathfrak{s_{\mathbf{1},\mathbf{l}}(\mathbf{v})})
\end{align*}
Taking the integral of both sides over $\mathbf{v}$, we get
\begin{align*}
\int_{\mathsf{Z}^{MT}}(\pi_{\theta}\otimes\Phi_{\theta})(\mathrm{d}\mathbf{v}) & \sum_{\mathbf{k}\in\left\llbracket M\right\rrbracket ^{T}}F(v^{\mathbf{(k)}})b_{\theta}(\mathbf{k}|\mathbf{v})\\
= & \sum_{\mathbf{l}\in\left\llbracket M\right\rrbracket ^{T}}\int_{\mathsf{Z}^{MT}}\bar{\psi}_{\theta}^{\mathfrak{s}_{\mathbf{1},\mathbf{l}}}(\mathrm{d}\mathbf{v})b_{\theta}(\mathbf{\mathbf{l}}|\mathfrak{s}_{\mathbf{1},\mathbf{l}}(\mathbf{v}))\sum_{\mathbf{k\in}\left\llbracket M\right\rrbracket ^{T}}b_{\theta}(\mathbf{k}|\mathfrak{s_{\mathbf{1},\mathbf{l}}(\mathbf{v})})F(\mathfrak{s}_{\mathbf{1},\mathbf{l}}(\mathbf{v)}^{(\mathbf{k})})\\
= & \sum_{\mathbf{l}\in\left\llbracket M\right\rrbracket ^{T}}\int_{\mathsf{Z}^{MT}}\bar{\psi}_{\theta}(\mathrm{d}\mathbf{v})b_{\theta}(\mathbf{\mathbf{l}}|\mathbf{v})\sum_{\mathbf{k}\in\left\llbracket M\right\rrbracket ^{T}}b_{\theta}(\mathbf{k}|\mathbf{v})F(\mathbf{v}^{(\mathbf{k})})\\
= & \int_{\mathsf{Z}^{MT}}\bar{\psi}_{\theta}(\mathrm{d}\mathbf{v})\sum_{\mathbf{k}\in\left\llbracket M\right\rrbracket ^{T}}b_{\theta}(\mathbf{k}|\mathbf{v})F(\mathbf{v}^{(\mathbf{k})})\sum_{\mathbf{l}\in\left\llbracket M\right\rrbracket ^{T}}b_{\theta}(\mathbf{\mathbf{l}}|\mathbf{v})\\
= & \sum_{\mathbf{k}\in\left\llbracket M\right\rrbracket ^{T}}\int_{\mathsf{Z}^{MT}}\bar{\psi}_{\theta}(\mathrm{d}\mathbf{v})b_{\theta}(\mathbf{k}|\mathbf{v})F(\mathbf{v}^{(\mathbf{k})})
\end{align*}
where in the first and third lines we apply a change in the order
of integration/summation, in the second line we apply a change of
variables $\mathbf{v}\rightarrow\mathfrak{s}_{\mathbf{1},\mathbf{l}}(\mathbf{v})$,
whose Jacobian is 1, and the last line follows since $b_{\theta}(\mathbf{l}|\mathbf{v})$
is a probability distribution for $\mathbf{l}$.
\end{proof}

\subsection{Unbiasedness for the acceptance ratio estimator of Algorithm \ref{alg: MHAAR-RB for SSM}
\label{subsec: Proof of unbiasedness for the Rao-Blackwellised estimator}}

We provide a proof of Theorem \ref{thm: SMC unbiased estimator of acceptance ratio}
that states the unbiasedness for the acceptance ratio estimator of
Algorithm \ref{alg: MHAAR-RB for SSM}.
\begin{proof}[Proof of Theorem \ref{thm: SMC unbiased estimator of acceptance ratio}]
 The expectation of $r_{\mathbf{1},\mathbf{v}}(\theta,\vartheta;\zeta)$
with respect to the law of the mechanism described in Theorem \ref{thm: SMC unbiased estimator of acceptance ratio}
is
\begin{align*}
\int_{\mathsf{Z}^{MT}}\left[\sum_{\mathbf{k}\in\left\llbracket M\right\rrbracket ^{T}}r_{v^{(\mathbf{1})},v^{(\mathbf{k})}}(\theta,\vartheta;\zeta)b_{\zeta}(\mathbf{k}|\mathbf{v})\right] & (\pi_{\theta}\otimes\Phi_{\zeta})(\mathrm{d}\mathbf{v})\\
=\int_{\mathsf{Z}^{MT}} & \left[\sum_{\mathbf{k}\in\left\llbracket M\right\rrbracket ^{T}}r_{v^{(\mathbf{1})},v^{(\mathbf{k})}}(\theta,\vartheta;\zeta)b_{\zeta}(\mathbf{k}|\mathbf{v})(\pi_{\theta}\otimes\Phi_{\zeta})(\mathrm{d}\mathbf{v})\right]
\end{align*}
Let $\gamma_{\theta}(z):=p_{\theta}(z,y)$ be the unnormalised density
for $\pi_{\theta}(z)$ so that $\gamma_{\theta}(z)=\pi_{\theta}(z)\ell_{\theta}(y)$.
Then, the term inside the sum on the RHS above can be written explicitly
as 
\begin{align*}
r_{v^{(\mathbf{1})},v^{(\mathbf{k})}}(\theta,\vartheta;\zeta)b_{\zeta}(\mathbf{k}|\mathbf{v})(\pi_{\theta}\otimes\Phi_{\zeta})(\mathrm{d}\mathbf{v})=\\
=\frac{\eta(\vartheta)q(\vartheta,\theta)}{\eta(\theta)q(\theta,\vartheta)} & \frac{\gamma_{\zeta}(v^{(\mathbf{1})})}{\gamma_{\theta}(v^{(\mathbf{1})})}\frac{\gamma_{\vartheta}(v^{(\mathbf{k})})}{\gamma_{\zeta}(v^{(\mathbf{k})})}b_{\zeta}(\mathbf{k}|\mathbf{v})\frac{\pi_{\theta}(v^{(\mathbf{1})})}{\pi_{\zeta}(v^{(\mathbf{1})})}(\pi_{\zeta}\otimes\Phi_{\zeta})(\mathrm{d}\mathbf{v})\\
=\frac{\eta(\vartheta)q(\vartheta,\theta)}{\eta(\theta)q(\theta,\vartheta)} & \frac{\ell_{\zeta}(y)}{\ell_{\theta}(y)}\frac{\gamma_{\vartheta}(v^{(\mathbf{k})})}{\gamma_{\zeta}(v^{(\mathbf{k})})}b_{\zeta}(\mathbf{k}|\mathbf{v})(\pi_{\zeta}\otimes\Phi_{\zeta})(\mathrm{d}\mathbf{v})
\end{align*}
Integrating both sides over $\mathbf{k}$ and $\mathbf{v}$, and applying
Lemma \ref{lem: change of measures}, we have
\begin{align}
\int_{\mathsf{Z}^{MT}}\sum_{\mathbf{k}\in\left\llbracket M\right\rrbracket ^{T}}r_{v^{(\mathbf{1})},v^{(\mathbf{k})}}(\theta,\vartheta;\zeta) & b_{\zeta}(\mathbf{k}|\mathbf{v})(\pi_{\theta}\otimes\Phi_{\zeta})(\mathrm{d}\mathbf{v})\nonumber \\
= & \frac{\eta(\vartheta)q(\vartheta,\theta)}{\eta(\theta)q(\theta,\vartheta)}\frac{\ell_{\zeta}(y)}{\ell_{\theta}(y)}\sum_{\mathbf{k}\in\left\llbracket M\right\rrbracket ^{T}}\int_{\mathsf{Z}^{MT}}\frac{\gamma_{\vartheta}(v^{(\mathbf{k})})}{\gamma_{\zeta}(v^{(\mathbf{k})})}b_{\zeta}(\mathbf{k}|\mathbf{v})\psi_{\zeta}(\mathrm{d}\mathbf{v})\label{eq: integral of the ratio of gammas}
\end{align}
Next, we will show that, the integral on the right hand side in \eqref{eq: integral of the ratio of gammas}
is $\frac{1}{M^{T}}\frac{\ell_{\vartheta}(y)}{\ell_{\zeta}(y)}$ for
every $\mathbf{k}\in\left\llbracket M\right\rrbracket ^{T}$. Indeed,
fixing $\mathbf{k}\in\left\llbracket M\right\rrbracket ^{T}$, we
have, using the symmetry of $\ensuremath{\bar{\psi}_{\zeta}}$, Lemma
\ref{lem: cSMC semi-reversibility}, c.o.v $\mathbf{v}\rightarrow\mathfrak{s}_{\mathbf{1},\mathbf{k}}(\mathbf{v})$
and $\ensuremath{\mathbf{v}\sim(\pi_{\zeta}\otimes\Phi_{\zeta})(\cdot)\Rightarrow v^{(\mathbf{1})}\sim\pi_{\zeta}}(\ensuremath{\cdot})$
\begin{align*}
\int_{\mathsf{Z}^{MT}}\frac{\gamma_{\vartheta}(v^{(\mathbf{k})})}{\gamma_{\zeta}(v^{(\mathbf{k})})}b_{\zeta}(\mathbf{k}|\mathbf{v})\bar{\psi}_{\zeta}(\mathrm{d}\mathbf{v})= & \int_{\mathsf{Z}^{MT}}\frac{\gamma_{\vartheta}(v^{(\mathbf{k})})}{\gamma_{\zeta}(v^{(\mathbf{k})})}b_{\zeta}(\mathbf{1}|\mathfrak{s}_{\mathbf{1},\mathbf{k}}(\mathbf{v}))\bar{\psi}_{\zeta}^{\mathfrak{s}_{\mathbf{1},\mathbf{k}}}(\mathrm{d}\mathbf{v})\\
= & \int_{\mathsf{Z}^{MT}}\frac{1}{M^{T}}\frac{\gamma_{\vartheta}(v^{(\mathbf{k})})}{\gamma_{\zeta}(v^{(\mathbf{k})})}(\pi_{\zeta}\otimes\Phi_{\zeta})^{\mathfrak{s}_{\mathbf{1},\mathbf{k}}}(\mathrm{d}\mathbf{v})\\
= & \int_{\mathsf{Z}^{MT}}\frac{1}{M^{T}}\frac{\gamma_{\vartheta}(v^{(\mathbf{1})})}{\gamma_{\zeta}(v^{(\mathbf{1})})}(\pi_{\zeta}\otimes\Phi_{\zeta})(\mathrm{d}\mathbf{v})\\
= & \int_{\mathsf{Z}^{MT}}\frac{1}{M^{T}}\frac{\gamma_{\vartheta}(v^{(\mathbf{1})})}{\gamma_{\zeta}(v^{(\mathbf{1})})}\pi_{\zeta}(\mathrm{d}v^{(\mathbf{1})})\\
= & \frac{1}{M^{T}}\frac{\ell_{\vartheta}(y)}{\ell_{\zeta}(y)},
\end{align*}
which does not depend on $\mathbf{k}$. Summing over $M^{T}$ possible
values of $\mathbf{\ensuremath{k}}$, and multiplying with the constant
ratio on the right hand side of in \eqref{eq: integral of the ratio of gammas},
we obtain the expectation as
\begin{align*}
\frac{\eta(\vartheta)q(\vartheta,\theta)}{\eta(\theta)q(\theta,\vartheta)}\frac{\ell_{\zeta}(y)}{\ell_{\theta}(y)}M^{T}\frac{1}{M^{T}}\frac{\ell_{\vartheta}(y)}{\ell_{\zeta}(y)}= & \frac{\eta(\vartheta)q(\vartheta,\theta)}{\eta(\theta)q(\theta,\vartheta)}\frac{\ell_{\vartheta}(y)}{\ell_{\theta}(y)}\\
= & r(\theta,\vartheta).
\end{align*}
as required.
\end{proof}

\subsection{Proofs for Algorithm \ref{alg: MHAAR-RB for SSM} \label{subsec: Proof of reversibility for Algorithms}}

\subsubsection{Acceptance ratio of Algorithm \ref{alg: MHAAR-RB for SSM}\label{subsec: Acceptance ratio of MHAAR-RB-SSM}}

The following lemmas can be verified by by inspection.
\begin{lem}
\label{lem: invariance of acceptance ratio to permutation - SSM}For
any $\theta,\vartheta,\zeta\in\Theta$, $\mathbf{v}\in\mathsf{Z}^{MT}$,
and $\mathbf{k}\in\left\llbracket M\right\rrbracket {}^{T}$, we have
\textup{$r_{\mathbf{k},\mathbf{v}}(\theta,\vartheta;\zeta)=r_{\mathbf{1},\mathfrak{s}_{\mathbf{1},\mathbf{k}}(\mathbf{v})}(\theta,\vartheta;\zeta)$}.
\end{lem}
\begin{lem}
\label{lem: independence of acceptance ratio to initial path}For
any $\mathbf{v}\in\mathsf{Z}^{MT}$, and $\mathbf{k},\mathbf{l}\in\left\llbracket M\right\rrbracket {}^{T}$,
we have $r_{\mathbf{k},\mathbf{v}}(\theta,\vartheta;\theta)=r_{\mathbf{l},\mathbf{v}}(\theta,\vartheta;\theta)$.
\end{lem}
We proceed to prove Theorem \ref{thm: acceptance ratio of MHAAR-RB-SSM}
that states the acceptance ratio of Algorithm \ref{alg: MHAAR-RB for SSM}.
\begin{proof}[Proof of Theorem \ref{thm: acceptance ratio of MHAAR-RB-SSM}]
 The joint distribution that corresponds to the moves of Algorithm
\ref{alg: MHAAR-RB for SSM} is
\begin{align*}
\mathring{\pi}(\mathrm{d}\xi)= & \frac{1}{2}\mathbb{I}_{1}(c)\pi(\mathrm{d}\theta)q(\theta,\mathrm{d}\vartheta)(\pi_{\theta}\otimes\Phi_{\zeta_{1}(\theta,\vartheta)})(\mathrm{d}\mathbf{v})\frac{r_{v^{(\mathbf{1})},v^{(\mathbf{k})}}(\theta,\vartheta;\zeta_{1}(\theta,\vartheta))b_{\zeta_{1}(\theta,\vartheta)}(\mathbf{k}|\mathbf{v})}{r_{\mathbf{1},\mathbf{v}}(\theta,\vartheta;\zeta_{1}(\theta,\vartheta))}\\
 & +\frac{1}{2}\mathbb{I}_{2}(c)\pi(\mathrm{d}\theta)q(\theta,\mathrm{d}\vartheta)(\pi_{\theta}\otimes\Phi_{\zeta_{2}(\theta,\vartheta)})(\mathrm{d}\mathbf{v})b_{\zeta_{2}(\theta,\vartheta)}(\mathbf{k}|\mathbf{v}).
\end{align*}
The proposed involution is $\varphi(\theta,\vartheta,\mathbf{v},\mathbf{k},c)=(\vartheta,\theta,\mathfrak{s}_{\mathbf{1},\mathbf{k}}(\mathbf{v}),\mathbf{k},3-c)$.
First, observe that, for any $z,z'\in\mathsf{Z}^{T}$, and $\theta,\vartheta,\zeta\in\Theta$,
equation \eqref{eq: AIS acceptance ratio for SSM} can be rewritten
as 
\begin{align}
\mathring{r}_{z,z'}(\theta,\vartheta;\zeta)= & \frac{q(\vartheta,\theta)}{q(\theta,\vartheta)}\frac{\eta(\vartheta)}{\eta(\theta)}\frac{\pi(\mathrm{d}(\vartheta,z'))}{\pi(\mathrm{d}(\zeta,z'))}\frac{\pi(\mathrm{d}(\zeta,z))}{\pi(\mathrm{d}(\theta,z))}\nonumber \\
= & r(\theta,\vartheta)\frac{\pi_{\vartheta}(\mathrm{d}z')}{\pi_{\zeta}(\mathrm{d}z')}\frac{\pi_{\zeta}(\mathrm{d}z)}{\pi_{\theta}(\mathrm{d}z)}.\label{eq: acceptance ratio modified}
\end{align}
where $r(\theta,\vartheta)$ is defined in \eqref{eq: marginal MCMC acceptance ratio - SSM}.
Letting $\zeta=\zeta_{1}(\theta,\vartheta)=\zeta_{2}(\vartheta,\theta)$,
when $c=1$, we arrive at the acceptance ratio

\begin{align}
\frac{\mathring{\pi}^{\varphi}(\mathrm{d}\xi)}{\mathring{\pi}(\mathrm{d}\xi)}= & r(\theta,\vartheta)\frac{(\pi_{\vartheta}\otimes\Phi_{\zeta})^{\mathfrak{s}_{\mathbf{1},\mathbf{k}}}(\mathrm{d}\mathbf{v})b_{\zeta}(\mathbf{k}|\mathfrak{s}_{\mathbf{1},\mathbf{k}}(\mathbf{v}))}{(\pi_{\theta}\otimes\Phi_{\zeta})(\mathrm{d}\mathbf{v})\frac{b_{\zeta}(\mathbf{k}|\mathbf{v})r_{v^{(\mathbf{1})},v^{(\mathbf{k})}}(\theta,\vartheta;\zeta)}{r_{\mathbf{1},\mathbf{v}}(\theta,\vartheta;\zeta)}}\nonumber \\
= & \frac{r(\theta,\vartheta)}{r_{v^{(\mathbf{1})},v^{(\mathbf{k})}}(\theta,\vartheta;\zeta)}\frac{(\pi_{\vartheta}\otimes\Phi_{\zeta})^{\mathfrak{s}_{\mathbf{1},\mathbf{k}}}(\mathrm{d}\mathbf{v})b_{\zeta}(\mathbf{k}|\mathfrak{s}_{\mathbf{1},\mathbf{k}}(\mathbf{v}))}{(\pi_{\theta}\otimes\Phi_{\zeta})(\mathrm{d}\mathbf{v})b_{\zeta}(\mathbf{k}|\mathbf{v})}r_{\mathbf{1},\mathbf{v}}(\theta,\vartheta;\zeta)\nonumber \\
= & \frac{r(\theta,\vartheta)}{r_{v^{(\mathbf{1})},v^{(\mathbf{k})}}(\theta,\vartheta;\zeta)}\frac{\frac{\pi_{\vartheta}(\mathrm{d}v^{(\mathbf{k})})}{\pi_{\zeta}(\mathrm{d}v^{(\mathbf{k})})}}{\frac{\pi_{\theta}(\mathrm{d}v^{(\mathbf{1})})}{\pi_{\zeta}(\mathrm{d}v^{(\mathbf{1})})}}\frac{(\pi_{\zeta}\otimes\Phi_{\zeta})^{\mathfrak{s}_{\mathbf{1},\mathbf{k}}}(\mathrm{d}\mathbf{v})b_{\zeta}(\mathbf{k}|\mathfrak{s}_{\mathbf{1},\mathbf{k}}(\mathbf{v}))}{(\pi_{\zeta}\otimes\Phi_{\zeta})(\mathrm{d}\mathbf{v})b_{\zeta}(\mathbf{k}|\mathbf{v})}r_{\mathbf{1},\mathbf{v}}(\theta,\vartheta;\zeta)\nonumber \\
= & r_{\mathbf{1},\mathbf{v}}(\theta,\vartheta;\zeta),\label{eq: RN derivative aMCMC for HMM all paths-1}
\end{align}
where for the last line we have used Corollary~\ref{cor: cSMC semi-reversibility}
and \eqref{eq: acceptance ratio modified}. For $c=2$, upon using
\eqref{eq:prop-r-circ-phi-inverse-r}, we write 
\begin{align*}
\mathring{r}(\theta,\vartheta,\mathbf{v},\mathbf{k},2)= & \mathring{r}(\vartheta,\theta,\mathfrak{s}_{\mathbf{1},\mathbf{k}}(\mathbf{v}),\mathbf{k},1)^{-1}\\
= & r_{\mathbf{1},\mathfrak{s}_{\mathbf{1},\mathbf{k}}(\mathbf{\mathbf{v}})}(\vartheta,\theta;\zeta_{1}(\vartheta,\theta))^{-1}\\
= & r_{\mathbf{k},\mathbf{\mathbf{v}}}(\vartheta,\theta;\zeta_{2}(\theta,\vartheta))^{-1},
\end{align*}
where the last line follows from Lemma \ref{lem: invariance of acceptance ratio to permutation - SSM},
which concludes the proof.
\end{proof}
The analysis in the proof above not only bears an alternative proof
of Theorem \ref{thm: SMC unbiased estimator of acceptance ratio}
on the unbiasedness of \eqref{eq: SMC acceptance ratio estimator all paths}
but also implicitly proves Corollary \ref{cor: SMC unbiased estimator of acceptance ratio};
as we show below.
\begin{proof}[Proof of Theorem \ref{thm: SMC unbiased estimator of acceptance ratio}]
 By the equality of the first and last lines of \eqref{eq: RN derivative aMCMC for HMM all paths-1},
for any $(\theta,\vartheta,\mathbf{v},\mathbf{k})\in\Theta^{2}\times\mathsf{Z}^{MT}\times\left\llbracket M\right\rrbracket ^{T}$,
we can write
\begin{equation}
(\pi_{\theta}\otimes\Phi_{\zeta})(\mathrm{d}\mathbf{v})\frac{r_{v^{(\mathbf{1})},v^{(\mathbf{k})}}(\theta,\vartheta;\zeta)b_{\zeta}(\mathbf{k}|\mathbf{v})}{\sum_{\mathbf{l}\in\left\llbracket M\right\rrbracket ^{T}}r_{v^{(\mathbf{1})},v^{(\mathbf{l})}}(\theta,\vartheta;\zeta)b_{\zeta}(\mathbf{l}|\mathbf{v})}\frac{r_{\mathbf{1},\mathbf{v}}(\theta,\vartheta;\zeta)}{r(\theta,\vartheta)}=(\pi_{\vartheta}\otimes\Phi_{\zeta})^{\mathfrak{s}_{\mathbf{1},\mathbf{k}}}(\mathrm{d}\mathbf{v})b_{\zeta}(\mathbf{k}|\mathfrak{s}_{\mathbf{1},\mathbf{k}}(\mathbf{v})).\label{eq: identity for alternative proof of Thm 5}
\end{equation}
Integrating both sides with respect to all the variables except $\theta$
and $\vartheta$ leads to 
\[
\int_{\mathsf{Z}^{MT}}r_{\mathbf{1},\mathbf{v}}(\theta,\vartheta;\zeta)(\pi_{\theta}\otimes\Phi_{\zeta})(\mathrm{d}\mathbf{v})=r(\theta,\vartheta)
\]
upon noticing that $r_{\mathbf{1},\mathbf{v}}(\theta,\vartheta;\zeta)$
does not depend on $\mathbf{k}$ and the right hand side is a probability
distribution for $(\mathbf{v},\mathbf{k})$. Recalling $z=v^{(\mathbf{1})}$
and noting that $(\pi_{\theta}\otimes\Phi_{\zeta})(\mathrm{d}\mathbf{v})$
is exactly the distribution of the mechanism described in Theorem
\ref{thm: SMC unbiased estimator of acceptance ratio} that generates
$r_{\mathbf{1},\mathbf{v}}(\theta,\vartheta;\zeta)$, we prove Theorem
\ref{thm: SMC unbiased estimator of acceptance ratio}.
\end{proof}
\begin{proof}[Proof of Corollary \ref{cor: SMC unbiased estimator of acceptance ratio}]
 Similarly to the previous proof, we make use of \eqref{eq: RN derivative aMCMC for HMM all paths-1}.
However, this time we write the identity in \eqref{eq: identity for alternative proof of Thm 5}
for $(\vartheta,\theta,\mathfrak{s}_{\mathbf{1},\mathbf{k}}(\mathbf{v}),\mathbf{k})$
to obtain

\[
(\pi_{\theta}\otimes\Phi_{\zeta})(\mathrm{d}\mathbf{v})b_{\zeta}(\mathbf{k}|\mathbf{v})\frac{r(\vartheta,\theta)}{r_{\mathbf{k},\mathbf{v}}(\vartheta,\theta;\zeta)}=(\pi_{\vartheta}\otimes\Phi_{\zeta})^{\mathfrak{s}_{\mathbf{1},\mathbf{k}}}(\mathrm{d}\mathbf{v})\frac{r_{v^{(\mathbf{k})},v^{(\mathbf{1})}}(\vartheta,\theta;\zeta)b_{\zeta}(\mathbf{k}|\mathfrak{s}_{\mathbf{1},\mathbf{k}}(\mathbf{v}))}{\sum_{\mathbf{l}\in\left\llbracket M\right\rrbracket ^{T}}r_{v^{(\mathbf{k})},\mathfrak{s}_{\mathbf{1},\mathbf{k}}(\mathbf{v})^{(\mathbf{l})}}(\vartheta,\theta;\zeta)b_{\zeta}(\mathbf{l}|\mathfrak{s}_{\mathbf{1},\mathbf{k}}(\mathbf{v}))}
\]
Again, integrating both sides with respect to all the variables $\mathbf{v},\mathbf{k}$,
we get 
\[
\int_{\mathsf{Z}^{MT}}\sum_{\mathbf{k}\in\left\llbracket M\right\rrbracket ^{T}}(1/r_{\mathbf{k},\mathbf{v}}(\vartheta,\theta;\zeta))b_{\zeta}(\mathbf{k}|\mathbf{v})(\pi_{\theta}\otimes\Phi_{\zeta})(\mathrm{d}\mathbf{v})=r(\theta,\vartheta).
\]
Since $1/r_{\mathbf{k},\mathbf{v}}(\vartheta,\theta;\zeta)$ is the
estimator in question in Corollary \ref{cor: SMC unbiased estimator of acceptance ratio}
and $(\pi_{\theta}\otimes\Phi_{\zeta})(\mathrm{d}\mathbf{v})b_{\zeta}(\mathbf{k}|\mathbf{v})$
is exactly the distribution of the described mechanism that generates
it, we prove Corollary~\ref{cor: SMC unbiased estimator of acceptance ratio}.
\end{proof}

\subsubsection{Delayed rejection step for Algorithm \ref{alg: MHAAR-RB for SSM}
\label{sec: Delayed rejection step for MHAAR-RB-SSM} }

When the delayed rejection step is included in Algorithm \ref{alg: MHAAR-RB for SSM},
the algorithm targets the modified joint distribution for $\check{\xi}=(\xi,\mathbf{l})$,
defined as
\begin{align*}
\check{\pi}(\mathrm{d}\check{\xi}) & =\mathring{\pi}(\mathrm{d}\xi)\left[\mathbb{I}_{1}(c)b_{\theta}(\mathbf{l}|\mathbf{v})+\mathbb{I}_{2}(c)b_{\vartheta}(\mathbf{l}|\mathfrak{s}_{\mathbf{1},\mathbf{k}}(\mathbf{v}))\right].
\end{align*}
The conditional probability of the extra variable $\mathbf{l}$ is
simply the backward sampling probability of the cSMC kernel run at
$\theta$. Now one iteration of the algorithm can be thought of as
a two-stage procedure, where the first stage is the regular MHAAR
update and the second stage is executed is conditional on the result
of the former. The two moves are given below. 
\begin{enumerate}
\item In the first stage, MHAAR attempts a transition for the joint variable
$\check{\xi}=(\xi,\mathbf{l})$ as 
\[
\check{\varphi}_{1}(\xi,\mathbf{l}):=(\varphi(\xi),\mathbf{l}).
\]
As $\tilde{\varphi}_{1}$ is an involution, it yields the acceptance
ratio as 
\begin{align*}
\check{r}_{1}(\check{\xi}):= & \frac{\check{\pi}^{\check{\varphi}_{1}}(\mathrm{d}\check{\xi})}{\check{\pi}(\mathrm{d}\check{\xi})}.\\
= & \frac{\mathring{\pi}^{\varphi}(\mathrm{d}\xi)}{\mathring{\pi}(\mathrm{d}\xi)}\frac{b_{\theta}(\mathbf{l}|\mathfrak{s}_{\mathbf{1},\mathbf{k}}\circ\mathfrak{s}_{\mathbf{1},\mathbf{k}}(\mathbf{v}))}{b_{\theta}(\mathbf{l}|\mathbf{v})}\\
= & \frac{\mathring{\pi}^{\varphi}(\mathrm{d}\xi)}{\mathring{\pi}(\mathrm{d}\xi)}=\mathring{r}(\xi)
\end{align*}
which is exactly the same acceptance ratio we would have for the basic
version of the algorithm that does not have the delayed rejection
step. As it can be seen from the above derivation, $\mathbf{l}$ does
not need to be sampled at this stage, i.e., prior to the delayed rejection
step, since the acceptance probability is independent of $\mathbf{l}$.
The delayed rejection step can be performed by the following involution. 
\item The proposed involution of delayed rejection is 
\[
\check{\varphi}_{2}(\xi\mathbf{,l}):=\begin{cases}
(\theta,\vartheta,\mathfrak{s}_{\mathbf{1},\mathbf{l}}(\mathbf{v}),\mathfrak{r}_{\mathbf{\mathbf{l}}}(\mathbf{k}),c,\mathbf{\mathbf{l}}), & c=1,\\
(\xi\mathbf{,l}), & c=2,
\end{cases}
\]
where $\mathfrak{r}_{\mathbf{l}}(\mathbf{k})$ is defined in equation
\eqref{eq: change index according to swap}. That is, we only perform
the delayed rejection move when $c=1$ and $\zeta_{1}(\theta,\vartheta)=\theta$
is chosen for the intermediate distribution. The acceptance ratio
of this move can be written as 
\begin{equation}
\check{r}_{2}(\check{\xi})=\frac{\check{\pi}^{\check{\varphi}_{2}}(\mathrm{d}\check{\xi})}{\check{\pi}(\mathrm{d}\check{\xi})}\frac{1-\min\left\{ 1,\check{r}_{1}\circ\check{\varphi}_{2}(\check{\xi})\right\} }{1-\min\left\{ 1,\check{r}_{1}(\check{\xi})\right\} }\label{eq: acceptance probability of delayed rejection - SSM}
\end{equation}
\end{enumerate}
\begin{thm}
\label{thm: acceptance ratio of delayed rejection equal to 1}Assume
$\zeta_{1}(\theta,\vartheta)=\theta$. Then, \textup{$\check{r}_{2}(\check{\xi})=1$}.
\end{thm}
\begin{proof}[Proof of Theorem \ref{thm: acceptance ratio of delayed rejection equal to 1}]
We prove the theorem by showing that both ratios in \eqref{eq: acceptance probability of delayed rejection - SSM}
are equal to $1$ if $\zeta_{1}(\theta,\vartheta)=\theta$. Assume
that $\zeta_{1}(\theta,\vartheta)=\theta$. When $c=1$, 
\begin{align*}
\frac{\check{\pi}^{\check{\varphi}_{2}}(\mathrm{d}\check{\xi})}{\check{\pi}(\mathrm{d}\check{\xi})}= & r(\theta,\vartheta)\frac{(\pi_{\theta}\otimes\Phi_{\theta})^{\mathfrak{s}_{\mathbf{1},\mathbf{l}}}(\mathrm{d}\mathbf{v})b_{\theta}(\mathfrak{r}_{\mathbf{\mathbf{l}}}(\mathbf{k})|\mathfrak{s}_{\mathbf{1},\mathbf{l}}(\mathbf{v}))b_{\theta}(\mathbf{l}|\mathfrak{s}_{\mathbf{1},\mathbf{l}}(\mathbf{v}))}{(\pi_{\theta}\otimes\Phi_{\theta})(\mathrm{d}\mathbf{v})b_{\theta}(\mathbf{k}|\mathbf{v})b_{\theta}(\mathbf{l}|\mathbf{v})}\\
= & r(\theta,\vartheta)\frac{(\pi_{\theta}\otimes\Phi_{\theta})^{\mathfrak{s}_{\mathbf{1},\mathbf{l}}}(\mathrm{d}\mathbf{v})b_{\theta}(\mathbf{l}|\mathfrak{s}_{\mathbf{1},\mathbf{l}}(\mathbf{v}))}{(\pi_{\theta}\otimes\Phi_{\theta})(\mathrm{d}\mathbf{v})b_{\theta}(\mathbf{l}|\mathbf{v})}\frac{b_{\theta}(\mathfrak{r}_{\mathbf{\mathbf{l}}}(\mathbf{k})|\mathfrak{s}_{\mathbf{1},\mathbf{l}}(\mathbf{v}))}{b_{\theta}(\mathbf{k}|\mathbf{v})},
\end{align*}
and all of the ratios are equal to 1. Moreover, the ratio involving
the rejection probabilities is
\begin{align*}
\frac{1-\min\left\{ 1,\check{r}_{1}\circ\check{\varphi}_{2}(\xi,\mathbf{l})\right\} }{1-\min\left\{ 1,\check{r}_{1}(\xi,\mathbf{l})\right\} }= & \frac{1-\min\left\{ 1,\mathring{r}(\theta,\vartheta,\mathfrak{s}_{\mathbf{1},\mathbf{l}}(\mathbf{v}),\mathfrak{r}_{\mathbf{\mathbf{l}}}(\mathbf{k}),1)\right\} }{1-\min\left\{ 1,\mathring{r}(\theta,\vartheta,\mathbf{v},\mathbf{k},1)\right\} }\\
= & \frac{1-\min\left\{ 1,r_{\mathbf{l},\mathbf{v}}(\theta,\vartheta)\right\} }{1-\min\left\{ 1,r_{\mathbf{1},\mathbf{v}}(\theta,\vartheta)\right\} }\\
= & 1,
\end{align*}
where the second line is by Lemma \ref{lem: invariance of acceptance ratio to permutation - SSM},
and the last line is by Lemma \ref{lem: independence of acceptance ratio to initial path}.
When $c=2$, the move $\check{\varphi}_{2}$ imposes no change, so
the acceptance ratio is trivially equal to $1$.
\end{proof}
Note that the conditions $c=1$ and $\zeta_{1}(\theta,\vartheta)=\theta$
are critical here: The proposed update of delayed rejection does not
change the sample for $\theta$ but changes the sample for $z$ via
backward sampling at $\theta$ conditional on the particles generated
by an cSMC kernel run at $\zeta_{1}(\theta,\vartheta)$. Accepting
this proposal with probability $1$ preserves invariance only if $\zeta_{1}(\theta,\vartheta)=\theta$.
We could, in theory, have a similar delayed rejection step when $c=2$
if $\zeta_{1}(\theta,\vartheta)=\vartheta$. However, the choice $\zeta_{1}(\theta,\vartheta)=\vartheta$
is senseless because it disables all the averaging in the MHAAR algorithm,
see equations \eqref{eq: AIS acceptance ratio for SSM} and \eqref{eq: SMC acceptance ratio estimator all paths}.

\subsection{The subsampled version of MHAAR-RB for SSM \label{subsec: Proof of reversibility for Algorithms-1}}

The subsampled version of MHAAR-RB-SSM, named MHAAR-S-SSM, which was
mentioned in Section \ref{subsec: Easing computational burden with subsampling: Multiple paths BS-SMC}
is presented in Algorithm \ref{alg: MHAAR-S for SSM}. Like in MHAAR-RB-SSM,
refreshing $z$ is also possible in Algorithm \ref{alg: MHAAR-S for SSM}
as well, but in a different fashion, see the step labeled as `optional'.
Specifically, when $c=1$ and $\zeta_{1}(\theta,\vartheta)=\theta$,
one can randomly swap $z$ with $u^{(i)}$ with a probability $1/N$
for all $i=1,\ldots,N$, owing to exchangeability arguments. Note
that this is not a delayed rejection step and the swapping has to
be performed before making a decision, as it affects the acceptance
ratio. However the computational cost of swapping two paths is negligible.
We explain why this move preserves invariance in Appendix \ref{subsec: Refreshing the latent variable in MHAAR-S-SSM}.

\begin{algorithm}[!h]
\caption{MHAAR-S for SSM - reduced computation via subsampling}
\label{alg: MHAAR-S for SSM}

\KwIn{Current sample $(\theta,z)$}

\KwOut{New sample}

Sample $\vartheta\sim q(\theta,\cdot)$ and $c\sim\text{Unif}(\left\{ 1,2\right\} )$,
and set $\zeta=\zeta_{c}(\theta,\vartheta)$. \\
\If{$c=1$}{

Run a ${\rm cSMC}(M,\zeta,z)$ to obtain the particles $\mathbf{v}$.\\
Sample $u^{(1)},\ldots,u^{(N)}\overset{{\rm iid}}{\sim}\sum_{\mathbf{l}\in\left\llbracket M\right\rrbracket ^{T}}\Phi_{\zeta}(\mathbf{l}|\mathbf{v})\delta_{v^{(\mathbf{l})}}(\cdot)$.\\
\If{$\zeta=\theta$}{

Swap $z$ with $u^{(j)}$ where $j\sim\text{Unif}(\left\llbracket N\right\rrbracket )$.
(optional refreshment of $z$)

}

Sample $k\sim\mathcal{P}\big(r_{z,u^{(1)}}(\theta,\vartheta;\zeta),\ldots,r_{z,u^{(N)}}(\theta,\vartheta;\zeta)\big)$
and set $z'=u^{(k)}$.\\
Return $(\vartheta,z')$ with probability $\min\{1,r_{z,\mathfrak{u}}^{N}(\theta,\vartheta;\zeta)\}$;
otherwise return $(\theta,z)$.

}\Else{

Run a ${\rm cSMC}(M,\zeta,z)$ to obtain particles $\mathbf{v}$.
\\
Sample $u^{(1)},\ldots,u^{(N)}\overset{{\rm iid}}{\sim}\sum_{\mathbf{l}\in\left\llbracket M\right\rrbracket ^{T}}\Phi_{\zeta}(\mathbf{l}|\mathbf{v})\delta_{v^{(\mathbf{l})}}(\cdot)$.\\
Sample $k\sim\text{Unif}(\left\llbracket N\right\rrbracket )$, set
$z'=u^{(k)}$, and change $u^{(k)}=z$.\\
Return $(\vartheta,z')$ with probability $\min\{1,1/r_{z',\mathfrak{u}}^{N}(\vartheta,\theta;\zeta)\}$;
otherwise return $(\theta,z)$.

}
\end{algorithm}

\subsubsection{Reversibility of Algorithm \ref{alg: MHAAR-S for SSM}\label{subsec: Reversibility of MHAAR-S-SSM}}

Next, we show the reversibility of Algorithm \ref{alg: MHAAR-S for SSM}
that uses a subsampled version of the Rao-Blackwellised acceptance
ratio estimator.

For any $\theta\in\Theta$, suppose $u^{(0)}\sim\pi_{\theta}(\cdot)$
and let $u^{(1)},\ldots,u^{(N)}$ be $N$ paths drawn via backward
sampling following cSMC at $\zeta$ conditioned on $u^{(0)}$. Then
the joint distribution of $\mathfrak{u}:=(u^{(0)},\ldots,u^{(N)})$
can be written as
\[
R_{\theta,\zeta}(\mathrm{d}\mathfrak{u})=\int_{\mathsf{Z}^{MT}}\left\{ \left[(\pi_{\theta}\otimes\Phi_{\zeta})\big({\rm d}\mathbf{v})\delta_{v^{(\mathbf{1})}}(\mathrm{d}u^{(0)})\right]\prod_{i=1}^{N}\left[\sum_{\mathbf{k}\in\left\llbracket M\right\rrbracket ^{T}}b_{\zeta}(\mathbf{k}|\mathbf{v})\delta_{v^{(\mathbf{k})}}(\mathrm{d}u^{(i)})\right]\right\} .
\]

\begin{lem}
\label{lem: exchangeability of cSMC}The following hold for \textup{$R_{\theta,\zeta}(\mathrm{d}(u^{(0)},\ldots,u^{(N)}))$}:
\begin{enumerate}
\item The marginal distribution of $u^{(0)}$ is $\pi_{\theta}(\cdot)$.
\item When $\theta=\zeta$, \textup{the variables $u^{(0)},u^{(1)},\ldots,u^{(N)}$}
are exchangeable and share $\pi_{\theta}(\cdot)$ as their marginal
distribution.
\item $R_{\theta,\zeta}(\mathrm{d}\mathfrak{u})=\frac{\pi_{\theta}(\mathrm{d}u^{(0)})}{\pi_{\zeta}(\mathrm{d}u^{(0)})}R_{\zeta,\zeta}(\mathrm{d}\mathfrak{u}).$
\end{enumerate}
\end{lem}
\begin{proof}[Proof of Lemma \ref{lem: exchangeability of cSMC}]
 The claims in the lemma can be proven by considering the joint distribution
\begin{equation}
\overline{R}_{\theta,\zeta}(\mathrm{d}(\mathfrak{u},\mathbf{v},\mathbf{k}_{0},\ldots,\mathbf{k}_{N})):=\frac{\pi_{\theta}(\mathrm{d}u^{(0)})}{\pi_{\zeta}(\mathrm{d}u^{(0)})}\bar{\psi}_{\zeta}\big({\rm d}\mathbf{v})\prod_{i=0}^{N}b_{\zeta}(\mathbf{k}_{i}|\mathbf{v})\delta_{v^{(\mathbf{k}_{i})}}(\mathrm{d}u^{(i)}).\label{eq: joint distribution for the exchangeable R-1}
\end{equation}
First, we show that $\overline{R}_{\theta,\zeta}(\mathrm{d}(\mathfrak{u},\mathbf{v},\mathbf{k}_{0},\ldots,\mathbf{k}_{N}))$
is a distribution whose marginal distribution for $\mathfrak{u}$
is $R_{\theta,\zeta}(\mathfrak{u})$. For this, first note the identity
\begin{equation}
\frac{\pi_{\theta}(\mathrm{d}u^{(0)})}{\pi_{\zeta}(\mathrm{d}u^{(0)})}\bar{\psi}_{\zeta}(\mathrm{d}\mathbf{v})b_{\zeta}(\mathbf{k}_{0}|\mathbf{v})\delta_{v^{(\mathbf{k}_{0})}}(\mathrm{d}u^{(0)})=\frac{\pi_{\theta}(\mathrm{d}u^{(0)})}{\pi_{\zeta}(\mathrm{d}u^{(0)})}\bar{\psi}_{\zeta}^{\mathfrak{s}_{\mathbf{1},\mathbf{k}_{0}}}(\mathrm{d}\mathbf{v})b_{\zeta}(\mathbf{1}|\mathfrak{s}_{\mathbf{1},\mathbf{k}_{0}}(\mathbf{v}))\delta_{v^{(\mathbf{1})}}(\mathrm{d}u^{(0)}),\label{eq: first identity for the joint distribution of barR}
\end{equation}
 which follows from Lemmas \ref{lem: invariance of joint dist of v - SMC}
and \ref{lem: manupilation of selection probability}. Next, integrating
the RHS of \eqref{eq: first identity for the joint distribution of barR}
with respect to $\mathbf{k}_{0}$ and $\mathbf{v}$, we obtain 
\begin{align}
\frac{\pi_{\theta}(\mathrm{d}u^{(0)})}{\pi_{\zeta}(\mathrm{d}u^{(0)})} & \sum_{\mathbf{k}_{0}\in\left\llbracket M\right\rrbracket ^{T}}\int_{\mathsf{Z}^{MT}}\bar{\psi}_{\zeta}^{\mathfrak{s}_{\mathbf{1},\mathbf{k}_{0}}}(\mathrm{d}\mathbf{v})b_{\zeta}(\mathbf{1}|\mathfrak{s}_{\mathbf{1},\mathbf{k}_{0}}(\mathbf{v}))\delta_{v^{(1)}}(\mathrm{d}u^{(0)})\nonumber \\
 & =\frac{\pi_{\theta}(\mathrm{d}u^{(0)})}{\pi_{\zeta}(\mathrm{d}u^{(0)})}\int_{\mathsf{Z}^{MT}}M^{T}\bar{\psi}_{\zeta}(\mathrm{d}\mathbf{v})b_{\zeta}(\mathbf{1}|\mathbf{v})\delta_{v^{(\mathbf{1})}}(\mathrm{d}u^{(0)})\nonumber \\
 & =\frac{\pi_{\theta}(\mathrm{d}u^{(0)})}{\pi_{\zeta}(\mathrm{d}u^{(0)})}\int_{\mathsf{Z}^{MT}}(\pi_{\zeta}\otimes\Phi_{\zeta})(\mathrm{d}\mathbf{v})\delta_{v^{(\mathbf{1})}}(\mathrm{d}u^{(0)})\nonumber \\
 & =\int_{\mathsf{Z}^{MT}}(\pi_{\theta}\otimes\Phi_{\zeta})(\mathrm{d}\mathbf{v})\delta_{v^{(\mathbf{1})}}(\mathrm{d}u^{(0)})\label{eq: substitute integrals in R}\\
 & =\pi_{\theta}(\mathrm{d}u^{(0)}),\label{eq: marginal of Rbar-1}
\end{align}
where in the second line we use a change of variable $\mathbf{v}\rightarrow\mathfrak{s}_{\mathbf{1},\mathbf{k}_{0}}(\mathbf{v})$
and end up with the same expression for all $\mathbf{k}_{0}$, and
the third line is by Lemma \ref{lem: cSMC semi-reversibility}. Using
\eqref{eq: substitute integrals in R} together with \eqref{eq: first identity for the joint distribution of barR},
we have
\begin{align*}
\int_{\mathrm{Z}^{MT}}\sum_{\mathbf{k}_{0:N}\in\left\llbracket M\right\rrbracket ^{TN}} & \overline{R}_{\theta,\zeta}(\mathrm{d}(\mathfrak{u},\mathbf{v},\mathbf{k}_{0},\ldots,\mathbf{k}_{N}))\\
 & =\int_{\mathsf{Z}^{MT}}(\pi_{\theta}\otimes\Phi_{\zeta})(\mathrm{d}\mathbf{v})\delta_{v^{(\mathbf{1})}}(\mathrm{d}u^{(0)})\sum_{\mathbf{k}_{1:N}\in\left\llbracket M\right\rrbracket ^{TN}}\prod_{i=1}^{N}\left[b_{\zeta}(\mathbf{k}_{i}|\mathbf{v})\delta_{v^{(\mathbf{k}_{i})}}(\mathrm{d}u^{(i)})\right]\\
 & =\int_{\mathsf{Z}^{MT}}\left\{ \left[(\pi_{\theta}\otimes\Phi_{\zeta})\big({\rm d}\mathbf{v})\delta_{v^{(\mathbf{1})}}(\mathrm{d}u^{(0)})\right]\prod_{i=1}^{N}\left[\sum_{\mathbf{k}\in\left\llbracket M\right\rrbracket ^{T}}b_{\zeta}(\mathbf{k}|\mathbf{v})\delta_{v^{(\mathbf{k})}}(\mathrm{d}u^{(i)})\right]\right\} \\
 & =R_{\theta,\zeta}(\mathrm{d}\mathfrak{u}).
\end{align*}
Now, we can proceed to proving the claims in the lemma. The first
claim can be proven by integrating \eqref{eq: joint distribution for the exchangeable R-1}
with respect to $\mathbf{k}_{1},\ldots,\mathbf{k}_{N},u^{(1)},\ldots,u^{(N)}$
and then with respect to $\mathbf{k}_{0}$ and $\mathbf{v}$, where
in the latter step we use \eqref{eq: marginal of Rbar-1}. For the
second claim, observe that when $\theta=\zeta$ we have 
\[
\overline{R}_{\theta,\theta}(\mathrm{d}(\mathfrak{u},\mathbf{v},\mathbf{k}_{0},\ldots,\mathbf{k}_{N})):=\bar{\psi}_{\theta}\big({\rm d}\mathbf{v})\prod_{i=0}^{N}b_{\theta}(\mathbf{k}_{i}|\mathbf{v})\delta_{v^{(\mathbf{k}_{i})}}(\mathrm{d}u^{(i)}).
\]
Taking the integral of both sides with respect to $\mathbf{v}$ and
$\mathbf{k}_{0},\ldots,\mathbf{k}_{N}$, we have 
\begin{equation}
R_{\theta,\theta}(\mathrm{d}\mathfrak{u})=\int_{\mathsf{Z}^{MT}}\bar{\psi}_{\theta}\big({\rm d}\mathbf{v})\prod_{i=0}^{N}\sum_{\mathbf{k}\in\left\llbracket M\right\rrbracket ^{T}}b_{\theta}(\mathbf{k}|\mathbf{v})\delta_{v^{(\mathbf{k})}}(\mathrm{d}u^{(i)})\label{eq: R seen as exchangeable}
\end{equation}
and the exchangeability of $u^{(0)},u^{(1)},\ldots,u^{(N)}$ is obvious
from the symmetry in \eqref{eq: R seen as exchangeable}. Moreover,
due to exchangeability, since $u^{(0)}$ has marginal $\pi_{\theta}(\mathrm{d}u^{(0)})$,
so do $u^{(1)},\ldots,u^{(N)}$. For the third claim, note the relation
\[
\overline{R}_{\theta,\zeta}(\mathrm{d}(\mathfrak{u},\mathbf{v},\mathbf{k}_{0},\ldots,\mathbf{k}_{N})):=\frac{\pi_{\theta}(\mathrm{d}u^{(0)})}{\pi_{\zeta}(\mathrm{d}u^{(0)})}\overline{R}_{\zeta,\zeta}(\mathrm{d}(\mathfrak{u},\mathbf{v},\mathbf{k}_{0},\ldots,\mathbf{k}_{N}))
\]
from \eqref{eq: joint distribution for the exchangeable R-1}. Taking
the integral of both sides over $\mathbf{v},\mathbf{k}_{0},\ldots,\mathbf{k}_{N}$,
we have the claimed equality.
\end{proof}
\begin{thm}
\label{thm: detailed balance for MHAAR-RB-subsampled}The transition
probability of Algorithm \ref{alg: MHAAR-S for SSM} satisfies detailed
balance with respect to $\pi(\mathrm{d}(\theta,z))$.
\end{thm}
\begin{proof}[Proof of Theorem \ref{thm: detailed balance for MHAAR-RB-subsampled}]
 The joint distribution corresponding to the moves of Algorithm \ref{alg: MHAAR-S for SSM}
can be shown to target the joint distribution for $\xi:=(\theta,\vartheta,\mathfrak{u},k,c)$,
defined as 
\begin{align*}
\mathring{\pi}\big({\rm d}\xi\big):= & \frac{1}{2}\mathbb{I}_{1}(c)\pi(\mathrm{d}\theta)q(\theta,{\rm d}\vartheta)R_{\theta,\zeta_{1}(\theta,\vartheta)}(\mathrm{d}\mathfrak{u})\frac{r_{u^{(0)},u^{(k)}}(\theta,\vartheta;\zeta_{1}(\theta,\vartheta))}{\sum_{i=1}^{N}r_{u^{(0)},u^{(i)}}(\theta,\vartheta;\zeta_{1}(\theta,\vartheta))}\\
+ & \frac{1}{2}\mathbb{I}_{2}(c)\pi(\mathrm{d}\theta)q(\theta,{\rm d}\vartheta)R_{\theta,\zeta_{2}(\theta,\vartheta)}(\mathrm{d}\mathfrak{u})\frac{1}{N}.
\end{align*}
where the latent variable is embedded in $\mathfrak{u}$ as $z=u^{(0)}$.
Then, Lemma \ref{lem: exchangeability of cSMC}, the marginal for
$(\theta,u^{(0)})$ is $\pi(x)$. The proposed involution is 
\[
\varphi(\theta,\vartheta,\mathfrak{u},k,c):=(\vartheta,\theta,\mathfrak{s}_{0,k}(\mathfrak{u}),k,3-c),
\]
where $\mathfrak{s}_{0,k}(\mathfrak{u})$ is an operator that swaps
$u^{(0)}$ and $u^{(k)}$ in $\mathfrak{u}.$ Next, we derive the
acceptance ratios
\[
\mathring{r}(\xi):=\frac{\mathring{\pi}^{\varphi}(\mathrm{d}\xi)}{\mathring{\pi}(\mathrm{d}\xi\big)}
\]
for $c=1$ and $c=2$. When $c=1$, we have 
\begin{align*}
\mathring{\pi}(\mathrm{d}\xi\big)= & \pi(\theta)q(\theta,{\rm d}\vartheta)\frac{\pi_{\theta}(\mathrm{d}u^{(0)})}{\pi_{\zeta_{1}(\theta,\vartheta)}(\mathrm{d}u^{(0)})}R_{\zeta_{1}(\theta,\vartheta),\zeta_{1}(\theta,\vartheta)}({\rm d}\mathfrak{u})\frac{r_{u^{(0)},u^{(k)}}(\theta,\vartheta;\zeta_{1}(\theta,\vartheta))}{\sum_{i=1}^{N}r_{u^{(0)},u^{(i)}}(\theta,\vartheta;\zeta_{1}(\theta,\vartheta))},\\
\mathring{\pi}^{\varphi}(\mathrm{d}\xi)= & \frac{1}{2}\pi(\mathrm{d}\vartheta)q(\vartheta,{\rm d}\theta)\frac{\pi_{\vartheta}(\mathrm{d}u^{(k)})}{\pi_{\zeta_{1}(\theta,\vartheta)}(\mathrm{d}u^{(k)})}R_{\zeta_{1}(\theta,\vartheta),\zeta_{1}(\theta,\vartheta)}^{\mathfrak{s}_{0,k}}({\rm d}\mathfrak{u})\frac{1}{N}\\
= & \frac{1}{2}\pi(\mathrm{d}\vartheta)q(\vartheta,{\rm d}\theta)\frac{\pi_{\vartheta}(\mathrm{d}u^{(k)})}{\pi_{\zeta_{1}(\theta,\vartheta)}(\mathrm{d}u^{(k)})}R_{\zeta_{1}(\theta,\vartheta),\zeta_{1}(\theta,\vartheta)}(\mathrm{d}\mathfrak{u})\frac{1}{N},
\end{align*}
where we have used Lemma \ref{lem: exchangeability of cSMC} in the
lines of both equations. Noting \eqref{eq: acceptance ratio modified},
we conclude that, for $c=1$, 
\[
\mathring{r}(\xi)=\frac{\mathring{\pi}^{\varphi}(\mathrm{d}\xi)}{\mathring{\pi}(\mathrm{d}\xi\big)}=\frac{1}{N}\sum_{i=1}^{N}r_{u^{(0)},u^{(i)}}(\theta,\vartheta;\zeta_{1}(\theta,\vartheta)).
\]
For $c=2$, we use \eqref{eq:prop-r-circ-phi-inverse-r} to get
\[
\mathring{r}(\theta,\vartheta,\mathfrak{u},k,2)=\left[\frac{1}{N}\sum_{i=0,i\neq k}^{N}r_{u^{(k)},u^{(i)}}(\vartheta,\theta;\zeta_{2}(\theta,\vartheta))\right]^{-1}
\]
\end{proof}

\subsubsection{Refreshing the latent variable in Algorithm \ref{alg: MHAAR-S for SSM}\label{subsec: Refreshing the latent variable in MHAAR-S-SSM}}

As MHAAR-S-SSM in Algorithm \ref{alg: MHAAR-S for SSM} suggests,
we consider refreshing $z$ only when $c=1$ and $\zeta_{1}(\theta,\vartheta)=\theta$.
When $c=1$, one iteration of the modified algorithm can be stated
as follows: Given $x=(\theta,z)$,
\begin{enumerate}
\item Sample $c\sim\text{Unif}(\left\{ 1,2\right\} )$, set $u^{(0)}=z$
and sample $N$ paths $(u^{(1)},\ldots,u^{(N)})$ using a single cSMC
conditioned on $u^{(0)}$.
\item If $c=1$, perform a random swap $u^{(0)}\leftrightarrow u^{(i)}$
with probability $1/N$ for all $i=1,\ldots,N$.
\item Sample $k$ with probability proportional to $r_{u^{(0)},u^{(k)}}(\theta,\vartheta;\theta)$.
\item Propose and accept/reject the move $(\theta,\vartheta,\mathfrak{u},k,1)\rightarrow(\vartheta,\theta,\mathfrak{s}_{0,k}(\mathfrak{u}),k,2)$.
\end{enumerate}
(For practical reasons, the order of steps 3 and 4 can be reversed.)
The step that refreshes $z$ is the second step. By the exchangeability
result for $R_{\theta,\theta}$ in Lemma \ref{lem: exchangeability of cSMC},
step 2 can be shown to target the conditional distribution (with respect
to $\mathring{\pi}$) of $\mathfrak{u}$ given $\theta,\vartheta$,
and $c$, while $\mathbf{k}$ is marginalised out. Therefore, the
fact that this swap move preserves invariance of $\mathring{\pi}$
follows from similar arguments for a collapsed Gibbs move.

Note that step 2 is not a delayed rejection step and it needs to be
implemented before steps 3 and 4. However, this is not an issue computationally,
since the computational complexity of the step is $\mathcal{O}(1)$.
\end{document}